\documentclass[11pt]{article}
\usepackage[utf8]{inputenc}
\usepackage{amsfonts,algorithmic}
\usepackage[para]{footmisc}
\usepackage{xcolor}
\usepackage{qcircuit}
\usepackage{bm}
\usepackage{amsmath}
\usepackage{enumerate}
\usepackage{appendix}
\usepackage{amsthm,mathrsfs}
\usepackage{multirow}
\usepackage{enumitem}
\usepackage{tikz}
\usepackage{multicol}
\usetikzlibrary{matrix,decorations.pathreplacing}
\usepackage{mleftright}
\usepackage{amssymb}
\usepackage[linesnumbered,ruled,vlined]{algorithm2e}
\usepackage{mathabx}
\usepackage{geometry}
\usepackage{comment}
\usepackage{subcaption}
\usepackage{soul}
\usepackage{bbm}
\usepackage[pagebackref]{hyperref}
\usepackage[capitalize,nameinlink,noabbrev]{cleveref}
\hypersetup{
    colorlinks,
    linkcolor={red!100!black},
    citecolor={blue!100!black},
}
\usepackage{authblk}
\usepackage[normalem]{ulem}
\usepackage{euscript}

\newcommand\independent{\protect\mathpalette{\protect\independenT}{\perp}}
\def\independenT#1#2{\mathrel{\rlap{$#1#2$}\mkern2mu{#1#2}}}

\newtheorem{theorem}{Theorem}
\newtheorem{lemma}[theorem]{Lemma}
\newtheorem{remark}[theorem]{Remark}

\newtheorem{definition}[theorem]{Definition}

\newtheorem{fact}[theorem]{Fact}

\newtheorem{corollary}[theorem]{Corollary}

\crefname{lemma}{Lemma}{lemma}
\crefname{definition}{Definition}{definition}
\crefname{corollary}{Corollary}{corollary}

\usepackage{mathtools}

\usepackage{titlesec}
\titleformat{\paragraph}
{\normalfont\normalsize\bfseries}{\theparagraph}{1em}{}
\titlespacing*{\paragraph}
{0pt}{3.25ex plus 1ex minus .2ex}{1.5ex plus .2ex}

\DeclareMathOperator*{\argmax}{arg\,max}

\usepackage{geometry}
\def\be{\begin{eqnarray}}
\def\ee{\end{eqnarray}}

\makeatletter
\newtheorem*{rep@theorem}{\rep@title}
\newcommand{\newreptheorem}[2]{%
\newenvironment{rep#1}[1]{%
 \def\rep@title{\cref{##1}}%
 \begin{rep@theorem}}%
 {\end{rep@theorem}}}
\makeatother

\newreptheorem{lemma}{Lemma}
\newreptheorem{fact}{Fact}
\newreptheorem{definition}{Definition}
\newreptheorem{theorem}{Theorem}
\newreptheorem{corollary}{Corollary}

\DeclareRobustCommand{\DE}[2]{#2}

\begin{document}

\title{Beyond Bell sampling: stabilizer state learning and quantum pseudorandomness lower bounds on qudits}

\author[1]{Jonathan Allcock\thanks{\scriptsize jonallcock@tencent.com}}
\author[2]{Joao F. Doriguello\thanks{\scriptsize doriguello@renyi.hu}}
\author[3]{G\'abor Ivanyos\thanks{\scriptsize gabor.ivanyos@sztaki.mta.hu}}
\author[4,5]{Miklos Santha\thanks{\scriptsize cqtms@nus.edu.sg}}
\affil[1]{Tencent Quantum Laboratory, Hong Kong, China}
\affil[2]{HUN-REN Alfr\'ed R\'enyi Institute of Mathematics, Budapest, Hungary}
\affil[3]{HUN-REN Institute for Computer Science and Control, Budapest, Hungary}
\affil[4]{Centre for Quantum Technologies, National University of Singapore, Singapore}
\affil[5]{CNRS, IRIF, Universit\'e Paris Cit\'e}

\date{\today}
\maketitle

\begin{abstract}        
    Bell sampling is a simple yet powerful measurement primitive that has recently attracted a lot of attention, and has proven to be a valuable tool in studying stabiliser states. Unfortunately, however, it is known that Bell sampling fails when used on qu\emph{d}its of dimension $d>2$. In this paper, we explore and quantify the limitations of Bell sampling on qudits, and propose new quantum algorithms to circumvent the use of Bell sampling in solving two important problems: learning stabiliser states and providing pseudorandomness lower bounds on qudits. More specifically, as our first result, we characterise the output distribution corresponding to Bell sampling on copies of a stabiliser state and show that the output can be uniformly random, and hence reveal no information. As our second result, for $d=p$ prime we devise a quantum algorithm to identify an unknown stabiliser state in $(\mathbb{C}^p)^{\otimes n}$ that uses $O(n)$ copies of the input state and runs in time $O(n^4)$. As our third result, we provide a quantum algorithm that efficiently distinguishes a Haar-random state from a state with non-negligible stabiliser fidelity. As a corollary, any Clifford circuit on qudits of dimension $d$ using $O(\log{n}/\log{d})$ auxiliary non-Clifford single-qudit gates cannot prepare computationally pseudorandom quantum states.
\end{abstract}

\section{Introduction}

Bell sampling, introduced by Montanaro~\cite{montanaro2017learning}, is a measurement primitive which has attracted a lot of attention recently due to its simplicity and far reaching applications~\cite{montanaro2017learning,gross2021schur,hangleiter2023bell,grewal2023efficient,grewal2023improved,grewal2024agnostic}. At a high level, Bell sampling consists of measuring two copies $|\psi\rangle^{\otimes 2}$ of some $n$-qubit state $|\psi\rangle\in(\mathbb{C}^2)^{\otimes n}$ in the Bell basis. The power of Bell sampling becomes evident when combined with the stabiliser formalism~\cite{gottesman1996class,gottesman1997stabilizer} in order to learn properties of stabiliser states. A \emph{stabiliser group} on $n$ qubits is a set of $2^n$ commuting Pauli operators, and a \emph{stabiliser state} is the joint $+1$-eigenvector of a stabiliser group. As shown by Montanaro~\cite{montanaro2017learning}, one Bell sample on two copies of a stabiliser state returns information on its associated stabiliser group, which in turn can be used to infer some of its properties~\cite{gs007,Anshu2024}. This procedure was later explored by Gross, Nezami, and Walter~\cite{gross2021schur}, who introduced a measurement primitive called \emph{Bell difference sampling} which, when performed on copies of a stabiliser state, returns the label of a generator of the associated stabiliser group. Bell difference sampling consists of performing Bell sampling twice and subtracting the two outcomes, thus using four copies $|\psi\rangle^{\otimes 4}$ of a state.

Part of Bell (difference) sampling's appeal comes from the importance of stabiliser states, which find applications in quantum error correction~\cite{shor1995scheme,calderbank1996good,gottesman1996class,gottesman1997stabilizer}, efficient classical simulations of quantum circuits~\cite{aaronson2004improved,bravyi2016trading,bravyi2019simulation}, low-rank recovery~\cite{kueng2016low}, randomized benchmarking~\cite{knill2008randomized,magesan2011scalable,helsen2019multiqubit}, measurement-based quantum computation~\cite{raussendorf2000quantum}, tensor networks for holography codes~\cite{hayden2016holographic,nezami2020multipartite}, and quantum learning algorithms~\cite{huang2020predicting}. As mentioned, Bell difference sampling can be used to learn several properties of a stabiliser state. Montanaro~\cite{montanaro2017learning} implicitly used Bell difference sampling to identify an unknown stabiliser state, which simplified or improved upon previous works on the same task~\cite{aaronson2008stabiliser,rotteler2009quantum,zhao2016fast}. Later, Gross, Nezami, and Walter~\cite{gross2021schur} proposed a stabiliserness testing algorithm based on Bell difference sampling to decide whether a given unknown state is a stabiliser state or is far from one. More recently, Hangleiter and Gullans~\cite{hangleiter2023bell} used Bell sampling to detect and learn circuit errors and to extract useful information, e.g.\ the depth and the number of T-gates of a circuit. Around the same time, Grewal, Iyer, Kretschmer, and Liang~\cite{grewal2022low,grewal2023efficient,grewal2023improved,grewal2024pseudoentanglement,grewal2024agnostic} employed Bell difference sampling in several tasks: obtaining the classical description of the output of a Clifford circuit augmented with a few non-Clifford single-qubit gates (called \emph{doped Clifford circuits}~\cite{leone2023learning,leone2024learning}), testing whether a state is Haar-random or the output of a doped Clifford circuit, approximating an arbitrary quantum state with a stabiliser state, and testing stabiliser fidelity of a given state, which is its maximum overlap with any stabiliser state~\cite{bravyi2019simulation}.

Unfortunately, it is known that Bell (difference) sampling cannot be used for qu\emph{d}its of dimension $d$ higher than $2$~\cite{gross2021schur}, or, more precisely, it reveals little information on the stabiliser group it samples from. At a high level, this stems from the well-known fact that the transpose map $\psi\mapsto \bar{\psi} = \psi^\top$ is not completely positive, or from the fact that the involution $x\mapsto -x$ is non-trivial if $x\in\mathbb{Z}_d$ for $d>2$. This means that Bell sampling falls short in solving the aforementioned problems related to stabiliser states on qudits and, as a result, little is known about them for $d>2$. In this work, we propose new algorithms to circumvent the use of Bell sampling and solve two important tasks on qudits: learning an unknown stabiliser state and deciding whether a given state is Haar-random or has non-negligible stabiliser fidelity.

\subsection{Our results}

\subsubsection{Exploring Bell difference sampling}

Bell sampling~\cite{montanaro2017learning} measures two copies of an $n$-qubit state $|\psi\rangle\in(\mathbb{C}^2)^{\otimes n}$ in the Bell basis $|\mathcal{W}_{\mathbf{x}}\rangle \triangleq (\mathcal{W}_{\mathbf{x}}\otimes \mathbb{I})|\Phi^+\rangle$, where $|\Phi^+\rangle \triangleq 2^{-n/2}\sum_{\mathbf{q}\in\mathbb{F}_2^{n}}|\mathbf{q},\mathbf{q}\rangle$ is a maximally entangled state and the operators $\mathcal{W}_{\mathbf{x}}$, called \emph{Weyl operators}, are defined as
\begin{align*}
    \mathcal{W}_{\mathbf{x}} \triangleq i^{\langle \mathbf{v},\mathbf{w}\rangle} (\mathsf{X}^{v_1}\mathsf{Z}^{w_1})\otimes \cdots \otimes (\mathsf{X}^{v_n}\mathsf{Z}^{w_n}) \qquad \forall\mathbf{x} = (\mathbf{v},\mathbf{w})\in\mathbb{F}_2^{2n},
\end{align*}
where $\langle \mathbf{v},\mathbf{w}\rangle = \sum_{i=1}^n v_i w_i$ is the usual scalar product. One can think of the Weyl operators as the Pauli operators labelled by a $2n$-bit string. The outcome of Bell sampling is thus some bit-string $\mathbf{x}\in\mathbb{F}_2^{2n}$. Bell difference sampling, as proposed by~\cite{gross2021schur}, performs Bell sampling twice and computes the difference of the two outcome strings. Gross, Nezami, and Walter~\cite{gross2021schur} proved that Bell difference sampling on an arbitrary state $|\psi\rangle^{\otimes 4}$ has a useful interpretation. More specifically, they proved that a string $\mathbf{x}\in\mathbb{F}_2^{2n}$ is sampled with probability
\begin{align}\label{eq:bell_difference_sampling}
    q_\psi(\mathbf{x}) \triangleq \sum_{\mathbf{y}\in\mathbb{F}_2^{2n}} p_\psi(\mathbf{y}) p_\psi(\mathbf{x}+\mathbf{y}), \qquad\text{where}~ p_\psi(\mathbf{x}) \triangleq 2^{-n}|\langle\psi|\mathcal{W}_{\mathbf{x}}|\psi\rangle|^2
\end{align}
is called the \emph{characteristic distribution} (see \cref{sec:Weyl_operators} for a proof that $p_\psi$ is indeed a probability distribution). In other words, the \emph{Weyl distribution} $q_\psi$ is the convolution of $p_\psi$ with itself (up to normalisation). If $|\psi\rangle$ is a stabiliser state $|\mathcal{S}\rangle$ of some stabiliser group $\mathcal{S}$ (i.e., a commuting group of $2^n$ Weyl operators, with $|\mathcal{S}\rangle$ the joint $+1$-eigenvector of all such Weyl operators), then $p_{\mathcal{S}}(\mathbf{x}) = 2^{-n}$ for all strings $\mathbf{x}\in\mathbb{F}_2^{2n}$ such that $\mathcal{W}_{\mathbf{x}}\in\mathcal{S}$ and $0$ otherwise. \cref{eq:bell_difference_sampling} thus implies that $q_{\mathcal{S}}(\mathbf{x}) = 2^{-n}$ if $\mathcal{W}_{\mathbf{x}} \in \mathcal{S}$ and $0$ otherwise. Bell difference sampling on copies of a stabiliser state thus returns an element of its stabiliser group, and that is where its power lies. By repeating Bell difference sampling $O(n)$ times it is possible to recover a basis of $\mathcal{S}$ and learn $|\mathcal{S}\rangle$~\cite{montanaro2017learning}, or by repeating it $O(1)$ times it is possible to decide if a given state is either a stabiliser state or far from one~\cite{gross2021schur}. Many properties of $q_\psi$, especially Fourier-related ones, were explored in more depth in~\cite{grewal2022low,grewal2023efficient,grewal2023improved,grewal2024agnostic}.

The situation with qudits, on the other hand, is not so well understood. Apart from general arguments provided in~\cite{gross2021schur} based on the fact that the transpose map $\psi\mapsto \bar{\psi} = \psi^\top$ is not completely positive, or that the involution $x\mapsto -x$ is non-trivial if $x\in\mathbb{Z}_d$ for $d>2$, and therefore that Bell difference sampling is no longer useful, little is known about its actual structure. Therefore, we ask, and give answers to, the following questions in \cref{sec:bell_difference_sampling}.

\begin{center}
    \emph{What is the measurement outcome probability distribution corresponding to performing Bell difference sampling on qudits? Is it possible to extract any information on the stabiliser group of a qudit stabiliser state via Bell difference sampling?}
\end{center}

In what follows, $p>2$ is prime and thus $\mathbb{F}_p^n$ is a vector space over $\mathbb{F}_p$. Let $\omega \triangleq e^{2\pi i/p}$ be the $p$-th root of unity. Our results shall be stated for qu\emph{d}its with dimension $d=p$ prime. 
The Pauli operators on qudits are defined using the shift and clock operators $\mathsf{X}$ and $\mathsf{Z}$, respectively, as
\begin{align*}
    \mathsf{X}|q\rangle = |q+1\rangle, \quad \mathsf{Z}|q\rangle = \omega^q |q\rangle, \quad \forall q\in\mathbb{F}_p.
\end{align*}
The Weyl operators on qudits are then defined to be
\begin{align*}
    \mathcal{W}_{\mathbf{x}} = \mathcal{W}_{\mathbf{v},\mathbf{w}} \triangleq \omega^{2^{-1}\langle \mathbf{v},\mathbf{w}\rangle} (\mathsf{X}^{v_1}\mathsf{Z}^{w_1})\otimes \cdots \otimes (\mathsf{X}^{v_n}\mathsf{Z}^{w_n}) \qquad \forall\mathbf{x} = (\mathbf{v},\mathbf{w})\in\mathbb{F}_p^{2n},
\end{align*}
where $2^{-1} = (p+1)/2$ denotes the multiplicative inverse of $2~\text{mod}~p$. Bell sampling on qudits is defined similarly as it is on qubits: measure two copies of an $n$-qudit state $|\psi\rangle\in(\mathbb{C}^p)^{\otimes n}$ in the generalised Bell basis $|\mathcal{W}_{\mathbf{x}}\rangle \triangleq (\mathcal{W}_{\mathbf{x}}\otimes \mathbb{I})|\Phi^+\rangle$, where $|\Phi^+\rangle \triangleq p^{-n/2}\sum_{\mathbf{q}\in\mathbb{F}_p^{n}}|\mathbf{q},\mathbf{q}\rangle$ is a maximally entangled state. The outcome is a string $\mathbf{x}\in\mathbb{F}_p^{2n}$. Bell difference sampling involves performing Bell sampling twice and taking the difference of the two outcomes. More formally:
\begin{repdefinition}{def:Bell_difference_sampling}
    Bell difference sampling corresponds to performing the projective measurement given by
    \begin{align*}
        \Pi_{\mathbf{x}} = \sum_{\mathbf{y}\in\mathbb{F}_p^{2n}}|\mathcal{W}_{\mathbf{y}}\rangle\langle \mathcal{W}_{\mathbf{y}}|\otimes |\mathcal{W}_{\mathbf{x}+\mathbf{y}}\rangle\langle \mathcal{W}_{\mathbf{x}+\mathbf{y}}|, \quad \mathbf{x}\in\mathbb{F}_p^{2n}.
    \end{align*}
\end{repdefinition}

One of our main results concerning Bell difference sampling is to show that its outcomes are distributed according to an ``involuted'' version of \cref{eq:bell_difference_sampling}. Let $J: \mathbb{F}_p^{2n} \to \mathbb{F}_p^{2n}$ be the involution defined as $J(\mathbf{x}) = J((\mathbf{v},\mathbf{w})) = (-\mathbf{v},\mathbf{w})$ for $\mathbf{x} = (\mathbf{v},\mathbf{w})\in\mathbb{F}_p^{2n}$~\cite{appleby2005symmetric,appleby2009properties,gross2021schur}. 
\begin{reptheorem}{thr:bell_difference}
    Bell difference sampling on $|\psi\rangle^{\otimes 4}$ corresponds to sampling from the distribution
    \begin{align*}
        b_\psi(\mathbf{x}) = \sum_{\mathbf{y}\in\mathbb{F}_p^{2n}} p_\psi(\mathbf{y})p_\psi(J(\mathbf{x} - \mathbf{y})), \qquad\text{where}~ p_\psi(\mathbf{x}) \triangleq p^{-n}|\langle\psi|\mathcal{W}_{\mathbf{x}}|\psi\rangle|^2.
    \end{align*}
\end{reptheorem}
We call the distribution $b_\psi(\mathbf{x})$ the \emph{involuted Weyl distribution} of the pure state $|\psi\rangle$. The above result is a generalisation of \cref{eq:bell_difference_sampling} since the involution is trivial for the case of qubits, i.e., $J(\mathbf{x}) = \mathbf{x}$ for all $\mathbf{x}\in\mathbb{F}_2^{2n}$. More interestingly, though, is to explore $b_\psi$ when $|\psi\rangle = |\mathcal{S}\rangle$ is a stabiliser state. In the following, given a matrix $\mathbf{M}\in\mathbb{F}_p^{m\times n}$, let $\operatorname{col}(\mathbf{M}) \triangleq \{\mathbf{M}\mathbf{v}\in\mathbb{F}_p^m : \mathbf{v}\in\mathbb{F}_p^n \}$ be its column space, and given subspaces $\mathscr{V}, \mathscr{W}$, let $\mathscr{V}+\mathscr{W}\triangleq\{\mathbf{v} + \mathbf{w} : \mathbf{v}\in \mathscr{V}, \mathbf{w}\in \mathscr{W}\}$ and $\mathscr{V}\times \mathscr{W} \triangleq \{(\mathbf{v},\mathbf{w}): \mathbf{v}\in\mathscr{V},\mathbf{w}\in\mathscr{W}\}$.    
\begin{reptheorem}{lem:bell_sampling}
    Let $\mathcal{S}$ be a stabiliser group with generators $\{\omega^{s_i}\mathcal{W}_{\mathbf{x}_i}\}_{i\in[n]}$, and let $|\mathcal{S}\rangle\in(\mathbb{C}^p)^{\otimes n}$ be its stabiliser state. Arrange the labels $\{\mathbf{x}_1,\dots,\mathbf{x}_n\}$ as the columns of the $2n\times n$ matrix $\big[\begin{smallmatrix} \mathbf{V} \\ \mathbf{W} \end{smallmatrix}\big]$, where $\mathbf{V},\mathbf{W}\in\mathbb{F}_p^{n\times n}$. Let $\mathscr{M} \triangleq {\operatorname{col}}\big(\big[\begin{smallmatrix} \mathbf{V} \\ \mathbf{W} \end{smallmatrix}\big]\big)$ and $J(\mathscr{M}) \triangleq \{J(\mathbf{x})\in\mathbb{F}_p^{2n}:\mathbf{x}\in\mathscr{M}\}$. Bell difference sampling on $|\mathcal{S}\rangle^{\otimes 4}$
    returns $\mathbf{x}\in\mathbb{F}_p^{2n}$ with probability
    \begin{align*}
        b_{\mathcal{S}}(\mathbf{x}) &= \begin{cases}
            p^{-2n}|\mathscr{M}\cap J(\mathscr{M})| &\quad\quad\text{if}~\mathbf{x}\in \mathscr{M} + J(\mathscr{M}),\\
            0 &\quad\quad\text{otherwise},
        \end{cases} \\
        &= \begin{cases}
            |\operatorname{col}(\mathbf{V})|^{-1}|\operatorname{col}(\mathbf{W})|^{-1} &\text{if}~\mathbf{x}\in\operatorname{col}(\mathbf{V})\times\operatorname{col}(\mathbf{W}),\\
            0 &\text{otherwise}.
            \end{cases}
    \end{align*}
\end{reptheorem}

This result shows that Bell difference sampling on $|\mathcal{S}\rangle^{\otimes 4}$ allows one to learn $\mathbf{V}$ and $\mathbf{W}$ separately, up to some change of basis, but not $\big[\begin{smallmatrix}\mathbf{V}\\\mathbf{W}\end{smallmatrix}\big]$ and thus $\mathcal{S}$. Equivalently, we only learn the subspace $\mathscr{M} + J(\mathscr{M})$. In a way, the involution function separates and mixes the clock and shift parts of the Weyl operators belonging to a stabiliser group. In the case when both $\mathbf{V}$ and $\mathbf{W}$ are full rank, Bell difference sampling returns a uniformly random $\mathbf{x}\in\mathbb{F}_p^{2n}$. This is in stark contrast to the case of qubits ($p=2$), where Bell difference sampling on $|\mathcal{S}\rangle^{\otimes 4}$ returns $\mathbf{x}\in {\operatorname{col}}\big(\big[\begin{smallmatrix} \mathbf{V} \\ \mathbf{W} \end{smallmatrix}\big]\big)$ with probability $q_{\mathcal{S}}(\mathbf{x}) = 2^{-n}$~\cite{montanaro2017learning,gross2021schur}, which can be seen from the above result since $J(\mathscr{M}) = \mathscr{M}$ for qubits.

\subsubsection{Learning stabiliser states}

As shown by several previous works~\cite{aaronson2008stabiliser,rotteler2009quantum,zhao2016fast,montanaro2017learning}, stabiliser states are among the classes of states that can be efficiently identified, in contrast to arbitrary $n$-qubit states $|\psi\rangle$ which require exponentially many (in $n$) copies of $|\psi\rangle$ to determine, by Holevo's theorem~\cite{holevo1973bounds}. Other examples of efficiently learnable states include matrix product states~\cite{cramer2010efficient,landon2010efficient}, non-interacting fermion states~\cite{aaronson2021efficient}, and low-degree phase states~\cite{arunachalam2022optimal}. 

The simplest and most efficient learning algorithm for stabiliser states on qubits is Montanaro's~\cite{montanaro2017learning} based on Bell (difference) sampling, which identifies an unknown stabiliser state $|\mathcal{S}\rangle\in(\mathbb{C}^2)^{\otimes n}$ in time $O(n^3)$ by using $O(n)$ copies of $|\mathcal{S}\rangle$ and making collective measurements across at most two copies of $|\mathcal{S}\rangle$ at a time. The situation, however, is not as clear in the case of qudits because, as explained above, Bell difference sampling cannot be used to learn stabiliser groups. We propose a new algorithm that circumvents the need for Bell difference sampling in order to learn an unknown stabiliser state. As far as we are aware, there have not previously been any direct results in this direction. Our algorithm is based on the observation that any stabiliser state can be written in a very specific form.
\begin{replemma}{lem:normalisation}
    Let $\mathcal{S}$ be a stabiliser group with generators $\{\omega^{s_i}\mathcal{W}_{\mathbf{x}_i}\}_{i\in[n]}$.
    Let $\mathbf{s} = (s_1,\dots,s_n)\in\mathbb{F}_p^n$ and arrange the labels $\{\mathbf{x}_1,\dots,\mathbf{x}_n\}$ as the columns of the $2n\times n$ matrix $\big[\begin{smallmatrix} \mathbf{V} \\ \mathbf{W} \end{smallmatrix}\big]$, where $\mathbf{V},\mathbf{W}\in\mathbb{F}_p^{n\times n}$. Then there exists $\mathbf{u}\in\mathbb{F}_p^n$ such that $\mathbf{V}^\top \mathbf{u} + \mathbf{s} \in \operatorname{col}(\mathbf{W}^\top)$, and the unique stabiliser state of $\mathcal{S}$ is
    \begin{align*}
        |\mathcal{S}\rangle = \frac{\sqrt{|\operatorname{col}(\mathbf{W})|}}{p^n}\sum_{\mathbf{q}\in\mathbb{F}_p^n}\omega^{\mathbf{s}^\top \mathbf{q} + \mathbf{u}^\top \mathbf{V}\mathbf{q} + 2^{-1}\mathbf{q}^\top \mathbf{V}^\top \mathbf{W}\mathbf{q}}|\mathbf{u}+\mathbf{W}\mathbf{q}\rangle
    \end{align*}
    for any $\mathbf{u}\in\mathbb{F}_p^n$ such that $\mathbf{V}^\top \mathbf{u} + \mathbf{s} \in \operatorname{col}(\mathbf{W}^\top)$.
\end{replemma}
We note that a similar, but somewhat more convoluted, form for stabiliser states was given in~\cite[Theorem~1]{hostens2005stabilizer}. Our proof of \cref{lem:normalisation} is quite simple in comparison. It tells us that any stabiliser state is a superposition of basis states from an affine subspace of $\mathbb{F}_p^n$ with relative phases given by a quadratic polynomial over $\mathbb{F}_p$. Based on this form, we leverage ideas from hidden subgroup/polynomial problems~\cite{bacon2005optimal,ivanyos2007efficient,krovi2008efficient,decker2013polynomial,IVANYOS201773} to develop a quantum algorithm that learns a stabiliser state $|\mathcal{S}\rangle\in(\mathbb{C}^p)^{\otimes n}$ given $O(n)$ copies of $|\mathcal{S}\rangle$.
\begin{reptheorem}{thr:algorithm_method_2}
    Let $\mathcal{S}$ be an unknown stabiliser group generated by $\{\omega^{s_i}\mathcal{W}_{\mathbf{v}_i,\mathbf{w}_i}\}_{i\in[n]}$ with stabiliser state $|\mathcal{S}\rangle\in(\mathbb{C}^{p})^{\otimes n}$. Let $\mathbf{W}\in\mathbb{F}_p^{n\times n}$ be the matrix with column vectors $\mathbf{w}_1,\dots,\mathbf{w}_n$. There is a quantum algorithm that identifies $\mathcal{S}$ by using $O(n)$ copies of $|\mathcal{S}\rangle$. It makes measurements only on the computational basis, runs in time $O(n^3\operatorname{rank}(\mathbf{W}))$, and fails with probability $\leq 2p^{-n}$.
\end{reptheorem}

The idea of our algorithm is to learn both $\mathbf{V}$ and $\mathbf{W}$ which, together, make up a set $\big[\begin{smallmatrix} \mathbf{V} \\ \mathbf{W} \end{smallmatrix}\big]$ of generators for $\mathcal{S}$. The first step is to simply measure several copies of $|\mathcal{S}\rangle$ in the computational basis to learn $\mathbf{W}$. In the second step, we learn the corresponding pair $\mathbf{V}$ of the previously obtained matrix $\mathbf{W}$. To do so, we employ techniques from hidden polynomial problems~\cite{decker2013polynomial,IVANYOS201773} to learn a quadratic function, in our case $\mathbf{s}^\top \mathbf{q} + \mathbf{u}^\top \mathbf{V}\mathbf{q} + 2^{-1}\mathbf{q}^\top \mathbf{V}^\top \mathbf{W}\mathbf{q}$, and, therefore, extract the matrix $\mathbf{V}$. In particular, we use a trick of introducing parameters $\delta_1,\delta_2,\delta_3\in\mathbb{F}_p$, not all $0$, such that $\delta_1^2 + \delta_2^2 + \delta_3^2 = 0$, into the phase one wants to learn. This kills any quadratic term and we are only left with a linear term that can be determined via a quantum Fourier transform. The three parameters $\delta_1,\delta_2,\delta_3\in\mathbb{F}_p$ are inputs to the algorithm and can be found in time $O(\operatorname{poly}\log{p})$ by using the Las Vegas method of~\cite{ivanyos2007efficient} or the deterministic method of~\cite{van2005deterministic}. The final step of the algorithm is to measure copies of $|\mathcal{S}\rangle$ in the eigenbasis of $\mathcal{W}_{\mathbf{v}_i,\mathbf{w}_i}$ to obtain the corresponding phases $\omega^{s_i}$. Our quantum algorithm only measures in the computational basis and, in the worst case, runs in time $O(n^4)$ if $\operatorname{rank}(\mathbf{W}) = \Omega(n)$. However, it makes extensive use of vector spaces, so it is not clear whether it can be extended to qudits of dimension $d$ not prime. 

We remark that, if one also has access to copies of the complex conjugate $|\mathcal{S}^\ast\rangle$ of the stabiliser state $|\mathcal{S}\rangle$, then it is possible to use Bell sampling to learn $|\mathcal{S}\rangle$. This is because Bell sampling on $|\mathcal{S}\rangle|\mathcal{S}^\ast\rangle$ returns $\mathbf{x}\in\mathbb{F}_p^{2n}$ with probability $p_{\mathcal{S}}(\mathbf{x}) = p^{-n}|\langle \mathcal{S}|\mathcal{W}_{\mathbf{x}}|\mathcal{S}\rangle|^2$ which equals $p^{-n}$ if $\mathcal{W}_{\mathbf{x}}\in\mathcal{S}$ and $0$ otherwise. 
\begin{reptheorem}{thr:algorithm_bell_sampling}
    Let $\mathcal{S}$ be an unknown stabiliser group with stabiliser state $|\mathcal{S}\rangle\in(\mathbb{C}^{p})^{\otimes n}$. There is a quantum algorithm that identifies $\mathcal{S}$ using $O(n)$ copies of $|\mathcal{S}\rangle$ and $|\mathcal{S}^\ast\rangle$. Its runtime is $O(n^3)$ and failure probability is at most $p^{-n}$.
\end{reptheorem}
\cref{thr:algorithm_bell_sampling} is valid for any dimension, not necessarily prime. We note that the phenomenon of having considerably more power from the access to the complex conjugate of a quantum state has been observed by King, Wan, and McClean~\cite{king2024exponential} in the context of shadow tomography.

The time complexities of \cref{thr:algorithm_method_2} and \cref{thr:algorithm_bell_sampling} (and throughout the paper) assume a computational model wherein classical operations over the field $\mathbb{F}_p$ (or ring $\mathbb{Z}_d$ for $d$ non-prime) take $O(1)$ time, and measurements in the computational basis and single and two-qudit gates from a universal gate set~\cite{Muthukrishnan2000multivalued,brylinski2002universal,Brennen2005criteria,brennen2006efficient} require $O(1)$ time to be performed. The latter assumption is supported by the fact that there exists quantum hardware that naturally encodes information in qudits~\cite{meth2023simulating}.

\subsubsection{Pseudorandomness lower bounds}

Pseudorandom states were introduced by Ji, Liu, and Song~\cite{ji2018pseudorandom} and have attracted a lot of attention due to their applicability in quantum cryptography and complexity theory~\cite{ji2018pseudorandom,kretschmer2021quantum,ananth2022cryptography,morimae2022quantum,hhan2023hardness,kretschmer2023quantum}. At a high level, a set of quantum states is pseudorandom if they are efficiently preparable by some quantum algorithm and mimic the Haar measure over $n$-qudit states. In the following, a function $\epsilon(\kappa)$ is negligible if $\epsilon(\kappa) = o(\kappa^{-c})$ $\forall c>0$.
\begin{definition}[{\cite[Definition~2]{ji2018pseudorandom}}]\label{def:pseudorandomness}
    Given a security parameter $\kappa$, a keyed family of $n$-qudit quantum states $\{|\phi_k\rangle\in(\mathbb{C}^d)^{\otimes n}\}_{k\in\{0,1\}^\kappa}$ is pseudorandom if (i) there is a polynomial-time quantum algorithm that generates $|\phi_k\rangle$ on input $k$, and (ii) for any $\operatorname{poly}(\kappa)$-time quantum adversary $\mathcal{A}$,
    \begin{align*}
        \left|\operatorname*{\mathbb{P}}_{k\sim\{0,1\}^\kappa}[\mathcal{A}(|\phi_k\rangle^{\otimes \operatorname{poly}(\kappa)}) = 1] - \operatorname*{\mathbb{P}}_{|\psi\rangle\sim\mu_{\rm Haar}}[\mathcal{A}(|\psi\rangle^{\otimes \operatorname{poly}(\kappa)}) = 1]\right| = \operatorname{negl}(\kappa),
    \end{align*}
    where $\mu_{\rm Haar}$ is the $n$-qudit Haar measure and $\operatorname{negl}(\kappa)$ is an arbitrary negligible function of $\kappa$.
\end{definition}
A few works have explored the resources required for constructing pseudorandom states, e.g.\ Ref.~\cite{aaronson2022quantum} explored possible constructions of pseudorandom states using limited entanglement. Grewal, Iyer, Kretschmer, and Liang~\cite{grewal2022low,grewal2023improved}, on the other hand, analysed pseudorandomness from the perspective of stabiliser complexity. More specifically, they proposed a quantum algorithm to test whether a quantum state is Haar-random or has non-negligible \emph{stabiliser fidelity}~\cite{bravyi2019simulation}. As a consequence, they proved that any Clifford+T circuit that uses $O(\log{n})$ T-gates cannot generate a set of $n$-qubit pseudorandom quantum states~\cite{grewal2022low}. This bound was later improved to $n/2$ non-Clifford single-qubit gates~\cite{grewal2023improved} by employing Bell difference sampling to distinguish Haar-random states from doped Clifford circuit outputs. 

In this work, we propose a quantum algorithm to distinguish Haar-random states from states with non-negligible stabiliser fidelity on qudits and, as a consequence, we prove similar pseudorandomness lower bounds for doped Clifford circuits on qudits. Here, the stabiliser fidelity of $|\psi\rangle$ is the maximum overlap between $|\psi\rangle$ and any stabiliser state. Moreover, a Clifford circuit on $n$ qudits of dimension $p$ is any $p^n\times p^n$ unitary operator that normalises the generalised Pauli group $\mathscr{P}^n_{p} = \{\omega^s \mathcal{W}_{\mathbf{x}}: s\in\mathbb{F}_p, \mathbf{x}\in\mathbb{F}_p^{2n}\}$, i.e., any element of $\mathscr{C}^n_{p} = \{\mathcal{U}\in\mathbb{C}^{p^n\times p^n}: \mathcal{U}^\dagger\mathcal{U} = \mathbb{I},~ \mathcal{U}\mathscr{P}^n_{p}\mathcal{U}^\dagger = \mathscr{P}^n_{p}\}$. 
\begin{reptheorem}{thr:algorithm_pseudorandomness}
    Let $\delta\in(0,1)$ and $|\psi\rangle\in(\mathbb{C}^p)^{\otimes n}$ be a state promised to be either Haar-random or to have stabiliser fidelity at least $k^{-1}$. There is a quantum algorithm that distinguishes the two cases with probability at least $1-\delta$, uses $ O(k^{8}\log(1/\delta))$ copies of $|\psi\rangle$, and has runtime $O(nk^{8}\log(1/\delta))$. 
\end{reptheorem}

Since Bell difference sampling has limited utility for qudits, we employ a completely different technique from the one used by~\cite{grewal2022low,grewal2023improved}. We instead utilise the binary POVM proposed by Gross, Nezami, and Walter~\cite{gross2021schur} -- originally used to test stabiliserness -- and show that it can be used to successfully distinguish Haar-random states from those with non-negligible stabiliser fidelity. Our analysis uses some tools from~\cite{grewal2022low,grewal2023improved} to bound the stabiliser fidelity of Haar-random states. As a byproduct, our quantum algorithm can also be used to distinguish between Haar-random states and those output by doped Clifford circuits, since the latter have non-negligible stabiliser fidelity. If the doped Clifford circuit uses at most $t$ non-Clifford single-qudit gates, then the sample and time complexities of our algorithm are $O(p^{4t}\log(1/\delta))$ and $O(np^{4t}\log(1/\delta))$, respectively. As a corollary, we have the following.
\begin{repcorollary}{cor:pseudorandomness_lower bound}
    Any doped Clifford circuit in $(\mathbb{C}^p)^{\otimes}$ that uses $O(\log{n}/\log{p})$ non-Clifford single-qudit gates cannot produce an ensemble of pseudorandom quantum states.
\end{repcorollary}
We point out the contrast between the above result and Ref.~\cite{hinsche2023single} which proved that learning the output distribution of a Clifford circuit over qubits with a single T-gate is as hard as the problem of learning parities with noise. The discrepancy with our results stems from the difference in access models: while we assume access to multiple copies of a state, \cite{hinsche2023single} considers only algorithms that deal with computational-basis measurements of the state.

Even though we prove \cref{thr:algorithm_pseudorandomness} and \cref{cor:pseudorandomness_lower bound} for qudits of dimension $p$ prime, they can be generalised to arbitrary dimension $d$ by using the POVM of~\cite{gross2021schur} for arbitrary $d$. If $r$ is any coprime with $d$, then this leads to an algorithm that uses $O(rk^{4r}\log(1/\delta))$ copies of $|\psi\rangle$ and runs in time $O(nrk^{4r}\log(1/\delta))$. The pseudorandomness bound is still $O(\log{n}/\log{d})$.

\subsection{Open problems}

A few open problems remain from our work. The first would be to extend our learning algorithm to qudits of arbitrary dimension $d$ (not necessarily prime). Ideally, a new technique akin to Bell difference sampling for qudits could resolve this problem, or a more clever version of our algorithm which does not rely on vector spaces so much. Another open problem is to improve our pseudorandomness lower bound of $O(\log{n}/\log{p})$. We believe it can be exponentially strengthened to $O(n)$ by designing a better algorithm. Finally, there are a plethora of results on qubits which were proven using Bell (difference) sampling that we have not considered, e.g.\ finding a succinct description of a stabiliser state close in stabiliser fidelity to a given state, learning the output of a doped Clifford circuit, giving pseudoentangled state lower bounds~\cite{grewal2023efficient,grewal2023improved,hangleiter2023bell,grewal2024pseudoentanglement,grewal2024agnostic}. It would be interesting to extend these results to qudits.

\section{Preliminaries}

Given $n\in\mathbb{N} \triangleq \{1,2,\dots\}$, let $[n] \triangleq \{1,\dots, n\}$. Throughout this paper, unless stated otherwise, let $p\in\mathbb{N}$ be prime and $p>2$. Let $\omega \triangleq e^{2\pi i/p}$ be the $p$-th root of unity and let $\tau \triangleq \omega^{2^{-1}}$, where $2^{-1}=(p+1)/2$ denotes the multiplicative inverse of $2\operatorname{mod}p$, so that $\tau^2 = \omega$. For any state $|\psi\rangle = \sum_{\mathbf{w}\in\mathbb{F}_p^n}a_{\mathbf{w}}|\mathbf{w}\rangle$ in $(\mathbb{C}^{p})^{\otimes n}$, let $|\psi^\ast\rangle = \sum_{\mathbf{w}\in\mathbb{F}_p^n}\overline{a_{\mathbf{w}}}|\mathbf{w}\rangle$ denote its complex conjugate with respect to the computational basis. We shall often write $\psi$ to denote $|\psi\rangle\langle\psi|$. By $\mathbb{I}$ we mean the identity matrix or operator, and $\mathbf{1}[\cdot]$ is the indicator function such that $\mathbf{1}[\text{statement}]$ equals $1$ if the statement is true and $0$ if it is false. We assume that the field operations of $\mathbb{F}_p$ (addition and multiplication) require constant time.

\subsection{Scalar product vector spaces and symplectic vector spaces}

Let $n,m\in\mathbb{N}$. We shall work with two types of vectors space over $\mathbb{F}_p$: (i) the vector space $\mathbb{F}_p^n$ equipped with the \emph{scalar product} $\langle \cdot, \cdot \rangle: \mathbb{F}_p^n \times \mathbb{F}_p^n \to \mathbb{F}_p$ defined as $\langle \mathbf{v},\mathbf{w}\rangle = \sum_{i=1}^n v_iw_i$, which we refer to as a scalar product vector space; (ii) the vector space $\mathbb{F}_p^{2n}$ equipped with the \emph{symplectic product} $[\cdot,\cdot] : \mathbb{F}_p^{2n} \times \mathbb{F}_p^{2n} \to \mathbb{F}_p$ defined as $[\mathbf{x},\mathbf{y}] = \sum_{i=1}^n (x_i y_{n+i} - x_{n+i}y_i)$, which we refer to as a symplectic vector space. It should be clear from context whether we are considering a scalar product or a symplectic vector space. Note that both scalar product $\langle \cdot, \cdot \rangle$ and symplectic product $[\cdot,\cdot]$ are \emph{non-degenerate}, meaning that the zero vector $\mathbf{0}$ is the only vector orthogonal to any vector in the vector space.

Given subspaces $\mathscr{V}, \mathscr{W}$ of the scalar product vector space $\mathbb{F}_p^n$ and function $f:\mathbb{F}_p^n \to \mathbb{F}_p^m$, let
\begin{align*}
    \mathscr{V}+\mathscr{W}\triangleq\{\mathbf{v} + \mathbf{w} : \mathbf{v}\in \mathscr{V}, \mathbf{w}\in \mathscr{W}\},  ~f(\mathscr{V}) \triangleq \{f(\mathbf{v}) : \mathbf{v}\in\mathscr{V}\},  ~\mathscr{V}\times \mathscr{W} \triangleq \{(\mathbf{v},\mathbf{w}): \mathbf{v}\in\mathscr{V},\mathbf{w}\in\mathscr{W}\}.
\end{align*}
As subcases, $\mathbf{u}+\mathscr{V}\triangleq \{\mathbf{u}+\mathbf{v}:\mathbf{v}\in \mathscr{V}\}$ and $\mathbf{M}\mathscr{V}\triangleq \{\mathbf{M}\mathbf{v} : \mathbf{v}\in \mathscr{V}\}$ for $\mathbf{u}\in\mathbb{F}_p^n$ and $\mathbf{M}\in\mathbb{F}_p^{m\times n}$. Similarly for the symplectic vector space $\mathbb{F}_p^{2n}$, given subspaces $\mathscr{X}, \mathscr{Y}\subseteq\mathbb{F}_p^{2n}$ and a function $f:\mathbb{F}_p^{2n} \to \mathbb{F}_p^{2m}$, define
\begin{align*}
    \mathscr{X}+\mathscr{Y}\triangleq\{\mathbf{x} + \mathbf{y} : \mathbf{x}\in \mathscr{X}, \mathbf{y}\in \mathscr{Y}\},  ~f(\mathscr{X}) \triangleq \{f(\mathbf{x}) : \mathbf{x}\in\mathscr{X}\},  ~\mathscr{X}\times \mathscr{Y} \triangleq \{(\mathbf{x},\mathbf{y}): \mathbf{x}\in\mathscr{X},\mathbf{y}\in\mathscr{Y}\}.
\end{align*}
As subcases, $\mathbf{z}+\mathscr{X}\triangleq \{\mathbf{z}+\mathbf{x}:\mathbf{x}\in \mathscr{X}\}$ and $\mathbf{N}\mathscr{X}\triangleq \{\mathbf{N}\mathbf{x} : \mathbf{x}\in \mathscr{X}\}$ for $\mathbf{z}\in\mathbb{F}_p^{2n}$ and $\mathbf{N}\in\mathbb{F}_p^{2m\times 2n}$.

Let $\mathbf{e}_i\in\mathbb{F}^n_p$ be the vector whose $i$-th component is $1$ and the remaining components are $0$. Let $\operatorname{vec}$ be the vectorisation operator defined by $\operatorname{vec}(|\mathbf{v}\rangle\langle \mathbf{w}|) = |\mathbf{v}\rangle|\mathbf{w}\rangle$ for computational basis states $\mathbf{v},\mathbf{w}\in\mathbb{F}_p^n$. Note that $\operatorname{vec}$ preserves inner products, i.e., $\langle \operatorname{vec}(\mathbf{A})|\operatorname{vec}(\mathbf{B})\rangle = \operatorname{Tr}[\mathbf{A}^\dagger \mathbf{B}]$ for all $\mathbf{A},\mathbf{B}\in\mathbb{F}_p^{m\times n}$. Given a matrix $\mathbf{M}\in\mathbb{F}_p^{m\times n}$, let 
\begin{align*}
    \operatorname{col}(\mathbf{M}) &\triangleq \{\mathbf{M}\mathbf{v}\in\mathbb{F}_p^m : \mathbf{v}\in\mathbb{F}_p^n \} &&\text{be its column space}, \\
    \operatorname{row}(\mathbf{M}) &\triangleq \operatorname{col}(\mathbf{M}^\top) &&\text{be its row space,}\\
    \operatorname{null}(\mathbf{M}) &\triangleq \{\mathbf{v}\in\mathbb{F}_p^n: \mathbf{M}\mathbf{v} = \mathbf{0}\} &&\text{be its null space}.
\end{align*}

Given a subspace $\mathscr{V}\subseteq\mathbb{F}_p^{n}$ of the scalar product vector space $\mathbb{F}_p^n$, its \emph{orthogonal complement} is $\mathscr{V}^\perp \triangleq \{\mathbf{w}\in\mathbb{F}_p^n: \langle\mathbf{v}, \mathbf{w}\rangle = 0, \forall \mathbf{v}\in\mathscr{V}\}$. Given a subspace $\mathscr{X}\subseteq\mathbb{F}_p^{2n}$ of the symplectic vector space $\mathbb{F}_p^{2n}$, its \emph{symplectic complement} is $\mathscr{X}^{\independent} \triangleq \{\mathbf{y}\in\mathbb{F}_p^{2n}: [\mathbf{x}, \mathbf{y}] = 0, \forall \mathbf{x}\in\mathscr{X}\}$. We shall use the following facts about orthogonal and symplectic complements throughout.
\begin{fact}
    Let $\mathscr{V},\mathscr{W}\subseteq \mathbb{F}_p^n$ and $\mathscr{X},\mathscr{Y}\subseteq \mathbb{F}_p^{2n}$ be subspaces of the scalar product and symplectic vector spaces $\mathbb{F}_p^n$ and $\mathbb{F}_p^{2n}$, respectively. Then
    \begin{multicols}{2}
    \begin{enumerate}[label=\alph*.]
        \item $\mathscr{V}^\perp$ is a subspace;
        \item $(\mathscr{V}^\perp)^\perp = \mathscr{V}$;
        \item $\operatorname{dim}(\mathscr{V}) + \operatorname{dim}(\mathscr{V}^\perp) = n$;
        \item $|\mathscr{V}||\mathscr{V}^\perp| = p^n$;
        \item $\mathscr{V}\subseteq \mathscr{W} \iff \mathscr{W}^\perp \subseteq \mathscr{V}^\perp$;
        \item $(\mathscr{V} + \mathscr{W})^\perp = \mathscr{V}^\perp \cap \mathscr{W}^\perp$;
        \item $\mathscr{X}^{\independent}$ is a subspace;
        \item $(\mathscr{X}^{\independent})^{\independent} = \mathscr{X}$;
        \item $\operatorname{dim}(\mathscr{X}) + \operatorname{dim}(\mathscr{X}^{\independent}) = 2n$;
        \item $|\mathscr{X}||\mathscr{X}^{\independent}| = p^{2n}$;
        \item $\mathscr{X}\subseteq \mathscr{Y} \iff \mathscr{Y}^{\independent} \subseteq \mathscr{X}^{\independent}$;
        \item $(\mathscr{X} + \mathscr{Y})^{\independent} = \mathscr{X}^{\independent} \cap \mathscr{Y}^{\independent}$.
    \end{enumerate}
    \end{multicols}
\end{fact}

The concept of isotropic and Lagrangian subspaces will be fundamental to our analysis.
\begin{definition}
    A subspace $\mathscr{X}\subset \mathbb{F}_p^{2n}$ of a symplectic vector space $\mathbb{F}_p^{2n}$ is \emph{isotropic} if $[\mathbf{x},\mathbf{y}] = 0$ for all $\mathbf{x},\mathbf{y}\in\mathscr{X}$, and is \emph{Lagrangian} if $\mathscr{X}^{\independent} = \mathscr{X}$. 
\end{definition}
\begin{fact}
    Every Lagrangian subspace of a $2n$-dimension vector space is isotropic and has dimension $n$. Every isotropic subspace can be extended to a Lagrangian one.
\end{fact}

\begin{definition}[Involution]
    The involution $J:\mathbb{F}_p^{2n}\to\mathbb{F}_p^{2n}$ on the symplectic vector space $\mathbb{F}_p^{2n}$ is defined as $J(\mathbf{x}) = J((\mathbf{v},\mathbf{w})) = (-\mathbf{v},\mathbf{w})$ where $\mathbf{x} = (\mathbf{v},\mathbf{w})$ and $\mathbf{v},\mathbf{w}\in\mathbb{F}_p^n$. 
\end{definition}
\begin{lemma}
    Let $J:\mathbb{F}_p^{2n}\to\mathbb{F}_p^{2n}$ be the involution. For any $\mathbf{x},\mathbf{y}\in\mathbb{F}_p^{2n}$ and subspace $\mathscr{X}\subseteq\mathbb{F}_p^{2n}$, it holds that $J(J(\mathbf{x})) = \mathbf{x}$, $[J(\mathbf{x}),J(\mathbf{y})] = -[\mathbf{x},\mathbf{y}]$, and $J(\mathscr{X}^{\independent}) = J(\mathscr{X})^{\independent}$.
\end{lemma}

The following fact will be useful. A proof is provided in \cref{app:proofs}.

\begin{lemma}\label{lem:sum}
    Let $\mathscr{V}\subseteq\mathbb{F}_p^n$ and $\mathscr{X}\subseteq\mathbb{F}_p^{2n}$ be subspaces and $\mathbf{w} \in\mathbb{F}_p^n$ and $\mathbf{y} \in\mathbb{F}_p^{2n}$. Then
    \begin{align*}
        \sum_{\mathbf{v}\in\mathscr{V}} \omega^{\langle\mathbf{v}, \mathbf{w}\rangle} = \begin{cases}
            |\mathscr{V}| &\text{if}~\mathbf{w} \in \mathscr{V}^\perp,\\
            0 &\text{if}~ \mathbf{w}\notin \mathscr{V}^\perp,
        \end{cases}\qquad\qquad
        \sum_{\mathbf{x}\in\mathscr{X}} \omega^{[\mathbf{x}, \mathbf{y}]} = \begin{cases}
            |\mathscr{X}| &\text{if}~\mathbf{y} \in \mathscr{X}^{\independent},\\
            0 &\text{if}~ \mathbf{y}\notin \mathscr{X}^{\independent}.
        \end{cases}
    \end{align*}
\end{lemma}

\subsection{Symplectic Fourier analysis}

Boolean Fourier analysis is a well-established field wherein a Boolean function is studied through its Fourier expansion using the characters $\omega^{\langle \cdot, \cdot \rangle}$ defined with respect to the scalar product. See e.g.~\cite{de2008brief,o2014analysis} for an introduction. A character of $\mathbb{F}_p^{2n}$ is a group homomorphism $\chi:\mathbb{F}_p^{2n} \to \mathbb{C}$ such that $\chi(\mathbf{x} + \mathbf{y}) = \chi(\mathbf{x})\chi(\mathbf{y})$ for all $\mathbf{x},\mathbf{y}\in\mathbb{F}_p^{2n}$. In this paper, we shall work with \emph{symplectic} Fourier analysis, which is similar to the usual Boolean Fourier analysis but the Fourier characters are instead defined with respect to the symplectic product as $\omega^{[\cdot, \cdot]}$. It is a standard result that every character of the symplectic vector space $\mathbb{F}_p^{2n}$ can be written as $\omega^{[\mathbf{x},\cdot]}$ for some $\mathbf{x}\in\mathbb{F}_p^{2n}$. Refs.~\cite{gross2021schur,grewal2023improved} studied some aspects of symplectic Fourier analysis, especially over the domain $\mathbb{F}_2^{2n}$ ($p=2$).
\begin{definition}
    Let $f:\mathbb{F}_p^{2n}\to\mathbb{C}$. The symplectic Fourier transform of $f$, $\widehat{f}:\mathbb{F}_p^{2n}\to\mathbb{C}$, is
    \begin{align*}
        \widehat{f}(\mathbf{y}) \triangleq \frac{1}{p^{2n}}\sum_{\mathbf{x}\in\mathbb{F}_p^{2n}}\omega^{[\mathbf{y},\mathbf{x}]}f(\mathbf{x}).
    \end{align*}
\end{definition}
It is possible to expand any function $f:\mathbb{F}_p^{2n}\to\mathbb{C}$ using its symplectic Fourier transform as
\begin{align*}
    f(\mathbf{x}) = \sum_{\mathbf{y}\in\mathbb{F}_p^{2n}}\omega^{[\mathbf{x},\mathbf{y}]}\widehat{f}(\mathbf{y}). 
\end{align*}
Indeed, by \cref{lem:sum}, 
\begin{align*}
    \sum_{\mathbf{y}\in\mathbb{F}_p^{2n}}\omega^{[\mathbf{x},\mathbf{y}]}\widehat{f}(\mathbf{y}) = \frac{1}{p^{2n}}\sum_{\mathbf{y},\mathbf{z}\in\mathbb{F}_p^{2n}}\omega^{[\mathbf{x},\mathbf{y}]+[\mathbf{y},\mathbf{z}]}f(\mathbf{z}) = \frac{1}{p^{2n}}\sum_{\mathbf{y},\mathbf{z}\in\mathbb{F}_p^{2n}}\omega^{[\mathbf{x}-\mathbf{z},\mathbf{y}]}f(\mathbf{z}) = \sum_{\mathbf{z}\in\mathbb{F}_p^{2n}} f(\mathbf{z}) \mathbf{1}[\mathbf{z} = \mathbf{x}] = f(\mathbf{x}).
\end{align*}

The transformation $f\mapsto \widehat{f}$ is unitary, meaning that (a generalised) Parseval's identity holds.
\begin{lemma}[Parseval's identity]\label{lem:parseval}
    Given $f,g:\mathbb{F}_p^{2n}\to\mathbb{C}$ and $\mathbf{t}\in\mathbb{F}_p^{2n}$, then
    \begin{align*}
        \frac{1}{p^{2n}}\sum_{\mathbf{x}\in\mathbb{F}_p^{2n}} \omega^{[\mathbf{t},\mathbf{x}]} f(\mathbf{x})g(\mathbf{x}) &= \sum_{\mathbf{y}\in\mathbb{F}_p^{2n}}\widehat{f}(\mathbf{y})\widehat{g}(\mathbf{t}-\mathbf{y}).
    \end{align*}
\end{lemma}
\begin{proof}
    $\begin{aligned}[t]
    \sum_{\mathbf{x}\in\mathbb{F}_p^{2n}} \omega^{[\mathbf{t},\mathbf{x}]}f(\mathbf{x}) g(\mathbf{x}) &= \sum_{\mathbf{x},\mathbf{y},\mathbf{z}\in\mathbb{F}_p^{2n}}\omega^{[\mathbf{x},\mathbf{y}+\mathbf{z}-\mathbf{t}]}\widehat{f}(\mathbf{y})\widehat{g}(\mathbf{z}) \\
    &= p^{2n}\sum_{\mathbf{y},\mathbf{z}\in\mathbb{F}_p^{2n}}\widehat{f}(\mathbf{y})\widehat{g}(\mathbf{z})\cdot\mathbf{1}[\mathbf{z} = \mathbf{t}-\mathbf{y}] = p^{2n}\sum_{\mathbf{y}\in\mathbb{F}_p^{2n}}\widehat{f}(\mathbf{y})\widehat{g}(\mathbf{t}-\mathbf{y}).
    \hspace{0.52cm}\qedhere
    \end{aligned}$
\end{proof}
The usual Parseval's identity follows from the above lemma by taking $\mathbf{t} = \mathbf{0}$ and $g(\mathbf{x}) = \overline{f(\mathbf{x})}$. We also define the convolution operation.
\begin{definition}[Convolution]
    Let $f,g:\mathbb{F}_p^{2n}\to\mathbb{C}$. The convolution between $f$ and $g$ is the function $f\ast g:\mathbb{F}_p^{2n}\to\mathbb{C}$ defined by
    \begin{align*}
        (f\ast g)(\mathbf{x}) \triangleq \frac{1}{p^{2n}}\sum_{\mathbf{y}\in\mathbb{F}_p^{2n}}f(\mathbf{y})g(\mathbf{x} - \mathbf{y}) = \frac{1}{p^{2n}}\sum_{\mathbf{y}\in\mathbb{F}_p^{2n}}f(\mathbf{x}-\mathbf{y})g(\mathbf{y}).
    \end{align*}
\end{definition}
\begin{lemma}\label{lem:convolution_Fourier}
    Let $f,g:\mathbb{F}_p^{2n}\to\mathbb{C}$. Then, for all $\mathbf{x}\in\mathbb{F}_p^{2n}$, $\widehat{f\ast g}(\mathbf{x}) = \widehat{f}(\mathbf{x})\widehat{g}(\mathbf{x})$.
\end{lemma}
\begin{proof}
    $\begin{aligned}[t]
        \widehat{f\ast g}(\mathbf{x}) &= \frac{1}{p^{2n}} \sum_{\mathbf{y}\in\mathbb{F}_p^{2n}}\omega^{[\mathbf{x},\mathbf{y}]}(f\ast g)(\mathbf{y}) = \frac{1}{p^{4n}}\sum_{\mathbf{y},\mathbf{z}\in\mathbb{F}_p^{2n}}\omega^{[\mathbf{x},\mathbf{y}]} f(\mathbf{z})g(\mathbf{y} - \mathbf{z}) \\
        &= \frac{1}{p^{4n}}\sum_{\mathbf{z}\in\mathbb{F}_p^{2n}} \omega^{[\mathbf{x},\mathbf{z}]}f(\mathbf{z})\sum_{\mathbf{y}\in\mathbb{F}_p^{2n}}\omega^{[\mathbf{x},\mathbf{y}-\mathbf{z}]}g(\mathbf{y}-\mathbf{z}) = \widehat{f}(\mathbf{x})\widehat{g}(\mathbf{x}). \hspace{3.3cm}\qedhere
    \end{aligned}$
\end{proof}

\subsection{Weyl operators}
\label{sec:Weyl_operators}

Define the unitary shift and clock operators $\mathsf{X}$ and $\mathsf{Z}$, respectively, as
\begin{align*}
    \mathsf{X}|q\rangle = |q+1\rangle, \quad \mathsf{Z}|q\rangle = \omega^q |q\rangle, \quad \forall q\in\mathbb{F}_p.
\end{align*}
The clock and shift operators generate the $n$-qudit Pauli group $\mathscr{P}_p^n$ over $\mathbb{F}_p$, also known as Weyl-Heisenberg group, i.e.,
\begin{align*}
    \mathscr{P}^n_{p} \triangleq \langle \tau\mathbb{I}, \mathsf{X}, \mathsf{Z}\rangle^{\otimes n}. 
\end{align*}
The corresponding $n$-qudit Clifford group $\mathscr{C}^n_{p}$ is defined as the normaliser of the Pauli group in the unitary group, modulo phases, i.e.,
\begin{align*}
    \mathscr{C}^n_{p} \triangleq \{\mathcal{U}\in\mathbb{C}^{p^n\times p^n}: \mathcal{U}^\dagger\mathcal{U} = \mathbb{I},~ \mathcal{U}\mathscr{P}^n_{p}\mathcal{U}^\dagger = \mathscr{P}^n_{p}\}.
    \end{align*}
Another way to characterise the Pauli group $\mathscr{P}^n_{p}$ is through the
\emph{Weyl operators}, also known as the generalised Pauli operators, defined as
\begin{align*}
    \mathcal{W}_{\mathbf{x}} = \mathcal{W}_{\mathbf{v},\mathbf{w}} = \tau^{\langle \mathbf{v},\mathbf{w}\rangle}(\mathsf{X}^{w_1}\mathsf{Z}^{v_1})\otimes \cdots \otimes (\mathsf{X}^{w_n}\mathsf{Z}^{v_n}), \qquad \forall\mathbf{x} = (\mathbf{v},\mathbf{w})\in\mathbb{Z}^{2n}.
\end{align*}
It is not hard to see that the Weyl operators only depend on $\mathbf{x}$ modulo $p$, since, for any $\mathbf{z} = (\mathbf{v}',\mathbf{w}')\in\mathbb{Z}^{2n}$,
\begin{align*}
    \mathcal{W}_{\mathbf{x} + p\mathbf{z}} = \tau^{p\langle \mathbf{v},\mathbf{w}'\rangle + p\langle \mathbf{v}',\mathbf{w}\rangle + p^2\langle \mathbf{v}',\mathbf{w}'\rangle}\mathcal{W}_{\mathbf{x}} = \mathcal{W}_{\mathbf{x}}.
\end{align*}
It is thus without loss of generality that $\mathbf{x}$ can be restricted to $\mathbb{F}_p^{2n}$. The action of $\mathcal{W}_{\mathbf{x}}$ on the Hilbert space $(\mathbb{C}^p)^{\otimes n}$ is
\begin{align*}
    \mathcal{W}_{\mathbf{x}}|\mathbf{q}\rangle = \mathcal{W}_{\mathbf{v},\mathbf{w}}|\mathbf{q}\rangle = \omega^{\langle \mathbf{q},\mathbf{v}\rangle+2^{-1}\langle \mathbf{v},\mathbf{w}\rangle}|\mathbf{q} + \mathbf{w}\rangle, \qquad \forall \mathbf{q}\in\mathbb{F}_p^n, \quad \forall\mathbf{x} = (\mathbf{v},\mathbf{w})\in\mathbb{F}_p^{2n}.
\end{align*}
The Weyl operators are fundamental in the phase space picture of finite-dimensional quantum mechanics~\cite{wootters1987wigner,appleby2005symmetric,gross2006hudson,de2011linearized}. It is clear that every Weyl operator is an element of the Pauli group $\mathscr{P}_p^n$. Conversely, every Pauli operator is equal to some Weyl operator up to a phase that is a power of $\tau$ or, equivalently since $p$ is odd, of $\omega$. This means that
\begin{align*}
    \mathscr{P}_p^n = \{\omega^s \mathcal{W}_{\mathbf{x}}: \mathbf{x}\in\mathbb{F}_p^{2n}, ~s\in\mathbb{F}_p\}.
\end{align*}
A Pauli gate is any element from $\mathscr{P}^n_p$ and a Clifford gate is any element from $\mathscr{C}^n_p$. More generally, an $n$-qudit quantum gate is any unitary in $\mathbb{C}^{p^n \times p^n}$. We assume that measuring a qudit on the computational basis and single and two-qudit quantum gates from a universal gate set~\cite{Muthukrishnan2000multivalued,brylinski2002universal,Brennen2005criteria,brennen2006efficient} all require constant time to be performed. An example of a universal gate set is the set of all single-qudit gates plus the controlled shift operator $\mathbb{I}_{p^2-p}\otimes \mathsf{X}$~\cite{brennen2006efficient}.

We say that $\{\omega^{s_i}\mathcal{W}_{\mathbf{x}_i}\}_{i\in[\ell]}\subset\mathscr{P}^n_p$, $\ell\in\mathbb{N}$, are dependent if there are $m_1,\dots,m_\ell\in\mathbb{F}_p$, not all $0$, such that $\prod_{i=1}^\ell (\omega^{s_i}\mathcal{W}_{\mathbf{x}_i})^{m_i} = \mathbb{I}$. The order of any $\sigma\in\mathscr{P}^n_{p}\setminus\{\mathbb{I}\}$ is $\operatorname{ord}(\sigma) = p$ since $p$ is prime. This means that $\sigma^p = \mathbb{I}$ and the eigenvalues of $\sigma$ are $\{1,\omega,\dots,\omega^{p-1}\}$. The center of $\mathscr{P}^n_{p}$ is $\mathrm{Z}(\mathscr{P}^n_{p}) \triangleq \{\phi\in\mathscr{P}^n_{p}: \phi\sigma = \sigma\phi, \forall \sigma\in\mathscr{P}^n_{p}\} = \{\omega^{s}\mathbb{I}\in\mathscr{P}^n_{p}:s\in\mathbb{F}_p\}$. 

The following properties hold for the Weyl operators.
\begin{lemma}\label{lem:weyl_properties}
    For any $\mathbf{x},\mathbf{y}\in\mathbb{F}_p^{2n}$,
    \begin{multicols}{2}
    \begin{enumerate}[label=\alph*.]
        \item $\mathcal{W}_{\mathbf{x}}\mathcal{W}_{\mathbf{y}} = \tau^{[ \mathbf{x},\mathbf{y}]}\mathcal{W}_{\mathbf{x}+\mathbf{y}} = \omega^{[\mathbf{x},\mathbf{y}]}\mathcal{W}_{\mathbf{y}}\mathcal{W}_{\mathbf{x}}$;
        \item $\mathcal{W}_{\mathbf{x}}^\dagger = \mathcal{W}_{-\mathbf{x}}$;
        \item $\mathcal{W}_{\mathbf{x}}^m = \mathcal{W}_{m\mathbf{x}} \quad\forall m\in\mathbb{F}_p$;
        \item $\operatorname{Tr}[\mathcal{W}_{\mathbf{y}}^\dagger\mathcal{W}_{\mathbf{x}}] = \operatorname{Tr}[\mathcal{W}_{\mathbf{y}}\mathcal{W}_{\mathbf{x}}^\dagger] = p^{n}\cdot \mathbf{1}[\mathbf{x} = \mathbf{y}]$.
    \end{enumerate}
    \end{multicols}
\end{lemma}
\begin{lemma}
    Let $\ell\in\mathbb{N}$ and $\{\mathcal{W}_{\mathbf{v}_i,\mathbf{w}_i}\}_{i\in[\ell]}$. Let $\mathbf{V},\mathbf{W}\in\mathbb{F}_p^{n\times \ell}$ be the matrices with column vectors $\mathbf{v}_1,\dots,\mathbf{v}_\ell$ and $\mathbf{w}_1,\dots,\mathbf{w}_\ell$, respectively. Then $\mathbf{V}^\top \mathbf{W} = \mathbf{W}^\top \mathbf{V}$ if and only if $\{\mathcal{W}_{\mathbf{v}_i,\mathbf{w}_i}\}_{i\in[\ell]}$ pair-wise commute.
\end{lemma}
\begin{proof}
    The elements $\{\mathcal{W}_{\mathbf{v}_i,\mathbf{w}_i}\}_{i\in[\ell]}$ pair-wise commute if and only if $[(\mathbf{v}_i,\mathbf{w}_i), (\mathbf{v}_j,\mathbf{w}_j)] = 0$ $\forall i,j\in[\ell]$ if and only if $\langle \mathbf{v}_i, \mathbf{w}_j\rangle = \langle \mathbf{v}_j, \mathbf{w}_i\rangle$ $\forall i,j\in[\ell]$, which is equivalent to $\mathbf{V}^\top \mathbf{W} = \mathbf{W}^\top \mathbf{V}$.
\end{proof}

From \cref{lem:weyl_properties}, $\mathcal{W}_{\mathbf{x}}$ and $\mathcal{W}_{\mathbf{y}}$ commute if and only if $[\mathbf{x},\mathbf{y}] = 0$. Moreover, the re-scaled Weyl operators $\{p^{-n/2}\mathcal{W}_{\mathbf{x}}:\mathbf{x}\in\mathbb{F}_p^{2n}\}$ form an orthonormal basis with respect to the Hilbert-Schmidt inner product $\langle \mathcal{A},\mathcal{B}\rangle = \operatorname{Tr}[\mathcal{A}^\dagger \mathcal{B}]$. Therefore, any operator $\mathcal{B}$ on $(\mathbb{C}^p)^{\otimes n}$ can be expanded as
\begin{align*}
    \mathcal{B} = p^{-n/2}\sum_{\mathbf{x}\in\mathbb{F}_p^{2n}} c_{\mathcal{B}}(\mathbf{x})\mathcal{W}_{\mathbf{x}} \quad \text{where}~ c_{\mathcal{B}}:\mathbb{F}_p^{2n}\to \mathbb{C} ~\text{given by}~c_{\mathcal{B}}(\mathbf{x}) = p^{-n/2}\operatorname{Tr}[\mathcal{W}_{\mathbf{x}}^\dagger \mathcal{B}]
\end{align*}
is the \emph{characteristic function} of the operator $\mathcal{B}$. We note that $\operatorname{Tr}[\mathcal{A}^\dagger \mathcal{B}] = \sum_{\mathbf{x}\in\mathbb{F}_p^{2n}}\overline{c_{\mathcal{A}}(\mathbf{x})}c_{\mathcal{B}}(\mathbf{x})$. In particular, for any quantum state $|\psi\rangle\in(\mathbb{C}^p)^{\otimes n}$,
\begin{align*}
    \psi \triangleq |\psi\rangle\langle\psi| = p^{-n/2}\sum_{\mathbf{x}\in\mathbb{F}_p^{2n}} c_{\psi}(\mathbf{x})\mathcal{W}_{\mathbf{x}} \quad\text{where}~c_\psi(\mathbf{x}) = p^{-n/2}\operatorname{Tr}[\mathcal{W}_{\mathbf{x}}^\dagger \psi] = p^{-n/2} \langle \psi|\mathcal{W}_{\mathbf{x}}^\dagger|\psi\rangle.
\end{align*}
Note that $\overline{c_\psi(\mathbf{x})} = c_\psi(-\mathbf{x})$ since $\mathcal{W}_{\mathbf{x}}^\dagger = \mathcal{W}_{-\mathbf{x}}$. From the characteristic function of the pure state $\psi$ we define its \emph{characteristic distribution}
\begin{align*}
    p_{\psi}:\mathbb{F}_p^{2n}\to[0,p^{-n}], \quad p_{\psi}(\mathbf{x}) = |c_{\psi}(\mathbf{x})|^2 = p^{-n} |\langle \psi|\mathcal{W}_{\mathbf{x}}|\psi\rangle|^2 = p^{-n}\operatorname{Tr}[\psi\mathcal{W}_{\mathbf{x}}\psi\mathcal{W}_{\mathbf{x}}^\dagger].
\end{align*}
Note that $p_\psi(-\mathbf{x}) = p_\psi(\mathbf{x})$ and that $p_\psi(\mathbf{x}) \leq p^{-n}$ since $|\langle\psi|\mathcal{W}_{\mathbf{x}}|\psi\rangle| \leq 1$ for all $\mathbf{x}\in\mathbb{F}_p^{2n}$. To check that $p_\psi$ is indeed a probability distribution, note that $\sum_{\mathbf{x}\in\mathbb{F}_p^{2n}}p_{\psi}(\mathbf{x}) = \sum_{\mathbf{x}\in\mathbb{F}_p^{2n}}\overline{c_{\psi}(\mathbf{x})}c_{\psi}(\mathbf{x}) = \operatorname{Tr}[|\psi\rangle\langle\psi|\psi\rangle\langle\psi|] = 1$. 

It is possible to combine the characteristic distribution and the involution into a new function 
\begin{align*}
    j_\psi:\mathbb{F}_p^{2n}\to[0,p^{-n}], \quad j_\psi(\mathbf{x}) = p_\psi(J(\mathbf{x})),
\end{align*}
called the \emph{involution characteristic distribution}. Note that $j_\psi$ is also a probability distribution since $\sum_{\mathbf{x}\in\mathbb{F}_p^{2n}}j_\psi(\mathbf{x}) = \sum_{\mathbf{x}\in\mathbb{F}_p^{2n}}p_\psi(\mathbf{x})$ by a simple change of variable. The next result, which is a slight generalisation of~\cite[Eq.~(3.5)]{gross2021schur} for $p>2$, shows that both $p_\psi$ and $j_\psi$ are invariant (up to renormalisation) under the symplectic Fourier transform.
\begin{lemma}\label{lem:invariant}
    For any $|\psi\rangle\in(\mathbb{C}^p)^{\otimes n}$ and $\mathbf{x}\in\mathbb{F}_p^{2n}$, $\widehat{p_\psi}(\mathbf{x)} = p^{-n}p_\psi(\mathbf{x})$, $\widehat{j_\psi}(\mathbf{x)} = p^{-n}j_\psi(\mathbf{x})$, and $\widehat{p_\psi}(J(\mathbf{x})) = \widehat{j_\psi}(\mathbf{x})$.
\end{lemma}
\begin{proof}
    With the aid of Parseval's identity,
    \begin{align*}
        \widehat{p_\psi}(\mathbf{x)} &= \frac{1}{p^{2n}}\sum_{\mathbf{y}\in\mathbb{F}_p^{2n}}\omega^{[\mathbf{x},\mathbf{y}]}p_\psi(\mathbf{y}) = \frac{1}{p^{3n}}\sum_{\mathbf{y}\in\mathbb{F}_p^{2n}}\omega^{[\mathbf{x},\mathbf{y}]}\langle \psi|\mathcal{W}_{\mathbf{y}}|\psi\rangle \langle \psi|\mathcal{W}_{\mathbf{y}}^\dagger|\psi\rangle\\
        &= \frac{1}{p^{3n}}\sum_{\mathbf{y}\in\mathbb{F}_p^{2n}} \langle \psi|\mathcal{W}_{\mathbf{y}}|\psi\rangle \langle \psi|\mathcal{W}_{\mathbf{x}}^\dagger\mathcal{W}_{\mathbf{y}}^\dagger\mathcal{W}_{\mathbf{x}}|\psi\rangle = \frac{1}{p^{2n}}\sum_{\mathbf{y}\in\mathbb{F}_p^{2n}} \overline{c_\psi(\mathbf{y})} c_{\mathcal{W}_{\mathbf{x}} \psi \mathcal{W}_{\mathbf{x}}^\dagger}(\mathbf{y})\\
        &= \frac{1}{p^{2n}}\operatorname{Tr}[\psi\mathcal{W}_{\mathbf{x}} \psi \mathcal{W}_{\mathbf{x}}^\dagger] = p^{-n}p_\psi(\mathbf{x}). 
    \end{align*}
    From the above we can also see that
    \begin{align*}
        p^{2n}\widehat{p_\psi}(J(\mathbf{x))} = \sum_{\mathbf{y}\in\mathbb{F}_p^{2n}}\omega^{[J(\mathbf{x}),\mathbf{y}]}p_\psi(\mathbf{y}) = \sum_{\mathbf{y}\in\mathbb{F}_p^{2n}}\omega^{[J(\mathbf{x}),J(\mathbf{y})]}p_\psi(J(\mathbf{y})) = \sum_{\mathbf{y}\in\mathbb{F}_p^{2n}}\omega^{-[\mathbf{x},\mathbf{y}]}j_\psi(\mathbf{y}) = p^{2n}\overline{\widehat{j_\psi}(\mathbf{x})}.
    \end{align*}
    Since $p_\psi$ and $j_\psi$ are real-valued, the above implies that $\widehat{p_\psi}$ and $\widehat{j_\psi}$ are also real-valued. Finally, $\widehat{j_\psi}(\mathbf{x)} = \widehat{p_\psi}(J(\mathbf{x})) = p^{-n}p_\psi(J(\mathbf{x})) = p^{-n}j_\psi(\mathbf{x})$.
\end{proof}

\subsection{Stabiliser groups}

In this section, we review \emph{stabiliser groups}, which are among the most important subgroups of $\mathscr{P}_p^n$. For completeness and the reader's convenience (and at the risk of being wordy), we provide proofs for a number of important claims. Longer proofs can be found in \cref{app:proofs}. We point the reader to~\cite{appleby2005symmetric,gross2006hudson,appleby2009properties,de2011linearized,gross2021schur} for more information. 
\begin{definition}[Stabiliser group]
    A stabiliser group $\mathcal{S}$ is a maximal subgroup of $\mathscr{P}^n_p$ which contains only the identity from the center of $\mathscr{P}^n_p$, i.e., $\mathcal{S}\cap\mathrm{Z}(\mathscr{P}_p^n) = \{\mathbb{I}\}$.
\end{definition}

Stabiliser groups are well studied and several properties are known, even in the case when $p$ is not prime (see e.g.~\cite{gross2006hudson,de2011linearized}). The next lemmas cover some of these properties. 
\begin{lemma}\label{lem:stabiliser_group_form}
    Any stabiliser group $\mathcal{S}\subset\mathscr{P}^n_p$ can be written as $\mathcal{S} = \{\omega^{[\mathbf{a},\mathbf{x}]} \mathcal{W}_{\mathbf{x}} :  \mathbf{x}\in \mathscr{M}\}$, where $\mathbf{a}\in\mathbb{F}_p^{2n}$ and $\mathscr{M}\subset\mathbb{F}_p^{2n}$ is a Lagrangian subspace such that $\operatorname{dim}(\mathscr{M}) = n$. As a consequence,
    $\mathcal{S}$ is commutative and $|\mathcal{S}| = p^n$.
\end{lemma}

Given a basis $\{\mathbf{x}_1,\dots,\mathbf{x}_n\}$ of $\mathscr{M}$, we can write $\mathcal{S} = \langle \{\omega^{[\mathbf{a},\mathbf{x}_i]}\mathcal{W}_{\mathbf{x}_i}\}_{i\in[n]}\rangle$ since $\{\mathbf{x}_i\}_{i\in[n]}$ are linearly independent if and only if $\{\omega^{[\mathbf{a},\mathbf{x}_i]}\mathcal{W}_{\mathbf{x}_i}\}_{i\in[n]}$ are independent as proven below.
\begin{lemma}\label{lem:linearly_independency}
    Given $\ell\in\mathbb{N}$ and a stabiliser group $\mathcal{S}\subset \mathscr{P}^n_p$, then $\{\omega^{s_i}\mathcal{W}_{\mathbf{x}_i}\}_{i\in[\ell]}\subset\mathcal{S}$ are dependent if and only if $\{\mathbf{x}_i\}_{i\in[\ell]}\subset\mathbb{F}_p^{2n}$ are linearly dependent.
\end{lemma}
\begin{proof}
    It follows from $\prod_{i=1}^\ell(\omega^{s_i}\mathcal{W}_{\mathbf{x}_i})^{m_i} = \omega^{\sum_{i=1}^\ell m_i s_i} \mathbb{I} = \mathbb{I} \iff \sum_{i=1}^\ell m_i\mathbf{x}_i = \mathbf{0}$, for some $m_1,\dots,m_\ell\in\mathbb{F}_p$, not all $0$, where the equality $\omega^{\sum_{i=1}^\ell m_i s_i} \mathbb{I} = \mathbb{I}$ holds since $\mathcal{S}\cap\mathrm{Z}(\mathscr{P}^n_p) = \{\mathbb{I}\}$.
\end{proof}
Note the subspace $\mathscr{M}\subseteq\mathbb{F}_p^{2n}$ can be written as $\mathscr{M} = {\operatorname{col}}\big(\big[\begin{smallmatrix} \mathbf{V} \\ \mathbf{W} \end{smallmatrix}\big]\big)$ for matrices $\mathbf{V},\mathbf{W}\in\mathbb{F}_p^{n\times n}$. In the next lemma, we explore several properties of the matrices $\mathbf{V},\mathbf{W}$ for when $\mathscr{M}$ is Lagrangian.
\begin{lemma}\label{lem:properties}
    Let $\mathcal{S} = \{\omega^{[\mathbf{a},\mathbf{x}]}\mathcal{W}_{\mathbf{x}}:\mathbf{x}\in\mathscr{M}\}$ be a stabiliser group. Let $\mathbf{V},\mathbf{W}\in\mathbb{F}_p^{n\times n}$ be matrices such that $\mathscr{M} = {\operatorname{col}}\big(\big[\begin{smallmatrix} \mathbf{V} \\ \mathbf{W} \end{smallmatrix}\big]\big)$. Then
    \begin{multicols}{2}
    \begin{enumerate}[label=\alph*.]
        \item $\mathbf{V}^\top \mathbf{W} = \mathbf{W}^\top \mathbf{V}$;
        \item ${\operatorname{rank}}\big(\bigl[\begin{smallmatrix} \mathbf{V} \\ \mathbf{W} \end{smallmatrix} \bigr]\big) = \operatorname{dim}(\mathscr{M}) = n$;
        \item $\operatorname{row}(\mathbf{V}) + \operatorname{row}(\mathbf{W}) = \mathbb{F}_p^n$;
        \item $\operatorname{null}(\mathbf{V})\cap\operatorname{null}(\mathbf{W}) = \{\mathbf{0}\}$;
        \item $\operatorname{null}([\begin{smallmatrix} \mathbf{V}^\top & \mathbf{W}^\top \end{smallmatrix} ]) = {\operatorname{col}}\big(\big[\begin{smallmatrix} \mathbf{W} \\ -\mathbf{V} \end{smallmatrix}\big]\big)$;
        \item $\operatorname{null}(\mathbf{V}^\top \mathbf{W}) = \operatorname{null}(\mathbf{V}) + \operatorname{null}(\mathbf{W})$;
        \item $\operatorname{null}(\mathbf{W}^\top) \subseteq \operatorname{col}(\mathbf{V})$;
        \item $\operatorname{null}(\mathbf{V}^\top) \subseteq \operatorname{col}(\mathbf{W})$.
    \end{enumerate}
    \end{multicols}
\end{lemma}

Together with the concept of stabiliser group is the concept of stabiliser state, which is the joint $+1$-eigenstate of all elements of a stabiliser group.
\begin{definition}[Stabiliser state]
    A non-zero state $|\psi\rangle\in(\mathbb{C}^{p})^{\otimes n}$ is \emph{stabilised} by $\sigma\in\mathscr{P}^n_p$ if $\sigma |\psi\rangle = |\psi\rangle$. A non-zero state $|\Psi\rangle\in(\mathbb{C}^{p})^{\otimes n}$ is a \emph{stabiliser state} of a stabiliser group $\mathcal{S}$ if $\sigma |\Psi\rangle = |\Psi\rangle$ for all $\sigma\in\mathcal{S}$. Let $\EuScript{S}_p^n$ be the set of all stabiliser states in $(\mathbb{C}^p)^{\otimes n}$.
\end{definition}

Consider the group $\mathcal{G} = \{\omega^{[\mathbf{a},\mathbf{x}]}\mathcal{W}_{\mathbf{x}}:\mathbf{x}\in\mathscr{X}\}$ where $\mathscr{X}\subset\mathbb{F}_p^{2n}$ is an isotropic subspace and let $V_{\mathcal{G}} \triangleq \{|\Psi\rangle\in(\mathbb{C}^{p})^{\otimes n}: \sigma|\Psi\rangle = |\Psi\rangle, \forall \sigma\in\mathcal{G}\}$. The next two lemmas construct an projector onto $V_{\mathcal{G}}$ and show that $\operatorname{dim}(V_{\mathcal{G}}) = p^{n-\operatorname{dim}(\mathscr{X})}$. As an immediate consequence, any stabiliser group has a unique stabiliser state. This result is also true for $p$ non-prime~\cite{gross2006hudson,gross2013stabilizer,gross2021schur}.
\begin{lemma}\label{lem:projection}
    Let $\mathcal{G} = \{\omega^{[\mathbf{a},\mathbf{x}]}\mathcal{W}_{\mathbf{x}}:\mathbf{x}\in\mathscr{X}\}$ be a group where $\mathscr{X}\subset\mathbb{F}_p^{2n}$ is an isotropic subspace, and let $V_{\mathcal{G}} \triangleq \{|\Psi\rangle\in(\mathbb{C}^{p})^{\otimes n}: \omega^{[\mathbf{a},\mathbf{x}]}\mathcal{W}_{\mathbf{x}}|\Psi\rangle = |\Psi\rangle, \forall \mathbf{x}\in\mathscr{X}\}$. Then the projector onto $V_{\mathcal{G}}$ is
    \begin{align*}
        \mathcal{P}_{\mathcal{G}} = \frac{1}{|\mathscr{X}|}\sum_{\mathbf{x}\in\mathscr{X}}\omega^{[\mathbf{a},\mathbf{x}]}\mathcal{W}_{\mathbf{x}}.
    \end{align*}
\end{lemma}
\begin{proof}
    It is easy to see that $\mathcal{P}_{\mathcal{G}}$ is Hermitian by the change of index $\mathbf{x} \mapsto -\mathbf{x}$. It is also idempotent,
    \begin{align*}
        \mathcal{P}_{\mathcal{G}}^2 = \frac{1}{|\mathscr{X}|^2}\sum_{\mathbf{x},\mathbf{y}\in\mathscr{X}}\omega^{[\mathbf{a},\mathbf{x}+\mathbf{y}]}\mathcal{W}_{\mathbf{x}}\mathcal{W}_{\mathbf{y}} = \frac{1}{|\mathscr{X}|^2}\sum_{\mathbf{x},\mathbf{y}\in\mathscr{X}}\omega^{[\mathbf{a},\mathbf{x}+\mathbf{y}]}\mathcal{W}_{\mathbf{x}+\mathbf{y}} = \frac{1}{|\mathscr{X}|^2}\sum_{\mathbf{x},\mathbf{z}\in\mathscr{X}}\omega^{[\mathbf{a},\mathbf{z}]}\mathcal{W}_{\mathbf{z}} = \mathcal{P}_{\mathcal{G}}.
    \end{align*}
    Therefore, $\mathcal{P}_{\mathcal{G}}$ is a projector. Moreover, $\mathcal{P}_{\mathcal{G}}|\Psi\rangle = |\Psi\rangle$ for any $|\Psi\rangle\in V_{\mathcal{G}}$. On the other hand, given $\mathcal{P}_{\mathcal{G}}|\psi\rangle$ for any $|\psi\rangle$, $\omega^{[\mathbf{a},\mathbf{x}]}\mathcal{W}_{\mathbf{x}}(\mathcal{P}_{\mathcal{G}}|\psi\rangle) = \mathcal{P}_{\mathcal{G}}|\psi\rangle$ for all $\mathbf{x}\in\mathscr{X}$, and thus $\mathcal{P}_{\mathcal{G}}|\psi\rangle\in V_{\mathcal{G}}$. This proves that $\mathcal{P}_{\mathcal{G}}$ is the projector onto $V_{\mathcal{G}}$.
\end{proof}

\begin{lemma}\label{lem:uniqueness}
    Let $\mathcal{G} = \{\omega^{[\mathbf{a},\mathbf{x}]}\mathcal{W}_{\mathbf{x}}:\mathbf{x}\in\mathscr{X}\}$ be a group where $\mathscr{X}\subset\mathbb{F}_p^{2n}$ is an isotropic subspace. Then $\operatorname{dim}(V_{\mathcal{G}}) = p^{n-\operatorname{dim}(\mathscr{X})}$, where $V_{\mathcal{G}} \triangleq \{|\Psi\rangle\in(\mathbb{C}^{p})^{\otimes n}: \omega^{[\mathbf{a},\mathbf{x}]}\mathcal{W}_{\mathbf{x}}|\Psi\rangle = |\Psi\rangle, \forall \mathbf{x}\in\mathscr{X}\}$. It then follows that any stabiliser group $\mathcal{S}$ has a unique stabiliser state, i.e., $\operatorname{dim}(V_{\mathcal{S}}) = 1$.
\end{lemma}

Following~\cite{gross2006hudson,gross2021schur}, we can denote the unique stabiliser state $|\mathcal{S}\rangle$ of a stabiliser group $\mathcal{S} = \{\omega^{[\mathbf{a},\mathbf{x}]}\mathcal{W}_{\mathbf{x}}:\mathbf{x}\in\mathscr{M}\}$ as $|\mathscr{M},\mathbf{a}\rangle$, where $\mathscr{M}\subset\mathbb{F}_p^{2n}$ and $\mathbf{a}\in\mathbb{F}_p^{2n}$ are the Lagrangian subspace and phase string defining $\mathcal{S}$, respectively. We note that the choice of $\mathbf{a}$ given $\mathscr{M}$ is not unique, but can be replaced with $\mathbf{a} + \mathbf{z}$ for any $\mathbf{z}\in\mathscr{M}$ and still define the same stabiliser group $\mathcal{S}$. In other words, any string from a coset $\mathcal{C}\in \mathbb{F}_p^{2n}/\mathscr{M}$ yields the same stabiliser state. We can then also write a stabiliser state as $|\mathscr{M},\mathcal{C}\rangle$ where $\mathcal{C}\in \mathbb{F}_p^{2n}/\mathscr{M}$. Conversely, given a Lagrangian subspace $\mathscr{M}$, any state that is a simultaneous eigenvector of $\{\mathcal{W}_{\mathbf{x}}\}_{\mathbf{x}\in\mathscr{M}}$ is a stabiliser state of some stabiliser group. The eigenvectors of $\{\mathcal{W}_{\mathbf{x}}\}_{\mathbf{x}\in\mathscr{M}}$ thus determine a stabiliser basis $\{|\mathscr{M},\mathcal{C}\rangle\}_{\mathcal{C}\in\mathbb{F}_p^{2n}/\mathscr{M}}$.

Consider a stabiliser group $\mathcal{S} = \{\omega^{[\mathbf{a},\mathbf{x}]}\mathcal{W}_{\mathbf{x}}:\mathbf{x}\in\mathscr{M}\}$. According to \cref{lem:projection}, the operator $p^{-n}\sum_{\mathbf{x}\in\mathscr{M}}\omega^{[\mathbf{a},\mathbf{x}]}\mathcal{W}_{\mathbf{x}}$ is the orthogonal projection onto the set $V_{\mathcal{S}}$ spanned by the stabiliser state $|\mathcal{S}\rangle$ of $\mathcal{S}$. Therefore, 
\begin{align}\label{eq:stabiliser_form}
    |\mathcal{S}\rangle\langle\mathcal{S}| = \frac{1}{p^{n}}\sum_{\mathbf{x}\in\mathscr{M}}\omega^{[\mathbf{a},\mathbf{x}]}\mathcal{W}_{\mathbf{x}}
\end{align}
and the characteristic function and probability of $|\mathcal{S}\rangle\langle\mathcal{S}|$ are
\begin{align}\label{eq:stabiliser_distribution}
    c_{\mathcal{S}}(\mathbf{x}) = \begin{cases}
        p^{-n/2}\omega^{[\mathbf{a},\mathbf{x}]} &\text{if}~\mathbf{x}\in\mathscr{M},\\
        0 &\text{otherwise},
    \end{cases} \qquad \qquad\text{and} \qquad \qquad 
    p_{\mathcal{S}}(\mathbf{x}) = \begin{cases}
        p^{-n} &\text{if}~\mathbf{x}\in\mathscr{M},\\
        0 &\text{otherwise}.
    \end{cases}
\end{align}
Thus $p_{\mathcal{S}}$ is the uniform distribution on the subspace $\mathscr{M}$, a fact which will be of paramount importance when studying Bell (difference) sampling in the next section. While we have obtained an expression for $|\mathcal{S}\rangle\langle\mathcal{S}|$, in the following lemma, we obtain an expression for $|\mathcal{S}\rangle$, which will be the basis of our quantum algorithm for learning stabiliser states in \cref{sec:learning_stabiliser_states}. We note that a similar form for stabiliser states was given in~\cite[Theorem~1]{hostens2005stabilizer}.

\begin{lemma}\label{lem:normalisation}
    Let $\mathcal{S}=\langle\{\omega^{s_i}\mathcal{W}_{\mathbf{v}_i,\mathbf{w}_i}\}_{i\in[n]}\rangle$ be a stabiliser group.
    Let $\mathbf{s} = (s_1,\dots,s_n)\in\mathbb{F}_p^n$ and let $\mathbf{V},\mathbf{W}\in\mathbb{F}_p^{n\times n}$ be the matrices with column vectors $\mathbf{v}_1,\dots,\mathbf{v}_n$ and $\mathbf{w}_1,\dots,\mathbf{w}_n$, respectively. Then there exists $\mathbf{u}\in\mathbb{F}_p^n$ such that $\mathbf{V}^\top \mathbf{u} + \mathbf{s} \in \operatorname{row}(\mathbf{W})$, and the unique stabiliser state of $\mathcal{S}$ is
    \begin{align*}
        |\mathcal{S}\rangle = \frac{\sqrt{|\operatorname{col}(\mathbf{W})|}}{p^n}\sum_{\mathbf{q}\in\mathbb{F}_p^n}\omega^{\mathbf{s}^\top \mathbf{q} + \mathbf{u}^\top \mathbf{V}\mathbf{q} + 2^{-1}\mathbf{q}^\top \mathbf{V}^\top \mathbf{W}\mathbf{q}}|\mathbf{u}+\mathbf{W}\mathbf{q}\rangle
    \end{align*}
    for any $\mathbf{u}\in\mathbb{F}_p^n$ such that $\mathbf{V}^\top \mathbf{u} + \mathbf{s} \in \operatorname{row}(\mathbf{W})$.
\end{lemma}
\begin{proof}
    First note that the unormalised state $|\mathcal{S}_{\mathbf{u}}\rangle \triangleq \sum_{\sigma\in \mathcal{S}}\sigma|\mathbf{u}\rangle$, for any $\mathbf{u}\in\mathbb{F}_p^n$, is either the stabiliser state $|\mathcal{S}\rangle$ of $\mathcal{S}$ (up to normalisation) or the zero state, since $p^{-n}\sum_{\sigma\in\mathcal{S}}\sigma$ is the orthogonal projection onto $|\mathcal{S}\rangle$. Because $\{\omega^{s_i}\mathcal{W}_{\mathbf{v}_i,\mathbf{w}_i}\}_{i\in[n]}$ are independent,
    \begin{align*}
         |\mathcal{S}_{\mathbf{u}}\rangle = \sum_{\sigma\in \mathcal{S}}\sigma|\mathbf{u}\rangle &= \sum_{\mathbf{q}\in\mathbb{F}_p^n} (\omega^{s_n}\mathcal{W}_{\mathbf{v}_n,\mathbf{w}_n})^{q_n}\cdots (\omega^{s_1}\mathcal{W}_{\mathbf{v}_1,\mathbf{w}_1})^{q_1} |\mathbf{u}\rangle\\
         &= \sum_{\mathbf{q}\in\mathbb{F}_p^n} \omega^{\mathbf{s}^\top \mathbf{q}}\mathcal{W}_{q_n\mathbf{v}_n,q_n\mathbf{w}_n}\cdots \mathcal{W}_{q_1\mathbf{v}_1,q_1\mathbf{w}_1} |\mathbf{u}\rangle\\
         &= \sum_{\mathbf{q}\in\mathbb{F}_p^n}\omega^{\mathbf{s}^\top \mathbf{q}}\mathcal{W}_{\mathbf{V}\mathbf{q},\mathbf{W}\mathbf{q}}|\mathbf{u}\rangle \tag{$[(\mathbf{v}_i,\mathbf{w}_i),(\mathbf{v}_j,\mathbf{w}_j)] = 0$ $\forall i,j\in[n]$}\\
         &= \sum_{\mathbf{q}\in\mathbb{F}_p^n}\omega^{\mathbf{s}^\top \mathbf{q} + \mathbf{u}^\top \mathbf{V}\mathbf{q} + 2^{-1}\mathbf{q}^\top \mathbf{V}^\top \mathbf{W}\mathbf{q}}|\mathbf{u}+\mathbf{W}\mathbf{q}\rangle\\
         &= \sum_{\mathbf{q}\in\mathbb{F}_p^n}\omega^{f(\mathbf{q})}|\mathbf{u}+\mathbf{W}\mathbf{q}\rangle,
    \end{align*}
    where we defined $f:\mathbb{F}_p^n\to\mathbb{F}_p$ as $f(\mathbf{q}) = \mathbf{s}^\top \mathbf{q} + \mathbf{u}^\top \mathbf{V}\mathbf{q} + 2^{-1}\mathbf{q}^\top \mathbf{V}^\top \mathbf{W}\mathbf{q}$. To analyse the case when $|\mathcal{S}_{\mathbf{u}}\rangle$ is the stabiliser state, we compute $|\langle \mathcal{S}_{\mathbf{u}}|\mathcal{S}_{\mathbf{u}}\rangle|$ as follows:
    \begin{align*}
        |\langle \mathcal{S}_{\mathbf{u}}|\mathcal{S}_{\mathbf{u}}\rangle|  &= \left|\sum_{\mathbf{q},\mathbf{q}'\in\mathbb{F}_p^n}\omega^{f(\mathbf{q}') - f(\mathbf{q})}\langle \mathbf{u}+\mathbf{W}\mathbf{q}|\mathbf{u}+\mathbf{W}\mathbf{q}'\rangle \right| = \left|\sum_{\mathbf{q}\in\mathbb{F}_p^n}\sum_{\mathbf{q}'\in \operatorname{null}(\mathbf{W})}\omega^{f(\mathbf{q}+\mathbf{q}') - f(\mathbf{q})}\right|.
    \end{align*}
    Since $f(\mathbf{q}+\mathbf{q}') - f(\mathbf{q}) = (\mathbf{s}^\top + \mathbf{u}^\top \mathbf{V} + \mathbf{q}^\top \mathbf{V}^\top \mathbf{W}) \mathbf{q}' + 2^{-1}\mathbf{q}'^\top \mathbf{V}^\top \mathbf{W} \mathbf{q}'$ (using that $\mathbf{V}^\top \mathbf{W} = \mathbf{W}^\top \mathbf{V}$), then
    \begin{align*}
        |\langle \mathcal{S}_{\mathbf{u}}|\mathcal{S}_{\mathbf{u}}\rangle| = \left|\sum_{\mathbf{q}\in\mathbb{F}_p^n}\sum_{\mathbf{q}'\in \operatorname{null}(\mathbf{W})}\omega^{(\mathbf{s} + \mathbf{V}^\top \mathbf{u})^\top \mathbf{q}'}\right| = p^n \left|\sum_{\mathbf{q}'\in \operatorname{null}(\mathbf{W})}\omega^{(\mathbf{s} + \mathbf{V}^\top \mathbf{u})^\top \mathbf{q}'}\right|.
    \end{align*}
    According to \cref{lem:sum},
    \begin{align*}
        \sum_{\mathbf{q}'\in \operatorname{null}(\mathbf{W})} \omega^{(\mathbf{s} + \mathbf{V}^\top\mathbf{u})^\top\mathbf{q}'} = \begin{cases}
            |\operatorname{null}(\mathbf{W})| &\text{if}~ \mathbf{s}+\mathbf{V}^\top \mathbf{u}\in \operatorname{null}(\mathbf{W})^\perp = \operatorname{row}(\mathbf{W}),\\
            0 &\text{if}~ \mathbf{s}+\mathbf{V}^\top \mathbf{u}\notin \operatorname{null}(\mathbf{W})^\perp = \operatorname{row}(\mathbf{W}).
            \end{cases}
    \end{align*}
    Thus $|\langle \mathcal{S}_{\mathbf{u}}|\mathcal{S}_{\mathbf{u}}\rangle| = p^n |\operatorname{null}(\mathbf{W})| = p^{2n}/|\operatorname{col}(\mathbf{W})|$ if $\mathbf{V}^\top \mathbf{u} + \mathbf{s}\in \operatorname{row}(\mathbf{W})$, and $|\langle \mathcal{S}_{\mathbf{u}}|\mathcal{S}_{\mathbf{u}}\rangle| = 0$ otherwise. It only remains to prove that one can always find $\mathbf{u}\in\mathbb{F}_p^n$ such that $\mathbf{V}^\top \mathbf{u} + \mathbf{s}\in \operatorname{row}(\mathbf{W})$, from which the corresponding $|\mathcal{S}_{\mathbf{u}}\rangle$ must be the stabiliser state of $\mathcal{S}$. But this is straightforward since $\{\mathbf{V}^\top \mathbf{u}:\mathbf{u}\in\mathbb{F}_p^n\} \triangleq \operatorname{row}(\mathbf{V})$ and, according to \cref{lem:properties}, $\operatorname{row}(\mathbf{V}) + \operatorname{row}(\mathbf{W}) = \mathbb{F}_p^n$. Thus, any $\mathbf{s}$ can be expressed as $\mathbf{V}^\top \mathbf{q}_1 + \mathbf{W}^\top \mathbf{q}_2$ for some $\mathbf{q}_1, \mathbf{q}_2\in\mathbb{F}_p^n$. Choosing $\mathbf{u} = -\mathbf{q}_1$ then gives $\mathbf{V}^\top \mathbf{u} + \mathbf{s} = \mathbf{W}^\top \mathbf{q}_2 \in \operatorname{row}(\mathbf{W})$.
\end{proof}

Sometimes we are interested in the Pauli operators that stabilise a given pure state $|\psi\rangle$ up to a phase. This is captured by the \emph{unsigned stabiliser group} of $|\psi\rangle$.
\begin{definition}[Unsigned stabiliser group and stabiliser dimension]
    Given $|\psi\rangle\in(\mathbb{C}^p)^{\otimes n}$, its \emph{unsigned stabiliser group} is $\operatorname{Weyl}(|\psi\rangle) \triangleq \{\mathbf{x}\in\mathbb{F}_p^{2n}:\mathcal{W}_{\mathbf{x}}|\psi\rangle = \omega^s |\psi\rangle ~\text{for some}~ s\in\mathbb{F}_p\}$ and its \emph{stabiliser dimension} is $\operatorname{dim}(\operatorname{Weyl}(|\psi\rangle))$.
\end{definition}
\begin{lemma}
    For any non-zero $|\psi\rangle\in(\mathbb{C}^p)^{\otimes n}$, $\operatorname{Weyl}(|\psi\rangle)$ is isotropic and $\operatorname{Weyl}(|\psi\rangle) \subseteq \operatorname{Weyl}(|\psi\rangle)^{\independent}$.
\end{lemma}
\begin{proof}
    Suppose there are $\mathbf{x},\mathbf{y}\in\operatorname{Weyl}(|\psi\rangle)$ such that $[\mathbf{x},\mathbf{y}] \neq 0$. On the one hand, $\mathcal{W}_{\mathbf{x}}\mathcal{W}_{\mathbf{y}}|\psi\rangle = \omega^{s(\mathbf{x}) + s(\mathbf{y})}|\psi\rangle$, while on the other hand, $\mathcal{W}_{\mathbf{x}}\mathcal{W}_{\mathbf{y}}|\psi\rangle = \omega^{[\mathbf{x},\mathbf{y}]}\mathcal{W}_{\mathbf{y}}\mathcal{W}_{\mathbf{x}}|\psi\rangle = \omega^{[\mathbf{x},\mathbf{y}] + s(\mathbf{x}) + s(\mathbf{y})}|\psi\rangle$, which is a contradiction. The inclusion $\operatorname{Weyl}(|\psi\rangle) \subseteq \operatorname{Weyl}(|\psi\rangle)^{\independent}$ follows from the isotropy property.
\end{proof}
\noindent The stabiliser dimension of a state is invariant under Clifford transformations as a corollary of the next result.
\begin{lemma}[{\cite[Lemma~2.1]{gross2021schur}}]
    For each $\mathcal{U}\in \mathscr{C}_p^n$, there is a symplectic matrix $\Gamma\in\mathbb{F}_p^{2n\times 2n}$ and a function $f:\mathbb{F}_p^{2n}\to\mathbb{F}_p$ such that $\mathcal{U}\mathcal{W}_{\mathbf{x}}\mathcal{U}^\dagger = \omega^{f(\mathbf{x})}\mathcal{W}_{\Gamma\mathbf{x}}$ for all $\mathbf{x}\in\mathbb{F}_p^{2n}$.
\end{lemma}

\begin{corollary}\label{cor:invariant}
    The stabiliser dimension is invariant under Clifford transformations. In other words, $\operatorname{dim}(\operatorname{Weyl}(\mathcal{U}|\psi\rangle)) = \operatorname{dim}(\operatorname{Weyl}(|\psi\rangle))$ for all $\mathcal{U}\in\mathscr{C}_p^n$ and $|\psi\rangle\in(\mathbb{C}^p)^{\otimes n}$.
\end{corollary}

The next lemma covers another important property of unsigned stabiliser groups: the stabiliser dimension is at most $n$ and is maximised for stabiliser states. See \cref{app:proofs} for an extended version of \cref{lem:stabiliser_dimension}.
\begin{lemma}\label{lem:stabiliser_dimension}
    Let $\mathcal{G} = \{\omega^{[\mathbf{a},\mathbf{x}]}\mathcal{W}_{\mathbf{x}}:\mathbf{x}\in\mathscr{X}\}$ be a group where $\mathscr{X}\subset\mathbb{F}_p^{2n}$ is an isotropic subspace. Let $V_{\mathcal{G}} \triangleq \{|\Psi\rangle\in(\mathbb{C}^{p})^{\otimes n}: \omega^{[\mathbf{a},\mathbf{x}]}\mathcal{W}_{\mathbf{x}}|\Psi\rangle = |\Psi\rangle, \forall \mathbf{x}\in\mathscr{X}\}$. Then 
    \begin{align*}
        \mathscr{X} \subseteq \operatorname{Weyl}(|\Psi\rangle) \subseteq \operatorname{Weyl}(|\Psi\rangle)^{\independent} \subseteq \mathscr{X}^{\independent}, \qquad \forall|\Psi\rangle\in V_{\mathcal{G}}.
    \end{align*}
    As a consequence, if $\mathcal{S} = \{\omega^{[\mathbf{a},\mathbf{x}]}\mathcal{W}_{\mathbf{x}}:\mathbf{x}\in\mathscr{M}\}$ is a stabiliser group with stabiliser state $|\mathcal{S}\rangle$, then $\operatorname{Weyl}(|\mathcal{S}\rangle) = \mathscr{M}$. 
\end{lemma}
\begin{proof}
    For any $|\Psi\rangle\in V_{\mathcal{G}}$, the inclusion $\mathscr{X} \subseteq \operatorname{Weyl}(|\Psi\rangle)$ follows from the definition of $\operatorname{Weyl}(|\Psi\rangle)$, while the inclusion $\operatorname{Weyl}(|\Psi\rangle) \subseteq \operatorname{Weyl}(|\Psi\rangle)^{\independent}$ follows from $\operatorname{Weyl}(|\Psi\rangle)$ being isotropic. 
\end{proof}

\begin{corollary}
    For any $|\psi\rangle\in (\mathbb{C}^p)^{\otimes n}$, $\operatorname{dim}(\operatorname{Weyl}(|\psi\rangle)) \leq n$, with equality if and only if $|\psi\rangle$ is a stabiliser state.
\end{corollary}

Alongside its stabiliser dimension, another important measure of the ``stabiliser complexity'' of a quantum state  is its \emph{stabiliser fidelity}~\cite{bravyi2019simulation}.
\begin{definition}[Stabiliser fidelity]
    Given an $n$-qudit quantum state $|\psi\rangle$, its \emph{stabiliser fidelity} is $F_{\mathcal{S}}(|\psi\rangle) \triangleq \max_{|\mathcal{S}\rangle \in \EuScript{S}_p^n} |\langle \mathcal{S}|\psi\rangle|^2$, where $\EuScript{S}_p^n$ is the set of all stabiliser states in $(\mathbb{C}^p)^{\otimes n}$.
\end{definition}

It is possible to relate the stabiliser fidelity $F_{\mathcal{S}}(|\psi\rangle)$ of a quantum state $|\psi\rangle$ to its characteristic distribution $p_{\psi}$ as shown next, which generalises~\cite[Lemma~5.2 and Corollary~7.4]{grewal2023improved}.
\begin{lemma}\label{lem:stabiliser_fidelity_inequality}
    Given a quantum state $|\psi\rangle\in(\mathbb{C}^p)^{\otimes n}$, let $|\mathscr{M},\mathbf{a}\rangle = \argmax_{|\mathcal{S}\rangle \in \EuScript{S}_p^n}|\langle \mathcal{S}|\psi\rangle|^2$ be the stabiliser state that maximises the stabiliser fidelity $F_{\mathcal{S}}(|\psi\rangle)$. Then
    \begin{align*}
        \sum_{\mathbf{x}\in \mathscr{M}} p_\psi(\mathbf{x}) \leq F_{\mathcal{S}}(|\psi\rangle) \leq \sqrt{\sum_{\mathbf{x}\in \mathscr{M}} p_\psi(\mathbf{x})}.
    \end{align*}
\end{lemma}
\begin{proof}
    For the upper bound,
    \begin{align*}
        F_{\mathcal{S}}(|\psi\rangle) &= |\langle \psi|\mathscr{M},\mathbf{a}\rangle\langle \mathscr{M},\mathbf{a}|\psi\rangle| \\
        &= \frac{1}{p^{n}}\left|\sum_{\mathbf{x}\in\mathscr{M}}\omega^{[\mathbf{a},\mathbf{x}]}\langle \psi|\mathcal{W}_{\mathbf{x}}|\psi\rangle \right|  \tag{\cref{eq:stabiliser_form}}\\
        &\leq \sqrt{\frac{1}{p^n}\sum_{\mathbf{x}\in \mathscr{M}} |\langle \psi|\mathcal{W}_{\mathbf{x}}|\psi\rangle|^2} \tag{Cauchy-Schwarz}\\
        &= \sqrt{\sum_{\mathbf{x}\in \mathscr{M}} p_\psi(\mathbf{x})}.
    \end{align*}
    For the lower bound, recall that $\{|\mathscr{M},\mathcal{C}\rangle\}_{\mathcal{C}\in\mathbb{F}_p^{2n}/\mathscr{M}}$ forms a basis of stabiliser states. Then
    \begin{align*}
        F_{\mathcal{S}}(|\psi\rangle) &= \max_{\mathcal{C}\in\mathbb{F}_p^{2n}/\mathscr{M}}\langle \mathscr{M},\mathcal{C}|\psi |\mathscr{M},\mathcal{C}\rangle \\
        &\geq \sum_{\mathcal{C}\in\mathbb{F}_p^{2n}/\mathscr{M}}\langle \mathscr{M},\mathcal{C}|\psi |\mathscr{M},\mathcal{C}\rangle^2 \tag{$\sum_{\mathcal{C}\in\mathbb{F}_p^{2n}/\mathscr{M}}\langle \mathscr{M},\mathcal{C}|\psi |\mathscr{M},\mathcal{C}\rangle = 1$}\\
        &= \frac{1}{p^{2n}}\sum_{\mathcal{C}\in\mathbb{F}_p^{2n}/\mathscr{M}} \sum_{\mathbf{x},\mathbf{y}\in\mathscr{M}}\omega^{[\mathbf{b}_{\mathcal{C}},\mathbf{x}+\mathbf{y}]}\langle \psi|\mathcal{W}_{\mathbf{x}}|\psi\rangle\langle\psi|\mathcal{W}_{\mathbf{y}}|\psi\rangle,
    \end{align*}
    where $\mathbf{b}_{\mathcal{C}}$ is any element of the coset $\mathcal{C}$. By taking the average over all elements $\mathbf{b}\in\mathcal{C}$, then
    \begin{align*}
        F_{\mathcal{S}}(|\psi\rangle) &\geq \frac{1}{p^{3n}}\sum_{\mathcal{C}\in\mathbb{F}_p^{2n}/\mathscr{M}} \sum_{\mathbf{b}\in\mathcal{C}}\sum_{\mathbf{x},\mathbf{y}\in\mathscr{M}}\omega^{[\mathbf{b},\mathbf{x}+\mathbf{y}]}\langle \psi|\mathcal{W}_{\mathbf{x}}|\psi\rangle\langle\psi|\mathcal{W}_{\mathbf{y}}|\psi\rangle\\
        &= \frac{1}{p^{3n}} \sum_{\mathbf{b}\in\mathbb{F}_p^{2n}}\sum_{\mathbf{x},\mathbf{y}\in\mathscr{M}}\omega^{[\mathbf{b},\mathbf{x}+\mathbf{y}]}\langle \psi|\mathcal{W}_{\mathbf{x}}|\psi\rangle\langle\psi|\mathcal{W}_{\mathbf{y}}|\psi\rangle \\
        &= \frac{1}{p^{n}} \sum_{\mathbf{x}\in\mathscr{M}}\langle \psi|\mathcal{W}_{\mathbf{x}}|\psi\rangle\langle\psi|\mathcal{W}_{-\mathbf{x}}|\psi\rangle \tag{\cref{lem:sum}}\\
        &= \sum_{\mathbf{x}\in \mathscr{M}} p_\psi(\mathbf{x}). \hspace{8cm}\qedhere
    \end{align*}
\end{proof}

\section{Exploring Bell difference sampling}
\label{sec:bell_difference_sampling}

In this section, we generalise and explore Bell (difference) sampling as proposed by~\cite{montanaro2017learning,gross2021schur} to qudits, ultimately uncovering its limitations. We start with defining the generalised Bell states. Let $|\Phi^+\rangle \triangleq p^{-n/2}\sum_{\mathbf{q}\in\mathbb{F}_p^n} |\mathbf{q}\rangle^{\otimes 2}$ be a maximally entangled state. The generalised Bell states over $\mathbb{F}_p^n$ are defined as
\begin{align*}
    |\mathcal{W}_{\mathbf{x}}\rangle \triangleq (\mathcal{W}_{\mathbf{x}}\otimes\mathbb{I}) |\Phi^+\rangle = \frac{1}{\sqrt{p^n}}\sum_{\mathbf{q}\in\mathbb{F}_p^n} \omega^{\langle \mathbf{q}, \mathbf{v}\rangle + 2^{-1}\langle \mathbf{v},\mathbf{w}\rangle}|\mathbf{q}+\mathbf{w}\rangle|\mathbf{q}\rangle, \quad\forall\mathbf{x}=(\mathbf{v},\mathbf{w})\in\mathbb{F}_p^{2n}.
\end{align*}
Notice that $|\mathcal{W}_{\mathbf{x}}\rangle = p^{-n/2}\operatorname{vec}(\mathcal{W}_{\mathbf{x}})$ since $\mathcal{W}_{\mathbf{x}} = \sum_{\mathbf{q}\in\mathbb{F}_p^n} \omega^{\langle \mathbf{q}, \mathbf{v}\rangle + 2^{-1}\langle \mathbf{v},\mathbf{w}\rangle}|\mathbf{q}+\mathbf{w}\rangle\langle\mathbf{q}|$. As previously mentioned, the re-scaled Weyl operators $\{p^{-n/2}\mathcal{W}_{\mathbf{x}}:\mathbf{x}\in\mathbb{F}_p^{2n}\}$ form an orthonormal basis in the space of operators, therefore the states $\{|\mathcal{W}_{\mathbf{x}}\rangle:\mathbf{x}\in\mathbb{F}_p^{2n}\}$ form an orthonormal basis in the Hilbert space $(\mathbb{C}^p)^{\otimes n}\otimes (\mathbb{C}^p)^{\otimes n}$. Indeed, $\langle \mathcal{W}_{\mathbf{x}}|\mathcal{W}_{\mathbf{y}}\rangle = p^{-n}\operatorname{Tr}[\mathcal{W}_{\mathbf{x}}^\dagger \mathcal{W}_{\mathbf{y}}] = \mathbf{1}[\mathbf{x} = \mathbf{y}]$. In addition, for all $\mathbf{x}=(\mathbf{v},\mathbf{w})\in\mathbb{F}_p^{2n}$,
\begin{align*}
    (\mathcal{W}_{\mathbf{x}}\otimes\mathcal{W}_{J(\mathbf{x})}) |\Phi^+\rangle = (\mathcal{W}_{\mathbf{v},\mathbf{w}}\otimes\mathcal{W}_{-\mathbf{v},\mathbf{w}}) |\Phi^+\rangle = \frac{1}{\sqrt{p^n}}\sum_{\mathbf{q}\in\mathbb{F}_p^n} |\mathbf{q}+\mathbf{w}\rangle|\mathbf{q}+\mathbf{w}\rangle
    = |\Phi^+\rangle.
\end{align*}
The following representation of $|\mathcal{W}_{\mathbf{x}}\rangle\langle \mathcal{W}_{\mathbf{x}}|$ will be useful.
\begin{lemma}\label{lem:bell_basis}
    For all $\mathbf{x}\in\mathbb{F}_p^{2n}$,
    \begin{align*}
        |\mathcal{W}_{\mathbf{x}}\rangle\langle \mathcal{W}_{\mathbf{x}}| = \frac{1}{p^{2n}}\sum_{\mathbf{y}\in\mathbb{F}_p^{2n}} \omega^{[\mathbf{x},\mathbf{y}]} \mathcal{W}_{\mathbf{y}}\otimes \mathcal{W}_{J(\mathbf{y})}.
    \end{align*}
\end{lemma}
\begin{proof}
    For all $\mathbf{z}, \mathbf{z'} \in\mathbb{F}_p^{2n}$,
    \begin{align*}
        \sum_{\mathbf{y}\in\mathbb{F}_p^{2n}} \omega^{[\mathbf{x},\mathbf{y}]}\langle\mathcal{W}_{\mathbf{z}}|\mathcal{W}_{\mathbf{y}}\otimes \mathcal{W}_{J(\mathbf{y})}|\mathcal{W}_{\mathbf{z}'}\rangle &= \sum_{\mathbf{y}\in\mathbb{F}_p^{2n}} \omega^{[\mathbf{x},\mathbf{y}]}\langle\Phi^+|(\mathcal{W}_{\mathbf{z}}^\dagger\otimes\mathbb{I})(\mathcal{W}_{\mathbf{y}}\otimes \mathcal{W}_{J(\mathbf{y})})(\mathcal{W}_{\mathbf{z}'}\otimes\mathbb{I})|\Phi^+\rangle\\
        &= \sum_{\mathbf{y}\in\mathbb{F}_p^{2n}} \omega^{[\mathbf{y},\mathbf{z}'-\mathbf{x}]} \langle\Phi^+|(\mathcal{W}_{\mathbf{z}}^\dagger \mathcal{W}_{\mathbf{z}'}\otimes\mathbb{I})(\mathcal{W}_{\mathbf{y}}\otimes \mathcal{W}_{J(\mathbf{y})})|\Phi^+\rangle\\
        &= \sum_{\mathbf{y}\in\mathbb{F}_p^{2n}} \omega^{[\mathbf{y},\mathbf{z}'-\mathbf{x}]} \langle\Phi^+|(\mathcal{W}_{\mathbf{z}}^\dagger \mathcal{W}_{\mathbf{z}'}\otimes\mathbb{I})|\Phi^+\rangle\\
        &= \mathbf{1}[\mathbf{z} = \mathbf{z}']\sum_{\mathbf{y}\in\mathbb{F}_p^{2n}} \omega^{[\mathbf{y},\mathbf{z}'-\mathbf{x}]}\\
        &= p^{2n}\mathbf{1}[\mathbf{z} = \mathbf{z}' = \mathbf{x}],
    \end{align*}
    using \cref{lem:sum} together with $(\mathbb{F}_p^{2n})^{\independent} = \{\mathbf{0}\}$ (the symplectic product is non-degenerate).
\end{proof}

Given a pure state of $2n$ qudits divided into systems $A_1,\dots,A_n$ and $B_1,\dots,B_n$, we call \emph{Bell sampling} the operation of measuring each pair $A_i B_i$ of qudits in the generalised Bell basis, which returns a vector in $\mathbb{F}_p^{2n}$. A similar operation, called \emph{Bell difference sampling}, performs Bell sampling twice and subtracts the results from each other.
\begin{definition}[Bell sampling]
    Bell sampling is the projective measurement given by $ |\mathcal{W}_{\mathbf{x}}\rangle\langle \mathcal{W}_{\mathbf{x}}|$ for $\mathbf{x}\in\mathbb{F}_p^{2n}$.
\end{definition}

\begin{definition}[Bell difference sampling]\label{def:Bell_difference_sampling}
    Bell difference sampling is defined as performing Bell sampling twice and subtracting the results from each other (modulo $p$). More precisely, it is the projective measurement given by
    \begin{align*}
        \Pi_{\mathbf{x}} = \sum_{\mathbf{y}\in\mathbb{F}_p^{2n}}|\mathcal{W}_{\mathbf{y}}\rangle\langle \mathcal{W}_{\mathbf{y}}|\otimes |\mathcal{W}_{\mathbf{x}+\mathbf{y}}\rangle\langle \mathcal{W}_{\mathbf{x}+\mathbf{y}}|, \quad \forall\mathbf{x}\in\mathbb{F}_p^{2n}.
    \end{align*}
\end{definition}
While Bell (difference) sampling was originally considered for copies of stabiliser states in $(\mathbb{C}^2)^{\otimes n}$~\cite{montanaro2017learning}, Gross, Nezami, and Walter~\cite{gross2021schur} proved that it is a meaningful procedure for copies of any pure state $|\psi\rangle \in (\mathbb{C}^2)^{\otimes n}$ and corresponds to sampling from the convolution of the characteristic distribution $p_\psi$ with itself (up to normalisation). Here we extend their results for $p>2$ prime. 
\begin{lemma}\label{lem:lem2}
    Given two pure states $|\psi_1\rangle,|\psi_2\rangle\in(\mathbb{C}^{p})^{\otimes n}$, Bell sampling on $|\psi_1\rangle|\psi_2\rangle$ returns $\mathbf{x}\in\mathbb{F}_p^{2n}$ with probability $p^{-n}|\langle \psi_1|\mathcal{W}_{\mathbf{x}}|\psi_2^\ast\rangle|^2$. In particular, if $|\psi_1\rangle = |\psi_2^\ast\rangle = |\psi\rangle$, Bell sampling on $|\psi\rangle|\psi^\ast\rangle$ corresponds to sampling from the distribution $p_\psi(\mathbf{x})$.
\end{lemma}
\begin{proof}
    We have $|\langle \mathcal{W}_{\mathbf{x}} |(|\psi_1\rangle|\psi_2\rangle)|^2 = p^{-n} |{\operatorname{Tr}}(\mathcal{W}_{\mathbf{x}}^\dagger |\psi_1\rangle\langle \psi_2^\ast| )|^2 = p^{-n}|\langle \psi_1|\mathcal{W}_{\mathbf{x}}|\psi_2^\ast\rangle|^2$, since $|\psi_1\rangle|\psi_2\rangle = \operatorname{vec}(|\psi_1\rangle\langle\psi_2^\ast|)$.
\end{proof}

\begin{theorem}\label{thr:bell_difference}
    Let $|\psi\rangle\in(\mathbb{C}^p)^{\otimes n}$ be a pure state. Bell difference sampling on $|\psi\rangle^{\otimes 4}$ corresponds to sampling from the distribution
    \begin{align*}
        b_\psi(\mathbf{x}) \triangleq \operatorname{Tr}[\Pi_{\mathbf{x}}\psi^{\otimes 4}] = p^{2n}(p_\psi\ast j_\psi)(\mathbf{x}) = \sum_{\mathbf{y}\in\mathbb{F}_p^{2n}} p_\psi(\mathbf{y})j_\psi(\mathbf{x} - \mathbf{y}).
    \end{align*}
\end{theorem}
We call the distribution $b_\psi(\mathbf{x})$ the \emph{involuted Weyl distribution} of the pure state $|\psi\rangle$.
\begin{proof}
    First note that, according to \cref{lem:bell_basis},
    \begin{align*}
        \Pi_{\mathbf{x}} &= \frac{1}{p^{4n}}\sum_{\mathbf{y},\mathbf{z},\mathbf{z}'\in\mathbb{F}_p^{2n}} \omega^{[\mathbf{y},\mathbf{z}']}\omega^{[\mathbf{x}+\mathbf{y},\mathbf{z}]} \mathcal{W}_{\mathbf{z}'}\otimes \mathcal{W}_{J(\mathbf{z}')}\otimes \mathcal{W}_{\mathbf{z}}\otimes \mathcal{W}_{J(\mathbf{z})}\\
        &= \frac{1}{p^{2n}} \sum_{\mathbf{z},\mathbf{z}'\in\mathbb{F}_p^{2n}}\mathbf{1}[\mathbf{z}' = -\mathbf{z}] \omega^{[\mathbf{x},\mathbf{z}]} \mathcal{W}_{\mathbf{z}'}\otimes \mathcal{W}_{J(\mathbf{z}')}\otimes \mathcal{W}_{\mathbf{z}}\otimes \mathcal{W}_{J(\mathbf{z})} \\
        &= \frac{1}{p^{2n}}\sum_{\mathbf{z}\in\mathbb{F}_p^{2n}}\omega^{[\mathbf{x},\mathbf{z}]}\mathcal{W}_{\mathbf{z}}^\dagger\otimes \mathcal{W}_{J(\mathbf{z})}^\dagger \otimes \mathcal{W}_{\mathbf{z}} \otimes \mathcal{W}_{J(\mathbf{z})}.
    \end{align*}
    Therefore
    \begin{align*}
        \operatorname{Tr}[\Pi_{\mathbf{x}}\psi^{\otimes 4}] &= \frac{1}{p^{2n}}\sum_{\mathbf{y}\in\mathbb{F}_p^{2n}} \omega^{[\mathbf{x},\mathbf{y}]}\operatorname{Tr}[\mathcal{W}_{\mathbf{y}}^\dagger\psi\otimes \mathcal{W}_{J(\mathbf{y})}^\dagger \psi\otimes \mathcal{W}_{\mathbf{y}}\psi \otimes \mathcal{W}_{J(\mathbf{y})}\psi] \\
        &= \sum_{\mathbf{y}\in\mathbb{F}_p^{2n}} \omega^{[\mathbf{x},\mathbf{y}]} p_\psi(\mathbf{y})p_\psi(J(\mathbf{y}))\\
        &= p^{2n}\sum_{\mathbf{y}\in\mathbb{F}_p^{2n}}  \widehat{p_\psi}(\mathbf{x}-\mathbf{y}) \widehat{j_\psi}(\mathbf{y}) \tag{Parseval's identity, \cref{lem:parseval}}\\
        &= \sum_{\mathbf{y}\in\mathbb{F}_p^{2n}}  p_\psi(\mathbf{x}-\mathbf{y}) j_\psi(\mathbf{y}) \tag{\cref{lem:invariant}}\\
        &= p^{2n} (p_\psi\ast j_\psi)(\mathbf{x}). \hspace{8cm}\qedhere
    \end{align*}
\end{proof}

In the case of qubits ($p=2$), the involution is trivial, i.e., $J(\mathbf{x}) = \mathbf{x}$ for all $\mathbf{x}\in\mathbb{F}_p^{2n}$. Thus $j_\psi = p_\psi$ and Bell difference sampling corresponds to sampling from the distribution $4^n (p_\psi\ast p_\psi)$. \cref{thr:bell_difference} is a clear generalisation of this fact and starts to show to role of the involution on higher dimensions.

We now explore the distributions $p_\psi$ and $b_\psi$ in more details. More specifically, we start by showing the following identities regarding the mass on a subspace $\mathscr{X}\subseteq\mathbb{F}_2^{2n}$ under $p_\psi$ or $b_\psi$, which are simple generalisations of~\cite[Theorem~3.1 and Theorem~3.2]{grewal2023improved} for $p>2$ prime.
\begin{lemma}\label{lem:properties_chateristic_Weyl}
    Let $\mathscr{X}\subseteq\mathbb{F}_p^{2n}$ be a subspace. Then
    \begin{align*}
        \sum_{\mathbf{x}\in\mathscr{X}}p_\psi(\mathbf{x}) = \frac{|\mathscr{X}|}{p^n}\sum_{\mathbf{y}\in\mathscr{X}^{\independent}} p_\psi(\mathbf{y}) \qquad\quad\text{and}\qquad\quad \frac{1}{|\mathscr{X}|}\sum_{\mathbf{x}\in\mathscr{X}}b_\psi(\mathbf{x}) = \sum_{\mathbf{y}\in\mathscr{X}^{\independent}} p_\psi(\mathbf{y})j_\psi(\mathbf{y}).
    \end{align*}
\end{lemma}
\begin{proof}
    For the second identity,
    \begin{align*}
        \sum_{\mathbf{x}\in\mathscr{X}} b_\psi(\mathbf{x}) &= \sum_{\mathbf{x}\in\mathscr{X}}\sum_{\mathbf{y}\in\mathbb{F}_p^{2n}} \omega^{[\mathbf{x},\mathbf{y}]} \widehat{b_\psi}(\mathbf{y}) \\
        &= p^{2n} \sum_{\mathbf{x}\in\mathscr{X}}\sum_{\mathbf{y}\in\mathbb{F}_p^{2n}} \omega^{[\mathbf{x},\mathbf{y}]} \widehat{p_\psi}(\mathbf{y}) \widehat{j_\psi}(\mathbf{y}) \tag{\cref{thr:bell_difference} and \cref{lem:convolution_Fourier}}\\
        &= \sum_{\mathbf{x}\in\mathscr{X}}\sum_{\mathbf{y}\in\mathbb{F}_p^{2n}} \omega^{[\mathbf{x},\mathbf{y}]} p_\psi(\mathbf{y}) j_\psi(\mathbf{y}) \tag{\cref{lem:invariant}}\\
        &= |\mathscr{X}| \sum_{\mathbf{y}\in\mathscr{X}^{\independent}} p_\psi(\mathbf{y}) j_\psi(\mathbf{y}). \tag{\cref{lem:sum}} 
    \end{align*}
    Regarding the first identity, we similarly have
    \begin{align*}
        \sum_{\mathbf{x}\in\mathscr{X}}p_\psi(\mathbf{x}) &= \sum_{\mathbf{x}\in\mathscr{X}}\sum_{\mathbf{y}\in\mathbb{F}_p^{2n}} \omega^{[\mathbf{x},\mathbf{y}]} \widehat{p_\psi}(\mathbf{y}) 
        = \frac{1}{p^n}\sum_{\mathbf{x}\in\mathscr{X}}\sum_{\mathbf{y}\in\mathbb{F}_p^{2n}} \omega^{[\mathbf{x},\mathbf{y}]} p_\psi(\mathbf{y}) 
        = \frac{|\mathscr{X}|}{p^n}\sum_{\mathbf{y}\in\mathscr{X}^{\independent}} p_\psi(\mathbf{y}), 
    \end{align*}
    where we used \cref{lem:invariant} and \cref{lem:sum}.
\end{proof}

\cref{lem:properties_chateristic_Weyl} allows one to further explore the characteristic and Weyl distributions, e.g.\ we can generalise \cref{lem:stabiliser_fidelity_inequality}.
\begin{lemma}\label{lem:stabiliser_fidelity_generalisation}
    Given a quantum state $|\psi\rangle\in(\mathbb{C}^p)^{\otimes n}$, let $|\mathscr{M},\mathbf{a}\rangle = \argmax_{|\mathcal{S}\rangle \in \EuScript{S}_p^n}|\langle \mathcal{S}|\psi\rangle|^2$ be the stabiliser state that maximises the stabiliser fidelity $F_{\mathcal{S}}(|\psi\rangle)$. Let $\mathscr{X},\mathscr{Y}$ be subspaces such that $\mathscr{X}\subseteq \mathscr{M}\subseteq\mathscr{Y}$. Then
    \begin{align*}
        \frac{p^n}{|\mathscr{Y}|}\sum_{\mathbf{x}\in \mathscr{Y}} p_\psi(\mathbf{x}) \leq F_{\mathcal{S}}(|\psi\rangle) \leq \sqrt{\frac{p^n}{|\mathscr{X}|}\sum_{\mathbf{x}\in \mathscr{X}} p_\psi(\mathbf{x})}.
    \end{align*}    
\end{lemma}
\begin{proof}
    Using \cref{lem:stabiliser_fidelity_inequality} and \cref{lem:properties_chateristic_Weyl},
    \begin{align*}
        \frac{p^n}{|\mathscr{Y}|}\sum_{\mathbf{x}\in \mathscr{Y}} p_\psi(\mathbf{x}) &= \sum_{\mathbf{x}\in \mathscr{Y}^{\independent}} p_\psi(\mathbf{x}) \leq \sum_{\mathbf{x}\in \mathscr{M}} p_\psi(\mathbf{x}) \leq F_{\mathcal{S}}(|\psi\rangle),\\
        \frac{p^n}{|\mathscr{X}|}\sum_{\mathbf{x}\in \mathscr{X}} p_\psi(\mathbf{x}) &= \sum_{\mathbf{x}\in \mathscr{X}^{\independent}} p_\psi(\mathbf{x}) \geq \sum_{\mathbf{x}\in \mathscr{M}} p_\psi(\mathbf{x}) \geq F_{\mathcal{S}}(|\psi\rangle)^2. \qedhere
    \end{align*}
\end{proof}

As another consequence of \cref{lem:properties_chateristic_Weyl}, the support of $p_\psi$ and $b_\psi$ are related to $\operatorname{Weyl}(|\psi\rangle)$, which is a generalisation of~\cite[Lemma~4.3 and Corollary~4.4]{grewal2023improved}.
\begin{lemma}\label{lem:support}
    Let $|\psi\rangle\in(\mathbb{C}^p)^{\otimes n}$ and let its characteristic and Weyl distributions $p_\psi$ and $b_\psi$, respectively. The support of $p_\psi$ is contained in $\operatorname{Weyl}(|\psi\rangle)^{\independent}$, while the support of $b_\psi$ is contained in $\operatorname{Weyl}(|\psi\rangle)^{\independent} + J(\operatorname{Weyl}(|\psi\rangle)^{\independent})$.
\end{lemma}
\begin{proof}
    We show that the mass of $p_\psi$ on $\operatorname{Weyl}(|\psi\rangle)^{\independent}$ equals $1$.
    \begin{align*}
        \sum_{\mathbf{x}\in\operatorname{Weyl}(|\psi\rangle)^{\independent}} p_\psi(\mathbf{x}) &= \frac{|\operatorname{Weyl}(|\psi\rangle)^{\independent}|}{p^n}\sum_{\mathbf{y}\in \operatorname{Weyl}(|\psi\rangle)} p_\psi(\mathbf{y}) \tag{\cref{lem:properties_chateristic_Weyl}}\\
        &= \frac{|\operatorname{Weyl}(|\psi\rangle)^{\independent}|}{p^n}\frac{|\operatorname{Weyl}(|\psi\rangle)|}{p^n} \tag{$p_\psi(\mathbf{x}) = p^{-n}$ iff $\mathbf{x}\in\operatorname{Weyl}(|\psi\rangle)$} = 1.
    \end{align*}
    By the same token, the involution characteristic distribution $j_\psi(\mathbf{x}) = p_\psi(J(\mathbf{x}))$ is supported on $J(\operatorname{Weyl}(|\psi\rangle)^{\independent})$ since $J(\operatorname{Weyl}(|\psi\rangle))^{\independent} = J(\operatorname{Weyl}(|\psi\rangle)^{\independent})$ and $|J(\operatorname{Weyl}(|\psi\rangle))| = |\operatorname{Weyl}(|\psi\rangle)|$. On the other hand, according to \cref{thr:bell_difference},
    \begin{align*}
        b_\psi(\mathbf{x}) = \sum_{\mathbf{y}\in\mathbb{F}_p^{2n}} p_\psi(\mathbf{y})j_\psi(\mathbf{x} - \mathbf{y}) &= \sum_{\mathbf{y}\in\operatorname{Weyl}(|\psi\rangle)^{\independent}} p_\psi(\mathbf{y}) j_\psi(\mathbf{x} - \mathbf{y}).
    \end{align*}
    %
    For $\mathbf{y}\in\operatorname{Weyl}(|\psi\rangle)^{\independent}$, $\mathbf{x}\notin \operatorname{Weyl}(|\psi\rangle)^{\independent} + J(\operatorname{Weyl}(|\psi\rangle)^{\independent}) \implies \mathbf{x} - \mathbf{y}\notin J(\operatorname{Weyl}(|\psi\rangle)^{\independent})$. This means that, if $\mathbf{x}\notin \operatorname{Weyl}(|\psi\rangle)^{\independent} + J(\operatorname{Weyl}(|\psi\rangle)^{\independent})$ and $\mathbf{y}\in\operatorname{Weyl}(|\psi\rangle)^{\independent}$, then $j_\psi(\mathbf{x} - \mathbf{y}) = 0$, and so $b_\psi(\mathbf{x}) = 0$. Therefore, the support of $b_\psi$ is contained in $\operatorname{Weyl}(|\psi\rangle)^{\independent} + J(\operatorname{Weyl}(|\psi\rangle)^{\independent})$.
\end{proof}

\cref{lem:support} is quite powerful. As we have previously seen, the characteristic function of a stabiliser state $|\mathcal{S}\rangle$ is uniformly distributed on its associated Lagrangian subspace $\mathscr{M}$, i.e., $p_{\mathcal{S}}(\mathbf{x}) = p^{-n}$ if $\mathbf{x}\in\mathscr{M}$ and $0$ otherwise. This fact also follows from \cref{lem:support} since $p_{\mathcal{S}}$ is supported on $\operatorname{Weyl}(|\mathcal{S}\rangle)^{\independent} = \mathscr{M}^{\independent} = \mathscr{M}$ and $p_{\mathcal{S}}(\mathbf{x}) = p^{-n}$ for $\mathbf{x}\in\operatorname{Weyl}(|\mathcal{S}\rangle) = \mathscr{M}$. \cref{lem:support} also implies that $b_{\mathcal{S}}$ is supported on $\operatorname{Weyl}(|\mathcal{S}\rangle)^{\independent} + J(\operatorname{Weyl}(|\mathcal{S}\rangle)^{\independent}) = \mathscr{M} + J(\mathscr{M})$, where $J(\mathscr{M}) \triangleq \{J(\mathbf{x})\in\mathbb{F}_p^{2n}:\mathbf{x}\in\mathscr{M}\}$. This is formalised in the next result, which is the main result of this section.
\begin{theorem}\label{lem:bell_sampling}
    Let $\mathcal{S} = \{\omega^{[\mathbf{a},\mathbf{x}]}\mathcal{W}_{\mathbf{x}}:\mathbf{x}\in\mathscr{M}\}$ be a stabiliser group with stabiliser state $|\mathcal{S}\rangle\in(\mathbb{C}^p)^{\otimes n}$. Let $\mathbf{V},\mathbf{W}\in\mathbb{F}_p^{n\times n}$ be matrices such that $\mathscr{M} = {\operatorname{col}}\big(\big[\begin{smallmatrix} \mathbf{V} \\ \mathbf{W} \end{smallmatrix}\big]\big)$. Bell difference sampling on $|\mathcal{S}\rangle^{\otimes 4}$
    returns $\mathbf{x}\in\mathbb{F}_p^{2n}$ with probability
    \begin{align*}
        b_{\mathcal{S}}(\mathbf{x}) &= \begin{cases}
            p^{-2n}|\mathscr{M}\cap J(\mathscr{M})| &\quad\quad\text{if}~\mathbf{x}\in \mathscr{M} + J(\mathscr{M}),\\
            0 &\quad\quad\text{otherwise},
        \end{cases} \\
        &= \begin{cases}
            |\operatorname{col}(\mathbf{V})|^{-1}|\operatorname{col}(\mathbf{W})|^{-1} &\text{if}~\mathbf{x}\in\operatorname{col}(\mathbf{V})\times\operatorname{col}(\mathbf{W}),\\
            0 &\text{otherwise}.
            \end{cases}
    \end{align*}
\end{theorem}
\begin{proof}
    According to \cref{thr:bell_difference} and \cref{eq:stabiliser_distribution}, Bell difference sampling on $|\mathcal{S}\rangle^{\otimes 4}$ returns $\mathbf{x}\in\mathbb{F}_p^{2n}$ with probability
    \begin{align*}
        b_{\mathcal{S}}(\mathbf{x}) = \sum_{\mathbf{y}\in\mathbb{F}_p^{2n}} p_{\mathcal{S}}(\mathbf{y})j_{\mathcal{S}}(\mathbf{x} - \mathbf{y}) &= \frac{1}{p^n}\sum_{\mathbf{y}\in\mathscr{M}}p_{\mathcal{S}}(J(\mathbf{x} - \mathbf{y})).
    \end{align*}
    Since $p_{\mathcal{S}}(J(\mathbf{x} - \mathbf{y})) = p^{-n}$ if $J(\mathbf{x}-\mathbf{y})\in\mathscr{M} \iff \mathbf{x} \in \mathbf{y}+J(\mathscr{M})$ and $\mathscr{M} + J(\mathscr{M})$ is a subspace, we conclude that $p_{\mathcal{S}}(J(\mathbf{x} - \mathbf{y})) = 0$ if $\mathbf{x}\notin \mathscr{M} + J(\mathscr{M})$. Let us compute then $b_{\mathcal{S}}(J(\mathbf{x}))$ for $\mathbf{x}\in \mathscr{M} + J(\mathscr{M})$. Write $\mathbf{x} = \mathbf{x}_1 + \mathbf{x}_2$, where $\mathbf{x}_1 \in \mathscr{M}$ and $\mathbf{x}_2\in J(\mathscr{M})$. Then
    \begin{align*}
        b_{\mathcal{S}}(J(\mathbf{x})) = \frac{1}{p^n}\sum_{\mathbf{y}\in J(\mathscr{M})}p_{\mathcal{S}}(\mathbf{x}_1 + \mathbf{x}_2 - \mathbf{y}) \overset{\mathbf{y} = \mathbf{y} + \mathbf{x}_2}{=} \frac{1}{p^n}\sum_{\mathbf{y}\in J(\mathscr{M})}p_{\mathcal{S}}(\mathbf{x}_1 - \mathbf{y}).
    \end{align*}
    If $\mathbf{y}\notin \mathscr{M}$, then $\mathbf{x}_1 - \mathbf{y}\notin \mathscr{M}$ and so $p_{\mathcal{S}}(\mathbf{x}_1 - \mathbf{y}) = 0$. Therefore, $p_{\mathcal{S}}(\mathbf{x}_1 - \mathbf{y}) = p^{-n}$ if $\mathbf{y}\in \mathscr{M}\cap J(\mathscr{M})$ and $0$ otherwise. All in all, we conclude that 
    \begin{align*}
        b_{\mathcal{S}}(\mathbf{x}) = \begin{cases}
            p^{-2n}|\mathscr{M}\cap J(\mathscr{M})| &\text{if}~\mathbf{x}\in \mathscr{M} + J(\mathscr{M}),\\
            0 &\text{otherwise}.
        \end{cases}
    \end{align*}
    Now, notice that  $J(\mathscr{M}) = J\big({\operatorname{col}}\big(\big[\begin{smallmatrix} \mathbf{V} \\ \mathbf{W} \end{smallmatrix}\big]\big)\big) ={\operatorname{col}}\big(J\big(\big[\begin{smallmatrix} \mathbf{V} \\ \mathbf{W} \end{smallmatrix}\big]\big)\big) = {\operatorname{col}}\big(\big[\begin{smallmatrix} -\mathbf{V} \\ \mathbf{W} \end{smallmatrix}\big]\big)$. Then 
    \begin{align*}
        \mathscr{M} + J(\mathscr{M}) = {\operatorname{col}}\big(\big[\begin{smallmatrix} \mathbf{V} \\ \mathbf{W} \end{smallmatrix}\big]\big) + {\operatorname{col}}\big(\big[\begin{smallmatrix} -\mathbf{V} \\ \mathbf{W} \end{smallmatrix}\big]\big) = \operatorname{col}(\mathbf{V})\times\operatorname{col}(\mathbf{W}),
    \end{align*}
    since $\big[\begin{smallmatrix} \mathbf{V} \\ \mathbf{W} \end{smallmatrix}\big]\frac{\mathbf{z}_1 + \mathbf{z}_2}{2} + \big[\begin{smallmatrix} -\mathbf{V} \\ \mathbf{W} \end{smallmatrix}\big]\frac{\mathbf{z}_1 - \mathbf{z}_2}{2} = \big[\begin{smallmatrix} \mathbf{V}\mathbf{z}_2 \\ \mathbf{W}\mathbf{z}_1 \end{smallmatrix}\big]$ for any $\mathbf{z}_1,\mathbf{z}_2\in\mathbb{F}_p^n$. Finally, $\mathscr{M} \cap J(\mathscr{M}) = (\mathscr{M} + J(\mathscr{M}))^{\independent}$ implies $|\mathscr{M} \cap J(\mathscr{M})| = p^{2n}/|\mathscr{M} + J(\mathscr{M})| = p^{2n}/|\operatorname{col}(\mathbf{V})\times\operatorname{col}(\mathbf{W})| = p^{2n}/(|\operatorname{col}(\mathbf{V})||\operatorname{col}(\mathbf{W})|)$.
\end{proof}

The above result reveals highlights the limitations of Bell difference sampling on qudits, which stems from the involution function. More precisely, it shows that Bell difference sampling on $|\mathcal{S}\rangle^{\otimes 4}$ allows one to learn $\mathbf{V}$ and $\mathbf{W}$ separately up to some change of basis, but not $\big[\begin{smallmatrix}\mathbf{V}\\\mathbf{W}\end{smallmatrix}\big]$ and thus $\mathscr{M}$. Equivalently, we only learn the subspace $\mathscr{M} + J(\mathscr{M})$. This is in stark contrast to the case of qubits ($p=2$), where Bell difference sampling on $|\mathcal{S}\rangle^{\otimes 4}$ returns $\mathbf{x}\in\mathscr{M}$ with probability $b_{\mathcal{S}}(\mathbf{x}) = q_{\mathcal{S}}(\mathbf{x}) = 2^{-n}$~\cite{gross2021schur}. Indeed, $4^n (p_\mathcal{S}\ast p_\mathcal{S}) = p_{\mathcal{S}}$ due to \cref{thr:bell_difference} and \cref{eq:stabiliser_distribution}.

\section{Learning stabiliser states}
\label{sec:learning_stabiliser_states}

After exploring the limitations of Bell difference sampling, we propose alternative quantum algorithms to tackle important problems regarding stabiliser states. In this section, we look at the problem of identifying an unknown stabiliser state $|\mathcal{S}\rangle\in(\mathbb{C}^p)^{\otimes n}$ given access to copies of $|\mathcal{S}\rangle$. We start in \cref{sec:method_1} with the simple observation that, if one also has access to the complex conjugate $|\mathcal{S}^\ast\rangle$ of $|\mathcal{S}\rangle$, then it is possible to apply Bell sampling, since, according to \cref{lem:lem2}, Bell sampling on $|\mathcal{S}\rangle|\mathcal{S}^\ast\rangle$ is equivalent to sampling from the characteristic distribution $p_{\mathcal{S}}$. In \cref{sec:method_2}, we develop our alternative quantum algorithm that only uses  copies of $|\mathcal{S}\rangle$.

\subsection{Method 1: Bell sampling}
\label{sec:method_1}

Consider a stabiliser group $\mathcal{S} = \{\omega^{[\mathbf{a},\mathbf{x}]}\mathcal{W}_{\mathbf{x}}:\mathbf{x}\in \mathscr{M}\}$. According to \cref{eq:stabiliser_distribution}, Bell sampling on $|\mathcal{S}\rangle|\mathcal{S}^\ast\rangle$ returns $\mathbf{x} \in \mathscr{M}$. Therefore, it is possible to learn the stabiliser state of a stabiliser group by several applications of Bell sampling on $|\mathcal{S}\rangle|\mathcal{S}^\ast\rangle$, as shown next.

\begin{algorithm}[h]
\caption{Learning $\mathcal{S}$ via Bell sampling}
\DontPrintSemicolon
\label{alg:method1}

    \KwIn{$3n$ copies of a stabiliser state $|\mathcal{S}\rangle\in(\mathbb{C}^p)^{\otimes n}$ and $2n$ copies of its complex conjugate $|\mathcal{S}^\ast\rangle$.}
    \KwOut{a succinct description of the stabiliser group $\mathcal{S}$, with probability $\geq 1 - p^{-n}$.}

   \For{$i\gets 1$ \KwTo $2n$}
   {    
        Bell sample on $|\mathcal{S}\rangle|\mathcal{S}^\ast\rangle$ to obtain $\mathbf{x}_i\in\mathbb{F}_p^{2n}$.
    }

    Determine a basis $\mathscr{B}$ for $\operatorname{span}(\{\mathbf{x}_i\}_{i\in[2n]})$.

    For each $\mathbf{x}\in \mathscr{B}$, measure a copy of $|\mathcal{S}\rangle$ in the eigenbasis of $\mathcal{W}_{\mathbf{x}}$ to determine $s(\mathbf{x})\in\mathbb{F}_p$ such that $\mathcal{W}_{\mathbf{x}}|\mathcal{S}\rangle = \omega^{-s(\mathbf{x})}|\mathcal{S}\rangle$.

    \Return $\mathcal{S} = \langle \{\omega^{s(\mathbf{x})}\mathcal{W}_{\mathbf{x}}\}_{\mathbf{x}\in\mathscr{B}}\rangle$.
\end{algorithm}

\begin{theorem}\label{thr:algorithm_bell_sampling}
    Let $\mathcal{S} = \{\omega^{[\mathbf{a},\mathbf{x}]}\mathcal{W}_{\mathbf{x}}:\mathbf{x}\in \mathscr{M}\}$ be an unknown stabiliser group with stabiliser state $|\mathcal{S}\rangle\in(\mathbb{C}^{p})^{\otimes n}$. There is a quantum algorithm {\rm (\cref{alg:method1})} that identifies $\mathcal{S}$ using $3n$ copies of $|\mathcal{S}\rangle$ and $2n$ copies of $|\mathcal{S}^\ast\rangle$. Its runtime is $O(n^3)$ and failure probability is at most $p^{-n}$.
\end{theorem}
\begin{proof}
    According to \cref{lem:lem2}, Bell sampling on $|\mathcal{S}\rangle|\mathcal{S}^\ast\rangle$ returns $\mathbf{x}\in\mathbb{F}_p^{2n}$ with probability $p_{\mathcal{S}}(\mathbf{x})$, i.e., $\mathbf{x}$ is uniformly distributed on the subspace $\mathscr{M}$. Thus, Bell sample $|\mathcal{S}\rangle|\mathcal{S}^\ast\rangle$ a number of $2n$ times to obtain $\mathbf{x}_1,\dots,\mathbf{x}_{2n}\in \mathscr{M}$. If $\operatorname{dim}(\operatorname{span}(\{\mathbf{x}_i\}_{i\in[2n]})) = n$, then any basis $\mathscr{B}$ of $\operatorname{span}(\{\mathbf{x}_i\}_{i\in[2n]})$ is a basis for $\mathscr{M}$, which can be identified by using Gaussian elimination in time $O(n^3)$. Finally, for each element $\mathbf{x}\in \mathscr{B}$, measure a copy of $|\mathcal{S}\rangle$ in the eigenbasis of $\mathcal{W}_{\mathbf{x}}$ to determine $s(\mathbf{x})\in\mathbb{F}_p$ such that $\mathcal{W}_{\mathbf{x}}|\mathcal{S}\rangle = \omega^{-s(\mathbf{x})}|\mathcal{S}\rangle$, which requires $O(n^2)$ time. The algorithm fails if the $2n$ samples $\mathbf{x}_1,\dots,\mathbf{x}_{2n}$ are contained in a subspace of dimension $n-1$. The probability that the samples are all contained in any such subspace is $p^{-2n}$. By a union bound over all $(p^n-1)/(p-1)$ subspaces of dimension $n-1$, the failure probability is at most $p^{-n}$.
\end{proof}

\subsection{Method 2: learning quadratic functions}
\label{sec:method_2}

In many situations, one might not have access to copies of $|\mathcal{S}^\ast\rangle$. This raises the question of whether it is possible learn a stabiliser state $|\mathcal{S}\rangle$ using only copies of $|\mathcal{S}\rangle$. It is well known that this is possible when $p=2$~\cite{montanaro2017learning}, since $|\mathcal{S}^\ast\rangle = \sigma|\mathcal{S}\rangle$ for some $\sigma\in\mathscr{P}^n_p$. Indeed, when $p=2$, Bell difference sampling on $|\mathcal{S}\rangle^{\otimes 4}$ returns $\mathbf{x}\in \mathscr{M}$. However, for $p>2$, \cref{lem:bell_sampling} shows that Bell difference sampling on $|\mathcal{S}\rangle^{\otimes 4}$ cannot be used to learn $|\mathcal{S}\rangle$. Nonetheless, we now provide an algorithm, based on a completely different technique, that learns $\mathcal{S}$ using only copies of $|\mathcal{S}\rangle$.

\begin{algorithm}[h]
\caption{Learning $\mathcal{S}$ by learning quadratic functions}
\DontPrintSemicolon
\label{alg:method2}

    \KwIn{$3(m + n) + 4$ copies of a stabiliser state $|\mathcal{S}\rangle\in(\mathbb{C}^p)^{\otimes n}$, where $m \triangleq 2n + \lceil\log_p{r}\rceil$, and $\delta_1, \delta_2, \delta_3\in\mathbb{F}_p$, not all $0$, such that $\delta_1^2 + \delta_2^2 + \delta_3^2 = 0$.}
    \KwOut{a succinct description of the stabiliser group $\mathcal{S}$, with probability $\geq 1 - 2p^{-n}$.}

    \For{$i\gets 0$ \KwTo $2n$}
    {    
        Measure $|\mathcal{S}\rangle$ to obtain $\mathbf{b}_i\in\mathbb{F}_p^n$.

        \If{$i\neq 0$}
        {
            $\mathbf{b}_i \gets \mathbf{b}_i-\mathbf{b}_0$.
        }
    }
    
    $r \gets \operatorname{dim}(\operatorname{span}(\{\mathbf{b}_i\}_{i\in[2n]}))$.
    
    Take $\mathbf{w}_i = \mathbf{e}_i$ for $i\in[r]$ and $\mathbf{w}_i = \mathbf{0}$ for $i\in\{r+1,\dots,n\}$.

    Take $\mathbf{v}_i = \mathbf{e}_i$ for $i\in\{r+1,\dots,n\}$.

    \For{$\ell\gets 0$ \KwTo $m$}
    {    
        Use $3$ copies of $|\mathcal{S}\rangle$ to create the state $p^{-n}\sqrt{|\operatorname{col}(\mathbf{W})|}\sum_{\mathbf{t}\in\mathbb{F}_p^n}|\mathbf{W}\mathbf{t}\rangle|\mathcal{S}\rangle^{\otimes 3}$. 

        For $i\in [3]$, apply $U_{\delta_i}:|\mathbf{t},\mathbf{q}\rangle \mapsto |\mathbf{t},\mathbf{q} - \delta_i \mathbf{t}\rangle$, $\forall \mathbf{t},\mathbf{q}\in\mathbb{F}_p^n$, onto the $i$-th copy of $|\mathcal{S}\rangle$ controlled on the register $|\mathbf{W}\mathbf{t}\rangle$.

        Apply an inverse quantum Fourier transform onto the register $|\mathbf{W}\mathbf{t}\rangle$.

        Measure the $4$ registers to obtain a tuple $(\mathbf{c}^{(\ell)},\mathbf{q}_1^{(\ell)},\mathbf{q}_2^{(\ell)},\mathbf{q}_3^{(\ell)})\in\mathbb{F}_p^{4n}$.

        \If{$\ell\neq 0$}
        {
            $(\mathbf{c}^{(\ell)},\mathbf{q}_1^{(\ell)},\mathbf{q}_2^{(\ell)},\mathbf{q}_3^{(\ell)}) \gets (\mathbf{c}^{(\ell)}-\mathbf{c}^{(0)},\mathbf{q}_1^{(\ell)}-\mathbf{q}_1^{(0)},\mathbf{q}_2^{(\ell)}-\mathbf{q}_2^{(0)},\mathbf{q}_3^{(\ell)}-\mathbf{q}_3^{(0)})$.
        }
    }

    \For{$k\gets 1$ \KwTo $r$}
    {
        Solve the linear system of $m$ equations $\sum_{j=1}^n V_{jk}(\sum_{i=1}^3 \delta_i \mathbf{q}_i^{(\ell)})_j  = (\mathbf{W}^\top \mathbf{c}^{(\ell)})_k$, $\ell\in[m]$, with unknown variables $V_{1k},\dots, V_{nk}$.
    }

    \For{$i\gets 1$ \KwTo $n$}
    {
        Measure a copy of $|\mathcal{S}\rangle$ in the eigenbasis of $\mathcal{W}_{\mathbf{v}_i,\mathbf{w}_i}$ to determine $s_i\in\mathbb{F}_p$ such that $\mathcal{W}_{\mathbf{v}_i,\mathbf{w}_i}|\mathcal{S}\rangle = \omega^{-s_i}|\mathcal{S}\rangle$.
    }

    \Return $\mathcal{S} = \langle \{\omega^{s_i}\mathcal{W}_{\mathbf{v}_i,\mathbf{w}_i}\}_{i\in[n]}\rangle$.
\end{algorithm}

\begin{theorem}\label{thr:algorithm_method_2}
    Let $\mathcal{S}=\langle\{\omega^{s_i}\mathcal{W}_{\mathbf{v}_i,\mathbf{w}_i}\}_{i\in[n]}\rangle$ be an unknown stabiliser group with stabiliser state $|\mathcal{S}\rangle\in(\mathbb{C}^{p})^{\otimes n}$. Let $\mathbf{W}\in\mathbb{F}_p^{n\times n}$ be the matrix with column vectors $\mathbf{w}_1,\dots,\mathbf{w}_n$. There is a quantum algorithm {\rm (\cref{alg:method2})} that identifies $\mathcal{S}$ by using $9n + 3\lceil\log_p\operatorname{rank}(\mathbf{W})\rceil + 4$ copies of $|\mathcal{S}\rangle$. Its runtime is $O(n^3\operatorname{rank}(\mathbf{W}))$ and failure probability is at most $2p^{-n}$.
\end{theorem}
\begin{proof}
    By \cref{lem:normalisation}, the stabiliser state of $\mathcal{S}=\langle\{\omega^{s_i}\mathcal{W}_{\mathbf{v}_i,\mathbf{w}_i}\}_{i\in[n]}\rangle$ is given by
    \begin{align*}
        |\mathcal{S}\rangle = \frac{\sqrt{|\operatorname{col}(\mathbf{W})|}}{p^n}\sum_{\mathbf{q}\in \mathbb{F}_p^n} \omega^{f(\mathbf{q})}|\mathbf{u}+\mathbf{W}\mathbf{q}\rangle,
    \end{align*}
    where $f(\mathbf{q}) = \mathbf{s}^\top \mathbf{q} + \mathbf{u}^\top \mathbf{V}\mathbf{q} + 2^{-1}\mathbf{q}^\top \mathbf{V}^\top \mathbf{W} \mathbf{q}$ and $\mathbf{u}\in\mathbb{F}_p^n$ is such that $\mathbf{V}^\top\mathbf{u} + \mathbf{s}\in\operatorname{row}(\mathbf{W})$. Here $\mathbf{V},\mathbf{W}\in\mathbb{F}_p^{n\times n}$ are the matrices with column vectors $\mathbf{v}_1,\dots,\mathbf{v}_n$ and $\mathbf{w}_1,\dots,\mathbf{w}_n$, respectively, and $\mathbf{s}=(s_1,\dots,s_n)\in\mathbb{F}_p^n$. Let $r = \operatorname{rank}(\mathbf{W})$. Without loss of generality, we can assume that $\mathbf{w}_{r+1} = \cdots = \mathbf{w}_n = \mathbf{0}$. In order to identify $\mathcal{S}$, we must learn $\mathbf{s}$ and $\big[\begin{smallmatrix} \mathbf{V} \\ \mathbf{W} \end{smallmatrix}\big]$.

    The first step is to learn $\mathbf{W}$. To do so, measure $|\mathcal{S}\rangle$ $2n+1$ times, in time complexity $O(n^2)$. The outcome vectors are $\mathbf{b}_0,\dots,\mathbf{b}_{2n}\in\mathbb{F}_p^n$, where $\mathbf{b}_i = \mathbf{u} + \mathbf{W}\mathbf{q}_i$ for some $\mathbf{q}_i\in\mathbb{F}_p^n$, $i\in\{0,\dots,2n\}$. By subtracting the vector $\mathbf{b}_0$ from all the others, we obtain $\mathbf{b}_i - \mathbf{b}_0 = \mathbf{W}(\mathbf{q}_i - \mathbf{q}_0)\in\operatorname{col}(\mathbf{W})$, $i\in[2n]$. We find a basis for $\operatorname{span}(\{\mathbf{b}_i-\mathbf{b}_0\}_{i\in[2n]})$ using Gaussian elimination in time $O(n^3)$. If $\operatorname{dim}(\operatorname{span}(\{\mathbf{b}_i-\mathbf{b}_0\}_{i\in[2n]})) = r$, then the vectors from the basis span $\mathbf{W}$. Without loss of generality, we can pick $\mathbf{w}_i = \mathbf{e}_i$ for $i\in[r]$. Learning $\mathbf{W}$ fails if $\operatorname{dim}(\operatorname{span}(\{\mathbf{b}_i-\mathbf{b}_0\}_{i\in[2n]})) < r$. The probability that all vectors $\{\mathbf{b}_i-\mathbf{b}_0\}_{i\in[2n]}$ are contained in any subspace in $\operatorname{col}(\mathbf{W})$ of dimension $r-1$ is $p^{-2n}$. By a union bound over all $(p^r - 1)/(p-1)$ subspaces in $\operatorname{col}(\mathbf{W})$ of dimension $r-1$, the failure probability is at most $p^{r-2n} \leq p^{-n}$.
    
    In possession of $\mathbf{W}$, we now learn the corresponding $\mathbf{V}$, starting with $\mathbf{v}_{r+1},\dots,\mathbf{v}_n$. Since $\langle \mathbf{v}_i, \mathbf{w}_j\rangle = \langle \mathbf{v}_j, \mathbf{w}_i\rangle$ for all $i,j\in[n]$, then $\mathbf{v}_{r+1},\dots,\mathbf{v}_n\in\operatorname{col}(\mathbf{W})^\perp$, i.e., $\mathbf{v}_{r+1},\dots,\mathbf{v}_n$ are orthogonal to $\operatorname{col}(\mathbf{W})$. Since $\{\omega^{s_i}\mathcal{W}_{\mathbf{v}_i,\mathbf{w}_i}\}_{i\in[n]}$ are independent, $\operatorname{dim}(\operatorname{span}(\{\mathbf{v}_i\}_{i=r+1}^n)) = n-r$. We conclude that $\operatorname{span}(\{\mathbf{v}_i\}_{i=r+1}^n) = \operatorname{col}(\mathbf{W})^\perp$. We can thus pick any basis for $\operatorname{col}(\mathbf{W})^\perp$ as the vectors $\mathbf{v}_{r+1},\dots,\mathbf{v}_n$, which can be done in time $O(n^3)$ using Gaussian elimination. As a possibility, we can simply take $\mathbf{v}_i = \mathbf{e}_i$ for $i\in\{r+1,\dots,n\}$.
    
    We now show how to learn the remaining corresponding vectors $\mathbf{v}_1,\dots,\mathbf{v}_r\in\mathbb{F}_p^n$ given the obtained $\mathbf{W}$. Let $m\triangleq 2n + \lceil \log_p{r}\rceil$. Let $\delta_1,\delta_2,\delta_3\in\mathbb{F}_p$, not all $0$, be such that $\delta_1^2 + \delta_2^2 + \delta_3^2 = 0$. These can be found in time $O(\operatorname{poly}\log{p})$ by using the Las Vegas method of~\cite{ivanyos2007efficient} or the deterministic method of~\cite{van2005deterministic}. We shall generate $m+1$ vector constraints of the form
    \begin{align}\label{eq:eq2}
        \mathbf{V}^\top \sum_{i=1}^3 \delta_i \mathbf{q}_i = \mathbf{W}^\top \mathbf{c} - \sum_{i=1}^3\delta_i\mathbf{s},
    \end{align}
    for random $\mathbf{q}_1,\mathbf{q}_2,\mathbf{q}_3\in \mathbf{u}+\operatorname{col}(\mathbf{W})$ and some $\mathbf{c}\in\mathbb{F}_p^n$. This is done by repeating the following procedure $m+1$ times. Start with $3$ copies of $|\mathcal{S}\rangle$ and an auxiliary register as
    \begin{align*}
        \left(\frac{p^{r/2}}{p^n}\sum_{\mathbf{t}\in\mathbb{F}_p^n}|\mathbf{W}\mathbf{t}\rangle\right)|\mathcal{S}\rangle^{\otimes 3} = \frac{p^{2r}}{p^{4n}}\sum_{\mathbf{t}\in\mathbb{F}_p^n}|\mathbf{W}\mathbf{t}\rangle\bigotimes_{i=1}^3 \sum_{\mathbf{q}_i\in\mathbb{F}_p^n}\omega^{f(\mathbf{q}_i)}|\mathbf{u}+\mathbf{W}\mathbf{q}_i\rangle.
    \end{align*}
    Let $U_\delta$ be the shift unitary such that $U_\delta|\mathbf{t}\rangle|\mathbf{q}\rangle = |\mathbf{t}\rangle|\mathbf{q}-\delta \mathbf{t}\rangle$ for all $\mathbf{t},\mathbf{q}\in\mathbb{F}_p^n$. For $i\in[3]$, apply $U_{\delta_i}$ onto the register $|\mathbf{u}+\mathbf{W}\mathbf{q}_i\rangle$ controlled on the register $|\mathbf{W}\mathbf{t}\rangle$ to obtain
    \begin{align*}
        \frac{p^{2r}}{p^{4n}}\sum_{\mathbf{t}\in\mathbb{F}_p^n}|\mathbf{W}\mathbf{t}\rangle \bigotimes_{i=1}^3 \sum_{\mathbf{q}_i\in\mathbb{F}_p^n}\omega^{f(\mathbf{q}_i)}|\mathbf{u}+\mathbf{W}(\mathbf{q}_i-\delta_i \mathbf{t})\rangle
        = \frac{p^{2r}}{p^{4n}}\sum_{\mathbf{t}\in\mathbb{F}_p^n}|\mathbf{W}\mathbf{t}\rangle \bigotimes_{i=1}^3 \sum_{\mathbf{q}_i\in\mathbb{F}_p^n}\omega^{f(\mathbf{q}_i+\delta_i \mathbf{t}) }|\mathbf{u}+\mathbf{W}\mathbf{q}_i\rangle.
    \end{align*}
    Applying $U_{\delta_i}$ requires $O(n)$ time, since it is composed of $n$ two-qudit shifts. We now notice that
    \begin{align*}
        f(\mathbf{q} + \delta \mathbf{t}) &= \mathbf{s}^\top (\mathbf{q}+\delta \mathbf{t}) + \mathbf{u}^\top \mathbf{V} (\mathbf{q}+\delta \mathbf{t}) + 2^{-1}(\mathbf{q} + \delta \mathbf{t})^\top \mathbf{V}^\top \mathbf{W} (\mathbf{q} + \delta \mathbf{t})\\
        &= f(\mathbf{q}) + \delta \mathbf{t}^\top \mathbf{s} + \delta \mathbf{t}^\top \mathbf{V}^\top \mathbf{u} + \delta \mathbf{t}^\top \mathbf{V}^\top \mathbf{W}\mathbf{q} + 2^{-1}\delta^2 \mathbf{t}^\top \mathbf{V}^\top \mathbf{W} \mathbf{t},
    \end{align*}
    where we used that $\mathbf{W}^\top \mathbf{V} = \mathbf{V}^\top \mathbf{W}$. Therefore
    \begin{align*}
        \sum_{i=1}^3 f(\mathbf{q}_i + \delta_i \mathbf{t}) - f(\mathbf{q}_i) &=  \mathbf{t}^\top \mathbf{V}^\top \mathbf{W}\left(\sum_{i=1}^3 \delta_i \mathbf{q}_i\right) + \left(\sum_{i=1}^3 \delta_i\right) \mathbf{t}^\top (\mathbf{s} + \mathbf{V}^\top \mathbf{u}) + 2^{-1}\left(\sum_{i=1}^3 \delta_i^2\right) \mathbf{t}^\top \mathbf{V}^\top \mathbf{W} \mathbf{t}\\
        &=  \mathbf{t}^\top \mathbf{W}^\top \mathbf{V} \left(\sum_{i=1}^3 \delta_i \mathbf{q}_i\right) + \left(\sum_{i=1}^3 \delta_i\right) \mathbf{t}^\top (\mathbf{s} + \mathbf{V}^\top \mathbf{u}).
    \end{align*}
    The resulting quantum state is
    \begin{align*}
        \frac{p^{2r}}{p^{4n}}\sum_{\mathbf{t}\in\mathbb{F}_p^n}|\mathbf{W}\mathbf{t}\rangle \bigotimes_{i=1}^3 \sum_{\mathbf{q}_i\in\mathbb{F}_p^n} \omega^{f(\mathbf{q}_i) + \delta_i\mathbf{t}^\top (\mathbf{W}^\top \mathbf{V}\mathbf{q}_i + \mathbf{s} + \mathbf{V}^\top\mathbf{u})}|\mathbf{u}+\mathbf{W}\mathbf{q}_i\rangle.
    \end{align*}
    Note that $\mathbf{s} + \mathbf{V}^\top \mathbf{u}\in\operatorname{row}(\mathbf{W}) \implies \mathbf{s} + \mathbf{V}^\top \mathbf{u} = \mathbf{W}^\top \mathbf{a}$ for some $\mathbf{a}\in\mathbb{F}_p^n$. Measure the registers $\bigotimes_{i=1}^3|\mathbf{u}+\mathbf{W}\mathbf{q}_i\rangle$ in time $O(n)$ to obtain $\mathbf{u}+\mathbf{W}\mathbf{q}_i'$ for some $\mathbf{q}'_i\in\mathbb{F}_p^n$, $i\in[3]$. Consider the inverse quantum Fourier transform given by
    \begin{align*}
        |\mathbf{t}\rangle = |t_1,\dots,t_n\rangle \mapsto \bigotimes_{i=1}^n\sum_{c_i\in\mathbb{F}_p}\omega^{-t_i c_i}|c_i\rangle = \sum_{\mathbf{c}\in\mathbb{F}_p^n}\omega^{-\mathbf{t}^\top\mathbf{c}}|\mathbf{c}\rangle, \qquad \forall \mathbf{t}\in\mathbb{F}_p^n.
    \end{align*}
    The inverse quantum Fourier transform requires $O(n)$ time, as it is composed of $n$ single-qudit rotations. Perform then an inverse quantum Fourier transform onto register $|\mathbf{W}\mathbf{t}\rangle$ to map (ignore the registers $\bigotimes_{i=1}^3|\mathbf{u}+\mathbf{W}\mathbf{q}_i'\rangle$)
    \begin{align*}
        \frac{p^{r/2}}{p^n}\sum_{\mathbf{t}\in\mathbb{F}_p^n} \omega^{\mathbf{t}^\top \mathbf{W}^\top \sum_{i=1}^3 \delta_i(\mathbf{V} \mathbf{q}_i' + \mathbf{a})} |\mathbf{W}\mathbf{t}\rangle &\mapsto \frac{p^{r/2}}{p^{3n/2}}\sum_{\mathbf{t},\mathbf{c}\in\mathbb{F}_p^n} \omega^{\mathbf{t}^\top \mathbf{W}^\top(\sum_{i=1}^3 \delta_i (\mathbf{V} \mathbf{q}_i' + \mathbf{a}) - \mathbf{c})} |\mathbf{c}\rangle \\
        &= \sqrt{\frac{p^r}{p^n}}\sum_{\substack{\mathbf{c}\in\mathbb{F}_p^n \\ \mathbf{W}^\top \mathbf{c} = \mathbf{W}^\top \sum_{i=1}^3 \delta_i (\mathbf{V} \mathbf{q}_i' + \mathbf{a})}}|\mathbf{c}\rangle,
    \end{align*}
    where we used that $\sum_{\mathbf{t}\in\mathbb{F}_p^n}\omega^{\mathbf{t}^\top \mathbf{z}}$ equals $0$ if $\mathbf{z} \neq \mathbf{0}\in\mathbb{F}_p^n$ and equals $p^n$ if $\mathbf{z} = \mathbf{0}\in\mathbb{F}_p^n$. By measuring the register $|\mathbf{c}\rangle$ (again in time $O(n)$), we get $\mathbf{c}'\in\mathbb{F}_p^n$ such that 
    \begin{align*}
        \mathbf{W}^\top \mathbf{c}' &= \mathbf{W}^\top \sum_{i=1}^3 \delta_i (\mathbf{V}\mathbf{q}_i' + \mathbf{a}) = \sum_{i=1}^3\delta_i (\mathbf{s} + \mathbf{V}^\top(\mathbf{u}+\mathbf{W}\mathbf{q}_i')),
    \end{align*}
    which is a vector constraint as in \cref{eq:eq2}.

    Repeat the above procedure $m+1$ times to get $m+1$ tuples $(\mathbf{c}^{(\ell)},\mathbf{q}_1^{(\ell)},\mathbf{q}_2^{(\ell)},\mathbf{q}_3^{(\ell)})\in\mathbb{F}_p^{4n}$, $\ell\in\{0,\dots,m\}$, satisfying \cref{eq:eq2}. Subtract the tuple $(\mathbf{c}^{(0)},\mathbf{q}_1^{(0)},\mathbf{q}_2^{(0)},\mathbf{q}_3^{(0)})$ from the other tuples to obtain the vector constraints
    \begin{align*}
        \mathbf{V}^\top \sum_{i=1}^3 \delta_i (\mathbf{q}_i^{(\ell)} - \mathbf{q}_i^{(0)}) = \mathbf{W}^\top (\mathbf{c}^{(\ell)}-\mathbf{c}^{(0)}), \quad \ell\in[m],
    \end{align*}
    with $rn$ unknown variables $V_{jk}$, $j\in[n]$, $k\in[r]$, which make up the $r$ vectors $\mathbf{v}_1,\dots,\mathbf{v}_r$. We notice that each vector constraint is made up of $r$ linear constraints that can be treated as independent of each other. In other words, we can solve each of the $r$ linear systems
    \begin{align*}
        \sum_{j=1}^n V_{jk}\left(\sum_{i=1}^3 \delta_i (\mathbf{q}_i^{(\ell)} - \mathbf{q}_i^{(0)})\right)_j  = (\mathbf{W}^\top (\mathbf{c}^{(\ell)}-\mathbf{c}^{(0)}))_k, \quad \ell\in[m],
    \end{align*}
    with unknown variables $V_{1k},\dots,V_{nk}$ independently from one another, where $k\in[r]$. The time complexity for solving each linear system with $n$ unknown variables is $O(n^3)$ using Gaussian elimination, for a total of $O(rn^3)$ runtime. Learning $\mathbf{v}_1,\dots,\mathbf{v}_r$ fails if, for some linear system, its $m$ rows lie in a subspace of $\mathbb{F}_p^n$ of dimension at most $n-1$. The probability that $m$ rows are contained in any one subspace of dimension $n-1$ is $p^{-m}$. By a union bound over all $(p^n - 1)/(p-1)$ subspaces of dimension $n-1$ and over all $r$ linear systems, the algorithm fails with probability at most $r p^{n-m} \leq p^{-n}$.  
    
    Finally, for $i\in[n]$, we measure a copy of $|\mathcal{S}\rangle$ in the eigenbasis of $\mathcal{W}_{\mathbf{v}_i,\mathbf{w}_i}$ to determine $s_i\in\mathbb{F}_p$ such that $\mathcal{W}_{\mathbf{v}_i,\mathbf{w}_i}|\mathcal{S}\rangle = \omega^{-s_i}|\mathcal{S}\rangle$. This completes the learning of $\mathbf{s}$ and $\big[\begin{smallmatrix} \mathbf{V} \\ \mathbf{W} \end{smallmatrix}\big]$, which identifies~$\mathcal{S}$. The algorithm uses $9n + 3\lceil\log_p{r}\rceil + 4$ copies of $|\mathcal{S}\rangle$. The most expensive step is performing $r$ Gaussian eliminations on $n$ variables, thus the total time complexity is $O(rn^3)$. The failure probability is at most $2p^{-n}$.
\end{proof}

\section{Pseudorandomness lower bounds}

In this section, we prove that quantum states with non-negligible stabiliser fidelity can be efficiently distinguished from Haar-random quantum states. As a consequence, the output of any Clifford circuit augmented with a few non-Clifford single-qudit gates can be efficiently distinguished from a Haar-random quantum state, and therefore Clifford circuits require several non-Clifford single-qudit gates in order to generate pseudorandom quantum states (\cref{def:pseudorandomness}). These results generalise the work of~\cite{grewal2022low,grewal2023improved} from qubits to qudits. The ideas behind our proof are somewhat similar to the ones from~\cite{grewal2022low,grewal2023improved}: a Haar-random state has vanishing stabiliser fidelity with very high probability, while the output of a Clifford circuit with just a few non-Clifford single-qudit gates has a high stabiliser dimension and thus stabiliser fidelity. Our algorithm, however, is quite different from the one of~\cite{grewal2022low,grewal2023improved}. Instead of using Bell difference sampling, we show that the stabiliser testing algorithm on qudits of Gross, Nezami, and Walter~\cite{gross2021schur} also works in this setting. Before presenting the algorithm, we need a few auxiliary results about Haar-random states.

\subsection{Anti-concentration of Haar-random states}
\label{sec:sec52}

In this section, we shall upper-bound the mass of $p_\psi(\mathbf{x})^2$ for when $|\psi\rangle$ is Haar-random. To do so, we will make use of L\'evy's lemma~\cite{milman1986asymptotic,ledoux2001concentration}.
\begin{fact}[L\'evy's lemma]
    Let $f: \mathbb{S}^d\to \mathbb{R}$ be a function defined on the $d$-dimensional hypersphere $\mathbb{S}^d$. Assume $f$ is $K$-Lipschitz, meaning that $|f(\psi) - f(\phi)| \leq K\|\psi - \phi\|$. Then, for every $\epsilon > 0$,
    \begin{align*}
        \operatorname*{\mathbb{P}}_{\psi\sim\mu_{\rm Haar}}[|f(\psi) - \mathbb{E}[f]| \geq \epsilon] \leq 2\exp\left(-\frac{(d+1)\epsilon^2}{9\pi^3 K^2}\right).
    \end{align*}
\end{fact}

The next result is basically~\cite[Lemma~20]{grewal2022low}. The proof can be found in \cref{app:proofs}.
\begin{lemma}\label{lem:Lipschitz}
    For any Weyl operator $\mathcal{W}_{\mathbf{x}}\in\mathscr{P}^n_p$, the function $f_{\mathbf{x}}:\mathbb{S}^{p^n}\to\mathbb{R}$ defined as $f_{\mathbf{x}}(|\psi\rangle) = \langle \psi|\mathcal{W}_{\mathbf{x}}|\psi\rangle$ is $2$-Lipschitz.
\end{lemma}
We are now ready to prove the main result of this section.
\begin{lemma}\label{lem:Haar-random}
    For any $\epsilon \geq p^{-n/2}$,
    \begin{align*}
        \operatorname*{\mathbb{P}}_{|\psi\rangle\sim\mu_{\rm Haar}}\left[ \sum_{\mathbf{x}\in\mathbb{F}_p^{2n}} p_\psi(\mathbf{x})^2 \geq 2\epsilon^4 \right] \leq 2p^{2n}\exp\left(-\frac{p^n\epsilon^2}{36\pi^3}\right).
    \end{align*}
\end{lemma}
\begin{proof}
    We first prove that, for any $\epsilon > 0$,
    \begin{align}\label{eq:auxiliary}
        \operatorname*{\mathbb{P}}_{|\psi\rangle\sim\mu_{\rm Haar}}\left[\exists \mathbf{x}\in\mathbb{F}_p^{2n}\setminus\{\mathbf{0}\}:|\langle \psi|\mathcal{W}_{\mathbf{x}}|\psi\rangle| \geq \epsilon \right] \leq 2p^{2n}\exp\left(-\frac{p^n\epsilon^2}{36\pi^3}\right).
    \end{align}
    Consider the function $f_{\mathbf{x}}(|\psi\rangle) = \langle \psi|\mathcal{W}_{\mathbf{x}}|\psi\rangle$ from \cref{lem:Lipschitz} for $\mathbf{x}\in\mathbb{F}_p^{2n}\setminus\{\mathbf{0}\}$. We know that $f_{\mathbf{x}}$ is $2$-Lipschitz. Moreover, $\mathbb{E}[f_{\mathbf{x}}] = 0$ over the Haar measure because $p^{-1}$-fraction of the eigenvalues of $\mathcal{W}_{\mathbf{x}}$ are $\omega^s$, $s\in\mathbb{F}_p$. \cref{eq:auxiliary} thus follows from L\'evy's lemma and a union bound over all $p^{2n}$ possible Weyl operators. We note that it generalises \cite[Corollary~22]{grewal2022low}.
    
    \cref{eq:auxiliary} readily implies that, with high probability and for $\epsilon \geq p^{-n/2}$,
    \[
        \sum_{\mathbf{x}\in\mathbb{F}_p^{2n}} p_\psi(\mathbf{x})^2 = \frac{1}{p^{2n}} + \frac{1}{p^{2n}}\sum_{\mathbf{x}\in\mathbb{F}_p^{2n}\setminus\{\mathbf{0}\}} |\langle\psi|\mathcal{W}_{\mathbf{x}}|\psi\rangle|^4 \leq \frac{1}{p^{2n}} + \epsilon^4 \leq 2\epsilon^4. \qedhere
    \]
\end{proof}

\subsection{Distinguishing non-negligible stabiliser fidelity from Haar-random}

In order to present our algorithm for distinguishing Haar-randomness from states with non-negligible stabiliser fidelity, we briefly review the stabiliser testing algorithm from~\cite{gross2021schur}. 
\begin{lemma}[\cite{gross2021schur}]
    Let $\mathcal{V} \triangleq p^{-n}\sum_{\mathbf{x}\in\mathbb{F}_p^{2n}}(\mathcal{W}_{\mathbf{x}}\otimes \mathcal{W}_{\mathbf{x}}^\dagger)^{\otimes 2}$. Then the operator $\Pi_{\rm accept} = \frac{1}{2}(\mathbb{I} + \mathcal{V})$ is a projector and the binary POVM $\{\Pi_{\rm accept}, \mathbb{I} - \Pi_{\rm accept}\}$ can be implemented in $O(n)$ time.
\end{lemma}
\begin{proof}
    It is easy to see that the operator $\mathcal{V}$ is Hermitian by a change of index. Moreover, $\mathcal{V}$ is unitary, since
    \begin{align*}
        \mathcal{V}^2 &= \frac{1}{p^{2n}}\sum_{\mathbf{x},\mathbf{y}\in\mathbb{F}_p^{2n}}(\mathcal{W}_{\mathbf{x}}\mathcal{W}_{\mathbf{y}}\otimes \mathcal{W}_{\mathbf{x}}^\dagger \mathcal{W}_{\mathbf{y}}^\dagger)^{\otimes 2} = \frac{1}{p^{2n}}\sum_{\mathbf{x},\mathbf{y}\in\mathbb{F}_p^{2n}}\omega^{2[\mathbf{x},\mathbf{y}]}(\mathcal{W}_{\mathbf{x}+\mathbf{y}}\otimes \mathcal{W}_{\mathbf{x}+\mathbf{y}}^\dagger)^{\otimes 2}\\
        &= \frac{1}{p^{2n}}\sum_{\mathbf{x},\mathbf{z}\in\mathbb{F}_p^{2n}}\omega^{2[\mathbf{x},\mathbf{z}]}(\mathcal{W}_{\mathbf{z}}\otimes \mathcal{W}_{\mathbf{z}}^\dagger)^{\otimes 2} = \mathbb{I}.
    \end{align*}
    Therefore, $\Pi_{\rm accept} = \frac{1}{2}(\mathbb{I} + \mathcal{V})$ is a projector, since $\Pi_{\rm accept}^2 = \frac{1}{4}(\mathbb{I} + 2\mathcal{V} + \mathcal{V}^2) = \Pi_{\rm accept}$.

    It is possible to implement the binary POVM $\{\Pi_{\rm accept}, \mathbb{I} - \Pi_{\rm accept}\}$ by employing quantum phase estimation with $O(1)$ ancillae, $O(1)$ controlled versions of $\mathcal{V}$, and $O(1)$ auxiliary operations (we do not require a high accuracy in determining the eigenvalues of $\mathcal{V}$ since they are simply $\pm 1$). And since we can write $\mathcal{V} = \big(p^{-1}\sum_{\mathbf{x}\in\mathbb{F}_p^{2}}(\mathcal{W}_{\mathbf{x}}\otimes \mathcal{W}_{\mathbf{x}}^\dagger)^{\otimes 2}\big)^{\otimes n}$, then a controlled version of $\mathcal{V}$ is equal to a composition of $n$ controlled versions of $p^{-1}\sum_{\mathbf{x}\in\mathbb{F}_p^{2}}(\mathcal{W}_{\mathbf{x}}\otimes \mathcal{W}_{\mathbf{x}}^\dagger)^{\otimes 2}$, each of which acts only on a constant (with respect to $n$) number of qudits. Therefore, the POVM $\{\Pi_{\rm accept}, \mathbb{I} - \Pi_{\rm accept}\}$ can be implemented in time $O(n)$.
\end{proof}
By using $\Pi_{\rm accept}$ as the accepting part of the binary POVM $\{\Pi_{\rm accept}, \mathbb{I} - \Pi_{\rm accept}\}$, the accepting probability when measuring the state $|\psi\rangle^{\otimes 4}$ is
\begin{align*}
    \mathbb{P}[{\rm accept}] = \operatorname{Tr}[\psi^{\otimes 4}\Pi_{\rm accept}] = \frac{1}{2}(1 + \operatorname{Tr}[\psi^{\otimes 4}\mathcal{V}]) = \frac{1}{2} + \frac{1}{2p^n}\sum_{\mathbf{x}\in\mathbb{F}_p^{2n}}|{\operatorname{Tr}}[\psi\mathcal{W}_{\mathbf{x}}]|^4  = \frac{1}{2} + \frac{p^{n}}{2}\sum_{\mathbf{x}\in\mathbb{F}_p^{2n}}p_\psi(\mathbf{x})^2.
\end{align*}
Notice that, if $|\psi\rangle$ is a stabiliser state, then the acceptance probability is $1$. Gross, Nezami, and Walter~\cite{gross2021schur} employed the above projection to test whether a given state is a stabiliser state or is far from one (i.e., has low stabiliser fidelity), conditioned on one of the cases being true. Their algorithm simply measures a few copies of $|\psi\rangle$ using $\{\Pi_{\rm accept}, \mathbb{I} - \Pi_{\rm accept}\}$. With high probability, if it always accepts, then $|\psi\rangle$ is a stabiliser state, otherwise $|\psi\rangle$ is far from one. 

We can use the same idea to decide whether a state $|\psi\rangle \in(\mathbb{C}^p)^{\otimes n}$ is Haar-random or has non-negligible stabiliser fidelity, i.e., by repetitively measuring copies of $|\psi\rangle$ with the POVM $\{\Pi_{\rm accept}, \mathbb{I} - \Pi_{\rm accept}\}$. With enough measurements we can estimate the mass $\sum_{\mathbf{x}\in\mathbb{F}_p^{2n}} p_\psi(\mathbf{x})^2$ which, according to the previous section and \cref{lem:stabiliser_fidelity_inequality}, tells us what the correct case is. This is elaborated in the next result.

\begin{algorithm}[h]
\caption{Distinguishing non-negligible-stabiliser-fidelity states from Haar-random}
\DontPrintSemicolon
\label{alg:clifford}
\SetKwInput{KwPromise}{Promise}

    \KwIn{$4m$ copies of $|\psi\rangle\in(\mathbb{C}^p)^{\otimes n}$, where $m \triangleq \frac{1}{4}\lceil 72k^{8}\ln(2/\delta)\rceil$.}
    \KwPromise{$|\psi\rangle$ is Haar-random or has stabiliser fidelity at least $k^{-1}$.}
    \KwOut{$0$ if $|\psi\rangle$ is Haar-random and $1$ otherwise, with probability at least $1-\delta$.}
    
    Let $X = 0$.

   \For{$i\gets 1$ \KwTo $m$}
   {    
        Measure $|\psi\rangle^{\otimes 4}$ using the POVM $\{\Pi_{\rm accept}, \mathbb{I} - \Pi_{\rm accept}\}$.

        If the outcome is $\Pi_{\rm accept}$, then $X \gets X + 1/m$, otherwise $X \gets X - 1/m$.
    }
    \Return $0$ if $X < 2k^{-4}/3$ and $1$ otherwise.
\end{algorithm}

\begin{theorem}\label{thr:algorithm_pseudorandomness}
    Let $|\psi\rangle\in(\mathbb{C}^p)^{\otimes n}$ be a state promised to be either Haar-random or to have stabiliser fidelity at least $k^{-1}$. Let $\delta\in(0,1)$ and assume that $k^{-1} \geq p^{-n/5}$. There is a quantum algorithm {\rm (\cref{alg:clifford})} that distinguishes the two cases with probability at least $1-\delta$ for large enough $n$, uses $\lceil 72k^{8}\ln(2/\delta)\rceil$ copies of $|\psi\rangle$, and has runtime $O(n k^{8}\log(1/\delta))$. 
\end{theorem}
\begin{proof}
    Let $m \triangleq \frac{1}{4}\lceil 72k^{8}\ln(2/\delta)\rceil \geq 18k^{8}\ln(2/\delta)$. Use the POVM $\{\Pi_{\rm accept}, \mathbb{I} - \Pi_{\rm accept}\}$ on $|\psi\rangle^{\otimes 4}$ a number of $m$ times, where $\Pi_{\rm accept} = \frac{1}{2}(\mathbb{I} + p^{-n}\sum_{\mathbf{x}\in\mathbb{F}_p^{2n}}(\mathcal{W}_{\mathbf{x}}\otimes \mathcal{W}_{\mathbf{x}}^\dagger)^{\otimes 2})$. For $i\in[m]$, define the $\{-1,1\}$ indicator random variable
    \begin{align*}
        X_i \triangleq \begin{cases}
            +1 &\text{if}~\Pi_{\rm accept}~\text{is measured},\\
            -1 &\text{otherwise}.
        \end{cases}
    \end{align*}
    As previously discussed, the acceptance probability of the POVM $\{\Pi_{\rm accept}, \mathbb{I} - \Pi_{\rm accept}\}$ on $|\psi\rangle^{\otimes 4}$ is $\mathbb{P}[{\rm accept}] = \frac{1}{2} + \frac{p^{n}}{2}\sum_{\mathbf{x}\in\mathbb{F}_p^{2n}}p_\psi(\mathbf{x})^2$, which implies that
    \begin{align*}
        \mathbb{E}[X_i] = p^{n}\sum_{\mathbf{x}\in\mathbb{F}_p^{2n}}p_\psi(\mathbf{x})^2.
    \end{align*}
    If $F_{\mathcal{S}}(|\psi\rangle)\geq k^{-1}$, then, according to \cref{lem:stabiliser_fidelity_inequality}, $p^n\sum_{\mathbf{x}\in\mathbb{F}_p^{2n}}p_\psi(\mathbf{x})^2 \geq F_{\mathcal{S}}(|\psi\rangle)^4 \geq k^{-4}$. Hence
    \begin{align*}
        \mathbb{E}[X_i~|~F_{\mathcal{S}}(|\psi\rangle) \geq k^{-1}] \geq k^{-4}.
    \end{align*}
    Therefore, by a Hoeffding's bound,
    \begin{align*}
        \mathbb{P}\left[\frac{1}{m}\sum_{i=1}^m X_i < k^{-4} - \eta ~\Big|~ F_{\mathcal{S}}(|\psi\rangle) \geq k^{-1} \right] \leq e^{-m\eta^2/2}.
    \end{align*}
    By taking $\eta = k^{-4}/3$, then $\frac{1}{m}\sum_{i=1}^m X_i < 2k^{-4}/3$ with probability at most $e^{-mk^{-8}/18} \leq \delta/2$.
    
    On the other hand, if $|\psi\rangle$ is Haar-random, then $\sum_{\mathbf{x}\in\mathbb{F}_p^{2n}}p_\psi(\mathbf{x})^2 \leq 2\epsilon^4$ with probability at least $1-2p^{2n}\exp(-p^{n}\epsilon^2/36\pi^3)$, according to \cref{lem:Haar-random}. By taking $\epsilon^2 = k^{-4}/\sqrt{6}$, this probability is at least $1-2p^{2n}\exp(-k^{-4}p^{n}/36\sqrt{6}\pi^3) \geq 1 - \delta/2$ for sufficiently large $n$ since $k^{-1} \geq p^{-n/5}$. Thus, with probability at least $1-\delta/2$,
    \begin{align}\label{eq:haar_expectation}
        \mathbb{E}[X_i~|~|\psi\rangle~\text{Haar-random}] \leq k^{-4}/3. 
    \end{align}
    Assuming that the above inequality holds, then, by a Hoeffding's bound,
    \begin{align*}
        \mathbb{P}\left[\frac{1}{m}\sum_{i=1}^m X_i > k^{-4}/3 + \eta ~\Big|~ |\psi\rangle~\text{Haar-random} \right] \leq e^{-m\eta^2/2}.
    \end{align*}
    By taking $\eta = k^{-4}/3$, then $\frac{1}{m}\sum_{i=1}^m X_i > 2k^{-4}/3$ with probability at most $e^{-mk^{-8}/18} \leq \delta/2$.

    Therefore, \cref{alg:clifford} can distinguish whether $|\psi\rangle$ is Haar-random or has stabiliser fidelity at least $k^{-1}$ by computing $\frac{1}{m}\sum_{i=1}^m X_i$. If $\frac{1}{m}\sum_{i=1}^m X_i > 2k^{-4}/3$, then $F_{\mathcal{S}}(|\psi\rangle) \geq k^{-1}$ with probability at least $1-\delta/2$. Otherwise, if $\frac{1}{m}\sum_{i=1}^m X_i < 2k^{-4}/3$, then $|\psi\rangle$ is Haar-random with probability at least $1-\delta$ (there is a $\delta$-failure probability of \cref{eq:haar_expectation} not holding). In total, the algorithm requires $4m = \lceil 72k^{8}\ln(2/\delta)\rceil$ copies of $|\psi\rangle$. Its runtime is $m$ times the runtime of using the POVM $\{\Pi_{\rm accept}, \mathbb{I} - \Pi_{\rm accept}\}$, i.e., $O(nm) = O(nk^{8}\log(1/\delta))$.
\end{proof}

\subsection{Distinguishing $t$-doped Clifford circuits from Haar-randomness}
\label{sec:sec51}

As a consequence of \cref{thr:algorithm_pseudorandomness}, we prove in this section that (i) output states of doped Clifford circuits can be efficiently distinguished from Haar-random states and (ii) a Clifford circuit requires several non-Clifford single-qudit gates in other to generate pseudorandom quantum states. Following~\cite{leone2024learning,leone2023learning}, we first define $t$-doped Clifford circuits over qudits. Recall that a Clifford gate is any element belonging to the Clifford group $\mathscr{C}^n_p$.
\begin{definition}[$t$-doped Clifford circuits]
    A $t$-doped Clifford circuit is a quantum circuit composed of Clifford gates and at most $t$ non-Clifford single-qudit gates that starts in the state $|0\rangle^{\otimes n}$.
\end{definition}

We now lower bound the stabiliser fidelity $F_{\mathcal{S}}(|\psi\rangle)$ for any output $|\psi\rangle$ of a $t$-doped Clifford circuit. We shall need the following auxiliary result, which is a simple generalisation of~\cite[Lemma~4.2]{grewal2023improved} (see \cref{app:proofs} for a proof).
\begin{lemma}\label{lem:t-doped}
    Let $|\psi\rangle\in(\mathbb{C}^p)^{\otimes n}$ be the output state of a $t$-doped Clifford circuit. Then the stabiliser dimension of $|\psi\rangle$ is at least $n-2t$, i.e., $\operatorname{dim}(\operatorname{Weyl}(|\psi\rangle)) \geq n-2t$.
\end{lemma}

\begin{lemma}\label{lem:stabiliser_fidelity_lower_bound}
    Let $|\psi\rangle\in(\mathbb{C}^p)^{\otimes n}$ be the output of a $t$-doped Clifford circuit with $t\leq n/2$. Then $F_{\mathcal{S}}(|\psi\rangle) \geq p^{-2t}$.
\end{lemma}
\begin{proof}
    Using \cref{lem:stabiliser_fidelity_generalisation}, \cref{lem:t-doped}, and that $\sum_{\mathbf{x}\in\operatorname{Weyl}(|\psi\rangle)^{\independent}} p_\psi(\mathbf{x}) = 1$,
    \[
        F_{\mathcal{S}}(|\psi\rangle) \geq \frac{p^n}{|\operatorname{Weyl}(|\psi\rangle)^{\independent}|}\sum_{\mathbf{x}\in\operatorname{Weyl}(|\psi\rangle)^{\independent}}p_\psi(\mathbf{x}) = \frac{p^n}{|\operatorname{Weyl}(|\psi\rangle)^{\independent}|} \geq p^{-2t}. \qedhere
    \]
\end{proof}

\cref{lem:stabiliser_fidelity_lower_bound} implies that we can employ \cref{alg:clifford} to distinguish between a Haar-random state and the output of a $t$-doped Clifford circuit.
\begin{theorem}\label{thr:t-doped_complexity}
    Let $|\psi\rangle\in(\mathbb{C}^p)^{\otimes n}$ be a state promised to be either Haar-random or the output of a $t$-doped Clifford circuit. Let $\delta\in(0,1)$ and assume that $t < n/10$. There is a quantum algorithm ({\rm \cref{alg:clifford}} with $k=p^{-2t}$) that distinguishes the two cases with probability at least $1-\delta$ for large enough $n$, uses $\lceil 72p^{16t}\ln(2/\delta)\rceil$ copies of $|\psi\rangle$, and has runtime $O(n p^{16t}\log(1/\delta))$. 
\end{theorem}
\begin{remark}
    The sample and time complexity of {\rm \cref{thr:t-doped_complexity}} can be improved to $\lceil 72p^{4t}\ln(2/\delta)\rceil$ and $O(n p^{4t}\log(1/\delta))$, respectively, by working directly with the lower bound 
    \begin{align*}
        p^n\sum_{\mathbf{x}\in\mathbb{F}_p^{2n}}p_{\psi}(\mathbf{x})^2 \geq p^n\sum_{\mathbf{x}\in\operatorname{Weyl}(|\psi\rangle)}p_{\psi}(\mathbf{x})^2 = \frac{|\operatorname{Weyl}(|\psi\rangle)|}{p^{n}} \geq p^{-2t}
    \end{align*}
    in the proof of {\rm \cref{thr:algorithm_pseudorandomness}} if $|\psi\rangle$ is the output of a $t$-doped Clifford circuit.
\end{remark}

If $t = O(\log{n}/\log{p})$, then \cref{thr:t-doped_complexity} has sample and  time complexity $\operatorname{poly}(n)$. As an immediate corollary, we have the following. 
\begin{corollary}\label{cor:pseudorandomness_lower bound}
    Any doped Clifford circuit that uses $O(\log{n}/\log{p})$ non-Clifford single-qudit gates cannot produce an ensemble of pseudorandom quantum states in $(\mathbb{C}^p)^{\otimes n}$.
\end{corollary}

Even though we have stated the results from this section for $p>2$ prime, they can be generalised to arbitrary dimension $d=p>2$. The results from \cref{sec:sec51} and \cref{sec:sec52} also hold for $d>2$, while \cref{alg:clifford} can be generalised by considering the POVM $\{\Pi_{\rm accept}, \mathbb{I} - \Pi_{\rm accept}\}$ from~\cite[Section~3.2]{gross2021schur} for arbitrary dimension, where $\Pi_{\rm accept} = \frac{1}{2}(\mathbb{I} + \mathcal{V}_r)$ with
\begin{align*}
    \mathcal{V}_r = \frac{1}{d^n}\sum_{\mathbf{x}\in\mathbb{Z}_d^{2n}}(\mathcal{W}_{\mathbf{x}}\otimes \mathcal{W}_{\mathbf{x}}^\dagger)^{\otimes r}
\end{align*}
for any $r\geq 2$ such that $\operatorname{gcd}(d,r) = 1$. The sample and time complexity of \cref{thr:algorithm_pseudorandomness} become $O(rk^{4r}\log(1/\delta))$ and $O(nrk^{4r}\log(1/\delta))$, respectively, while the sample and time complexity of \cref{thr:t-doped_complexity} become $O(rd^{4t}\log(1/\delta))$ and $O(nrd^{4t}\log(1/\delta))$, respectively. The pseudorandomness bound is still $O(\log{n}/\log{d})$.

\section*{Acknowledgements}

JFD thanks Sabee Grewal, Vishnu Iyer, and William Kretschmer for valuable discussions. JFD thanks Sabee Grewal for providing the proof of \cref{lem:Lipschitz}, Vishnu Iyer for suggesting to work with stabiliser fidelity, and Robbie King and William Kretschmer for pointing out Ref.~\cite{king2024exponential}. This research is funded by ERC grant No.\ 810115-DYNASNET and by the National Research Foundation, Singapore and A*STAR under its CQT Bridging Grant and its Quantum Engineering Programme under grant NRF2021-QEP2-02-P05. 
The research of GI was also supported by the Hungarian Ministry of Innovation and Technology NRDI Office within the framework of the the Artificial Intelligence National Laboratory Program. MS also acknowledges support from the joint Israel-Singapore NRFISF Research grant NRF2020-NRF-ISF004-3528. This project was finalised while JFD and MS were visiting researchers at the Simons Institute for the Theory of Computing. 

\DeclareRobustCommand{\DE}[2]{#2}
\bibliographystyle{alpha}
\bibliography{Pauli.bib}

\appendix

\section{Auxiliary results}
\label{app:proofs}

Here we prove auxiliary results from the main text.

\begin{lemma}\label{lem:cardinality}
    Let $\mathscr{V}\subseteq \mathbb{F}_p^n$ and $\mathscr{X}\subseteq \mathbb{F}_p^{2n}$ be subspaces and $\mathbf{w}\notin\mathscr{V}^\perp$ and $\mathbf{y}\notin\mathscr{X}^{\independent}$. Consider the sets $\mathscr{S}_k \triangleq \{\mathbf{v}\in\mathscr{V}: \langle \mathbf{v},\mathbf{w}\rangle = k\}$ and $\mathscr{T}_k \triangleq \{\mathbf{x}\in\mathscr{X}: [\mathbf{x},\mathbf{y}] = k\}$ for $k\in\mathbb{F}_p$. Then $|\mathscr{S}_k| = |\mathscr{V}|/p$ and $|\mathscr{T}_k| = |\mathscr{X}|/p$ for all $k\in\mathbb{F}_p$.
\end{lemma}
\begin{proof}
    The proof that $|\mathscr{T}_k| = |\mathscr{X}|/p$ is exactly the same as for $|\mathscr{S}_k| = |\mathscr{V}|/p$, so we focus on the latter. Since $\mathbf{w}\notin \mathscr{V}^\perp$, there exist $k\in\mathbb{F}_p\setminus\{0\}$ such that $\mathscr{S}_k$ is non-empty. Given the set $\mathscr{S}_{k'}$ for $k'\in\mathbb{F}_p\setminus\{0\}$, $|\mathscr{S}_k| = |\mathscr{S}_{k'}|$ due to the bijection $\mathscr{S}_k \to \mathscr{S}_{k'}$ defined by $\mathbf{v}\mapsto (k'/k)\mathbf{v}$, which is well defined since $\mathscr{V}$ is a subspace and $p$ is prime. Regarding $\mathscr{S}_0$, for all $\mathbf{v}\in \mathscr{S}_k$, $\mathbf{v}+\mathscr{S}_0 \subseteq \mathscr{S}_k$, thus $|\mathscr{S}_0| \leq |\mathscr{S}_k|$. On the other hand, for all $\mathbf{v},\mathbf{v}'\in \mathscr{S}_k$, $\mathbf{v}-\mathbf{v}'\in \mathscr{S}_0$, thus $|\mathscr{S}_k| \leq |\mathscr{S}_0|$. Therefore $|\mathscr{S}_0| = |\mathscr{S}_k|$, which concludes the proof.
\end{proof}

\begin{replemma}{lem:sum}
    Let $\mathscr{V}\subseteq\mathbb{F}_p^n$ and $\mathscr{X}\subseteq\mathbb{F}_p^{2n}$ be subspaces and $\mathbf{w} \in\mathbb{F}_p^n$ and $\mathbf{y} \in\mathbb{F}_p^{2n}$. Then
    \begin{align*}
        \sum_{\mathbf{v}\in\mathscr{V}} \omega^{\langle\mathbf{v}, \mathbf{w}\rangle} = \begin{cases}
            |\mathscr{V}| &\text{if}~\mathbf{w} \in \mathscr{V}^\perp,\\
            0 &\text{if}~ \mathbf{w}\notin \mathscr{V}^\perp,
        \end{cases}\qquad\qquad
        \sum_{\mathbf{x}\in\mathscr{X}} \omega^{[\mathbf{x}, \mathbf{y}]} = \begin{cases}
            |\mathscr{X}| &\text{if}~\mathbf{y} \in \mathscr{X}^{\independent},\\
            0 &\text{if}~ \mathbf{y}\notin \mathscr{X}^{\independent}.
        \end{cases}
    \end{align*}
\end{replemma}
\begin{proof}
    The proof for the symplectic product is the same as for the inner product, so we focus on the latter. The case $\mathbf{w}\in\mathscr{V}^\perp$ is straightforward. If $\mathbf{w}\notin \mathscr{V}^\perp$, note that, according to \cref{lem:cardinality}, the sets $\mathscr{S}_k \triangleq \{\mathbf{v}\in\mathscr{V}: \langle \mathbf{v},\mathbf{w}\rangle = k\}$ have the same cardinality $|\mathscr{V}|/p$, where $k\in\mathbb{F}_p$. Therefore
    \[
        \sum_{\mathbf{v}\in\mathscr{V}}\omega^{\langle \mathbf{v},\mathbf{w}\rangle} = |\mathscr{S}_0| + \sum_{\mathbf{v}\in\mathscr{V}:\langle \mathbf{v},\mathbf{w}\rangle \neq 0}\omega^{\langle \mathbf{v},\mathbf{w}\rangle} = |\mathscr{S}_0| + \sum_{k=1}^{p-1}\sum_{\mathbf{v}\in\mathscr{V}:\langle \mathbf{v},\mathbf{w}\rangle =k} \omega^{k} = \frac{|\mathscr{V}|}{p}\left(1 +\sum_{k=1}^{p-1} \omega^{k}\right) = 0. \qedhere
    \]
\end{proof}

\begin{replemma}{lem:stabiliser_group_form}
    Any stabiliser group $\mathcal{S}\subset\mathscr{P}^n_p$ can be written as $\mathcal{S} = \{\omega^{[\mathbf{a},\mathbf{x}]} \mathcal{W}_{\mathbf{x}} :  \mathbf{x}\in \mathscr{M}\}$, where $\mathbf{a}\in\mathbb{F}_p^{2n}$ and $\mathscr{M}\subset\mathbb{F}_p^{2n}$ is a Lagrangian subspace such that $\operatorname{dim}(\mathscr{M}) = n$. As a consequence,
    $\mathcal{S}$ is commutative and $|\mathcal{S}| = p^n$.
\end{replemma}
\begin{proof}
    There cannot be two operators $\omega^{s}\mathcal{W}_{\mathbf{x}},\omega^{s'}\mathcal{W}_{\mathbf{x}}\in\mathcal{S}$ with $s\neq s'$ since $\mathcal{S}\cap\mathrm{Z}(\mathscr{P}_p^n) = \{\mathbb{I}\}$. Thus $\mathcal{S} = \{\omega^{s(\mathbf{x})} \mathcal{W}_{\mathbf{x}} :  \mathbf{x}\in \mathscr{M}\}$ for some set $\mathscr{M}\subset\mathbb{F}_p^{2n}$ and $s:\mathscr{M}\to\mathbb{F}_p$. That $\mathscr{M}$ is a subspace follows from $\mathcal{S}$ being a group. It is isotropic because, given $\omega^{s(\mathbf{x})}\mathcal{W}_{\mathbf{x}}, \omega^{s(\mathbf{y})}\mathcal{W}_{\mathbf{y}}\in\mathcal{S}$, then $(\omega^{s(\mathbf{x})}\mathcal{W}_{\mathbf{x}})^\dagger (\omega^{s(\mathbf{y})}\mathcal{W}_{\mathbf{y}})^\dagger (\omega^{s(\mathbf{x})}\mathcal{W}_{\mathbf{x}})(\omega^{s(\mathbf{y})}\mathcal{W}_{\mathbf{y}}) = \omega^{[\mathbf{x},\mathbf{y}]}\cdot\mathbb{I}$ is an element of $\mathcal{S}$, which implies that $[\mathbf{x},\mathbf{y}] = 0$ since $\mathcal{S}\cap \mathrm{Z}(\mathscr{P}_p^n) = \{\mathbb{I}\}$. The function $s$ must be linear since $(\omega^{s(\mathbf{x})}\mathcal{W}_{\mathbf{x}})(\omega^{s(\mathbf{y})}\mathcal{W}_{\mathbf{y}}) = \omega^{s(\mathbf{x}) + s(\mathbf{y})}\mathcal{W}_{\mathbf{x}+\mathbf{y}}$ (due to commutativity), and therefore $s(\mathbf{x}) + s(\mathbf{y}) = s(\mathbf{x} + \mathbf{y})$ (by the uniqueness of phase of each Weyl operator). Moreover, $s(\mathbf{0}) = 0$ again due to $\mathcal{S}\cap\mathrm{Z}(\mathscr{P}_p^n) = \{\mathbb{I}\}$. Therefore, $\omega^{s(\cdot)}$ is a character and we can write $s(\mathbf{x}) = [\mathbf{a},\mathbf{x}]$ for some $\mathbf{a}\in\mathbb{F}_p^{2n}$.
    Let us now prove that $\operatorname{dim}(\mathscr{M}) = n$. Write $m \triangleq \operatorname{dim}(\mathscr{M})$. Thus there are $m$ linearly independent elements $\mathbf{x}_1,\dots,\mathbf{x}_m\in\mathbb{F}_p^{2n}$ such that $[\mathbf{x}_i,\mathbf{x}_j] = 0$ for all $i,j\in[m]$. Write $\mathbf{x}_i = (\mathbf{v}_i,\mathbf{w}_i)$ for all $i\in[m]$ and let $\mathbf{V},\mathbf{W}\in\mathbb{F}_p^{n\times m}$ be the matrices with column vectors $\mathbf{v}_1,\dots,\mathbf{v}_,$ and $\mathbf{w}_1,\dots,\mathbf{w}_,$, respectively. Then $\mathbf{V}^\top \mathbf{W} = \mathbf{W}^\top \mathbf{V}$ and ${\operatorname{rank}}\big(\bigl[\begin{smallmatrix} \mathbf{V} \\ \mathbf{W} \end{smallmatrix} \bigr]\big) = {\operatorname{rank}}\big(\bigl[\begin{smallmatrix} \mathbf{W} \\ -\mathbf{V} \end{smallmatrix} \bigr]\big) = m$. By the rank-nullity theorem, we have that ${\operatorname{rank}}([\begin{smallmatrix} \mathbf{V}^\top & \mathbf{W}^\top \end{smallmatrix}]) + \operatorname{dim}(\operatorname{null}([\begin{smallmatrix} \mathbf{V}^\top & \mathbf{W}^\top \end{smallmatrix}])) = 2n \implies \operatorname{dim}(\operatorname{null}([\begin{smallmatrix} \mathbf{V}^\top & \mathbf{W}^\top \end{smallmatrix}])) = 2n-m$. Moreover, ${\operatorname{col}}\big(\big[\begin{smallmatrix} \mathbf{W} \\ -\mathbf{V} \end{smallmatrix}\big]\big)\subseteq \operatorname{null}([\begin{smallmatrix} \mathbf{V}^\top & \mathbf{W}^\top \end{smallmatrix} ])$ since $[\begin{smallmatrix} \mathbf{V}^\top & \mathbf{W}^\top \end{smallmatrix} ]\big[\begin{smallmatrix} \mathbf{W} \\ -\mathbf{V} \end{smallmatrix}\big]\mathbf{q} = \mathbf{V}^\top \mathbf{W}\mathbf{q} - \mathbf{W}^\top \mathbf{V}\mathbf{q} = 0$ for all $\mathbf{q}\in\mathbb{F}_p^n$. Therefore $\operatorname{dim}(\operatorname{null}([\begin{smallmatrix} \mathbf{V}^\top & \mathbf{W}^\top \end{smallmatrix}])) \geq {\operatorname{rank}}\big(\bigl[\begin{smallmatrix} \mathbf{W} \\ -\mathbf{V} \end{smallmatrix} \bigr]\big) \implies 2n-m \geq m \implies m \leq n$. If $m < n$, then $\operatorname{dim}(\mathscr{M} + \mathbf{x}) > \operatorname{dim}(\mathscr{M})$ for any non-zero vector $\mathbf{x}\in{\operatorname{col}}\big(\big[\begin{smallmatrix} \mathbf{W} \\ -\mathbf{V} \end{smallmatrix}\big]\big)\setminus \operatorname{null}([\begin{smallmatrix} \mathbf{V}^\top & \mathbf{W}^\top \end{smallmatrix}])$, which means that the stabiliser group $\mathcal{S}$ is not maximal, a contradiction. Thus $\operatorname{dim}(\mathscr{M}) = n$.  Finally, $\mathscr{M}$ is Lagrangian as a consequence of being isotropic and maximal.
\end{proof}

\begin{replemma}{lem:properties}
    Let $\mathcal{S} = \{\omega^{[\mathbf{a},\mathbf{x}]}\mathcal{W}_{\mathbf{x}}:\mathbf{x}\in\mathscr{M}\}$ be a stabiliser group. Let $\mathbf{V},\mathbf{W}\in\mathbb{F}_p^{n\times n}$ be matrices such that $\mathscr{M} = {\operatorname{col}}\big(\big[\begin{smallmatrix} \mathbf{V} \\ \mathbf{W} \end{smallmatrix}\big]\big)$. Then
    \begin{multicols}{2}
    \begin{enumerate}[label=\alph*.]
        \item $\mathbf{V}^\top \mathbf{W} = \mathbf{W}^\top \mathbf{V}$;
        \item ${\operatorname{rank}}\big(\bigl[\begin{smallmatrix} \mathbf{V} \\ \mathbf{W} \end{smallmatrix} \bigr]\big) = \operatorname{dim}(\mathscr{M}) = n$;
        \item $\operatorname{row}(\mathbf{V}) + \operatorname{row}(\mathbf{W}) = \mathbb{F}_p^n$;
        \item $\operatorname{null}(\mathbf{V})\cap\operatorname{null}(\mathbf{W}) = \{\mathbf{0}\}$;
        \item $\operatorname{null}([\begin{smallmatrix} \mathbf{V}^\top & \mathbf{W}^\top \end{smallmatrix} ]) = {\operatorname{col}}\big(\big[\begin{smallmatrix} \mathbf{W} \\ -\mathbf{V} \end{smallmatrix}\big]\big)$;
        \item $\operatorname{null}(\mathbf{V}^\top \mathbf{W}) = \operatorname{null}(\mathbf{V}) + \operatorname{null}(\mathbf{W})$;
        \item $\operatorname{null}(\mathbf{W}^\top) \subseteq \operatorname{col}(\mathbf{V})$;
        \item $\operatorname{null}(\mathbf{V}^\top) \subseteq \operatorname{col}(\mathbf{W})$.
    \end{enumerate}
    \end{multicols}
\end{replemma}
\begin{proof}
    The first fact follows from the commutativity of the stabiliser group, while ${\operatorname{rank}}\big(\bigl[\begin{smallmatrix} \mathbf{V} \\ \mathbf{W} \end{smallmatrix} \bigr]\big) = \operatorname{dim}(\mathscr{M}) = n$ follows from \cref{lem:stabiliser_group_form}. Note that $[\begin{smallmatrix} \mathbf{V}^\top & \mathbf{W}^\top \end{smallmatrix} ]\big[\begin{smallmatrix} \mathbf{v} \\ \mathbf{w} \end{smallmatrix}\big] = \mathbf{V}^\top \mathbf{v} + \mathbf{W}^\top \mathbf{w}$ for $\mathbf{v},\mathbf{w}\in\mathbb{F}_p^n$, hence ${\operatorname{row}}\big(\bigl[\begin{smallmatrix} \mathbf{V} \\ \mathbf{W} \end{smallmatrix} \bigr]\big) = \operatorname{row}(\mathbf{V}) + \operatorname{row}(\mathbf{W})$. Thus ${\operatorname{rank}}\big(\bigl[\begin{smallmatrix} \mathbf{V} \\ \mathbf{W} \end{smallmatrix} \bigr]\big) = n \implies \operatorname{row}(\mathbf{V}) + \operatorname{row}(\mathbf{W}) = \mathbb{F}_p^n$ and so $\operatorname{null}(\mathbf{V})\cap\operatorname{null}(\mathbf{W}) = (\operatorname{row}(\mathbf{V}) + \operatorname{row}(\mathbf{W}))^\perp = (\mathbb{F}_p^n)^\perp = \{\mathbf{0}\}$. 

    Now note that $\operatorname{rank}([\begin{smallmatrix} \mathbf{V}^\top & \mathbf{W}^\top \end{smallmatrix}]) + \operatorname{dim}(\operatorname{null}([\begin{smallmatrix} \mathbf{V}^\top & \mathbf{W}^\top \end{smallmatrix}])) = 2n \implies \operatorname{dim}(\operatorname{null}([\begin{smallmatrix} \mathbf{V}^\top & \mathbf{W}^\top \end{smallmatrix}])) = n$ by the rank–nullity theorem. Therefore, $\operatorname{null}([\begin{smallmatrix} \mathbf{V}^\top & \mathbf{W}^\top \end{smallmatrix} ]) = {\operatorname{col}}\big(\big[\begin{smallmatrix} \mathbf{W} \\ -\mathbf{V} \end{smallmatrix}\big]\big)$ since $[\begin{smallmatrix} \mathbf{V}^\top & \mathbf{W}^\top \end{smallmatrix} ]\big[\begin{smallmatrix} \mathbf{W} \\ -\mathbf{V} \end{smallmatrix}\big]\mathbf{q} = \mathbf{V}^\top \mathbf{W}\mathbf{q} - \mathbf{W}^\top \mathbf{V}\mathbf{q} = \mathbf{0}$ for all $\mathbf{q}\in\mathbb{F}_p^n$ and $\operatorname{dim}(\operatorname{null}([\begin{smallmatrix} \mathbf{V}^\top & \mathbf{W}^\top \end{smallmatrix}])) = {\operatorname{rank}}\big(\big[\begin{smallmatrix} \mathbf{W} \\ -\mathbf{V} \end{smallmatrix}\big]\big) = n$. Moreover, $\operatorname{null}(\mathbf{V}^\top)\times\operatorname{null}(\mathbf{W}^\top) \subseteq \operatorname{null}([\begin{smallmatrix} \mathbf{V}^\top & \mathbf{W}^\top \end{smallmatrix}]) = {\operatorname{col}}\big(\big[\begin{smallmatrix} \mathbf{W} \\ -\mathbf{V} \end{smallmatrix}\big]\big)$, which means that $\operatorname{null}(\mathbf{V}^\top) \subseteq \operatorname{col}(\mathbf{W})$ and $\operatorname{null}(\mathbf{W}^\top) \subseteq \operatorname{col}(\mathbf{V})$. Finally, it is clear that $\operatorname{null}(\mathbf{V}) + \operatorname{null}(\mathbf{W})\subseteq\operatorname{null}(\mathbf{V}^\top \mathbf{W})$. On the other direction, given $\mathbf{q}\in\operatorname{null}(\mathbf{V}^\top \mathbf{W})$, then $\mathbf{W}\mathbf{q}\in\operatorname{null}(\mathbf{V}^\top)$ and  $\mathbf{V}\mathbf{q}\in\operatorname{null}(\mathbf{W}^\top)$. By $\operatorname{null}(\mathbf{V}^\top)\times\operatorname{null}(\mathbf{W}^\top) \subseteq  {\operatorname{col}}\big(\big[\begin{smallmatrix} \mathbf{W} \\ -\mathbf{V} \end{smallmatrix}\big]\big)$, this means that $\mathbf{W}\mathbf{q} = \mathbf{W}\mathbf{u}$ and $\mathbf{V}\mathbf{q} = -\mathbf{V}\mathbf{u}$ for some $\mathbf{u}\in\mathbb{F}_p^n$. Thus $\mathbf{q} \in \mathbf{u} + \operatorname{null}(\mathbf{W})$ and $\mathbf{q}\in -\mathbf{u}+\operatorname{null}(\mathbf{V})$, from which $\mathbf{q}\in\operatorname{null}(\mathbf{V}) + \operatorname{null}(\mathbf{W})$. Thus $\operatorname{null}(\mathbf{V}^\top \mathbf{W}) = \operatorname{null}(\mathbf{V}) + \operatorname{null}(\mathbf{W})$.
\end{proof}

\begin{replemma}{lem:uniqueness}
    Let $\mathcal{G} = \{\omega^{[\mathbf{a},\mathbf{x}]}\mathcal{W}_{\mathbf{x}}:\mathbf{x}\in\mathscr{X}\}$ be a group where $\mathscr{X}\subset\mathbb{F}_p^{2n}$ is an isotropic subspace. Then $\operatorname{dim}(V_{\mathcal{G}}) = p^{n-\operatorname{dim}(\mathscr{X})}$, where $V_{\mathcal{G}} \triangleq \{|\Psi\rangle\in(\mathbb{C}^{p})^{\otimes n}: \omega^{[\mathbf{a},\mathbf{x}]}\mathcal{W}_{\mathbf{x}}|\Psi\rangle = |\Psi\rangle, \forall \mathbf{x}\in\mathscr{X}\}$. It then follows that any stabiliser group $\mathcal{S}$ has a unique stabiliser state, i.e., $\operatorname{dim}(V_{\mathcal{S}}) = 1$.
\end{replemma}
\begin{proof}
    According to \cref{lem:projection}, the operator $\mathcal{P}_{\mathbf{a}} = |\mathscr{X}|^{-1}\sum_{\mathbf{x}\in\mathscr{X}}\omega^{[\mathbf{a},\mathbf{x}]}\mathcal{W}_{\mathbf{x}}$ is the projector onto $V_{\mathcal{G}}$. Now, due to the Pontryagin duality theorem (see~\cite[Theorem~1.7.2]{rudin2017fourier}) and $\mathscr{X}$ being a finite group, there are $p^{\operatorname{dim}(\mathscr{X})}$ characters of $\mathscr{X}$ (since the set of characters of $\mathscr{X}$ is isomorphic to $\mathscr{X}$). Each character gives rise to a distinct projector $\mathcal{P}_{\mathbf{b}}$ with equal rank. In other words, each left coset $\mathcal{C} \in \mathbb{F}_p^{2n}/\mathscr{X}^{\independent}$ leads to a distinct projector $\mathcal{P}_{\mathbf{b}}$ for any $\mathbf{b}\in\mathcal{C}$. Two different projectors are orthogonal since they belong to different eigenvalues of at least one Weyl operator or, alternatively,
    \begin{align*}
        \mathcal{P}_{\mathbf{a}}\mathcal{P}_{\mathbf{b}} = \frac{1}{|\mathscr{X}|^2}\sum_{\mathbf{x},\mathbf{y}\in\mathscr{X}}\omega^{[\mathbf{a},\mathbf{x}]+[\mathbf{b},\mathbf{y}]}\mathcal{W}_{\mathbf{x}}\mathcal{W}_{\mathbf{y}} = \frac{1}{|\mathscr{X}|^2}\sum_{\mathbf{x},\mathbf{z}\in\mathscr{X}}\omega^{[\mathbf{a},\mathbf{x}]+[\mathbf{b},\mathbf{z}-\mathbf{x}]}\mathcal{W}_{\mathbf{z}} = 0,
    \end{align*}
    using that $\sum_{\mathbf{x}\in\mathscr{X}}\omega^{[\mathbf{a}-\mathbf{b},\mathbf{x}]} = 0$ since $\mathbf{a} - \mathbf{b} \neq \mathscr{X}^{\independent}$ (\cref{lem:sum}). This means that the rank of each projector must be a $p^{\operatorname{dim}(\mathscr{X})}$-fraction of the Hilbert space dimension $p^n$, i.e., $p^{n-\operatorname{dim}(\mathscr{X})}$. This implies that $\operatorname{dim}(V_{\mathcal{G}}) = \operatorname{rank}(\mathcal{P}_{\mathbf{a}}) = p^{n-\operatorname{dim}(\mathscr{X})}$.
\end{proof}

\begin{replemma}{lem:stabiliser_dimension}
    Let $\mathcal{G} = \{\omega^{[\mathbf{a},\mathbf{x}]}\mathcal{W}_{\mathbf{x}}:\mathbf{x}\in\mathscr{X}\}$ be a group where $\mathscr{X}\subset\mathbb{F}_p^{2n}$ is an isotropic subspace. Let $V_{\mathcal{G}} \triangleq \{|\Psi\rangle\in(\mathbb{C}^{p})^{\otimes n}: \omega^{[\mathbf{a},\mathbf{x}]}\mathcal{W}_{\mathbf{x}}|\Psi\rangle = |\Psi\rangle, \forall \mathbf{x}\in\mathscr{X}\}$. Then, for any $|\Psi\rangle\in V_{\mathcal{G}}$, $\mathscr{X} \subseteq \operatorname{Weyl}(|\Psi\rangle) \subseteq \operatorname{Weyl}(|\Psi\rangle)^{\independent} \subseteq \mathscr{X}^{\independent}$. 
    Moreover, $\bigcap_{|\Psi\rangle\in V_{\mathcal{G}}}\operatorname{Weyl}(|\Psi\rangle) = \bigcap_{|\Psi\rangle\in V_{\mathcal{G}}}\operatorname{Weyl}(|\Psi\rangle)^{\independent} = \mathscr{X}$. 
\end{replemma}
\begin{proof}
    For any $|\Psi\rangle\in V_{\mathcal{G}}$, the inclusion $\mathscr{X} \subseteq \operatorname{Weyl}(|\Psi\rangle)$ follows from the definition of $\operatorname{Weyl}(|\Psi\rangle)$, while the inclusion $\operatorname{Weyl}(|\Psi\rangle) \subseteq \operatorname{Weyl}(|\Psi\rangle)^{\independent}$ follows from $\operatorname{Weyl}(|\Psi\rangle)$ being isotropic. Now, it is straightforward to see that $\mathscr{X} \subseteq \bigcap_{|\Psi\rangle\in V_{\mathcal{G}}}\operatorname{Weyl}(|\Psi\rangle) \subseteq \bigcap_{|\Psi\rangle\in V_{\mathcal{G}}}\operatorname{Weyl}(|\Psi\rangle)^{\independent}$. In the other direction, consider two extensions of $\mathscr{X}$ to Lagrangian subspaces $\mathscr{M}_1$ and $\mathscr{M}_2$ such that $\mathscr{M}_1 \cap \mathscr{M}_2 = \mathscr{X}$. Let $|\mathcal{S}_1\rangle = |\mathscr{M}_1,\mathbf{a}_1\rangle$ and $|\mathcal{S}_2\rangle = |\mathscr{M}_2,\mathbf{a}_2\rangle$ be stabiliser states of stabiliser groups $\mathcal{S}_1$ and $\mathcal{S}_2$ determined by $\mathscr{M}_1$ and $\mathscr{M}_2$, respectively ($\mathcal{G}$ is a subgroup of both $\mathcal{S}_1$ and $\mathcal{S}_2$). We must then have that $|\mathcal{S}_1\rangle,|\mathcal{S}_2\rangle\in V_{\mathcal{G}}$, and also that $\operatorname{Weyl}(|\mathcal{S}_1\rangle) = \operatorname{Weyl}(|\mathcal{S}_1\rangle)^{\independent} = \mathscr{M}_1$ and $\operatorname{Weyl}(|\mathcal{S}_2\rangle) = \operatorname{Weyl}(|\mathcal{S}_2\rangle)^{\independent} = \mathscr{M}_2$. Therefore, $\bigcap_{|\Psi\rangle\in V_{\mathcal{G}}}\operatorname{Weyl}(|\Psi\rangle) \subseteq \bigcap_{|\Psi\rangle\in V_{\mathcal{G}}}\operatorname{Weyl}(|\Psi\rangle)^{\independent} \subseteq \operatorname{Weyl}(|\mathcal{S}_1\rangle)^{\independent} \cap \operatorname{Weyl}(|\mathcal{S}_2\rangle)^{\independent} = \mathscr{M}_1 \cap \mathscr{M}_2 = \mathscr{X}$.
\end{proof}

\begin{replemma}{lem:t-doped}
    Let $|\psi\rangle\in(\mathbb{C}^p)^{\otimes n}$ be the output state of a $t$-doped Clifford circuit. Then the stabiliser dimension of $|\psi\rangle$ is at least $n-2t$.
\end{replemma}
\begin{proof}
    The proof is by induction on $t$. For $t=0$, the output state is a stabiliser state, which has stabiliser dimension $n$. Assume then the induction hypothesis for $t-1$. Write $|\psi\rangle = \mathcal{C}\mathcal{U}|\phi\rangle$, where $|\phi\rangle$ is the output of a $(t-1)$-doped Clifford circuit, $\mathcal{U}$ is a single-qudit gate, and $\mathcal{C}$ is a Clifford circuit. Since the stabiliser dimension is unchanged by Clifford circuits (\cref{cor:invariant}), we just need to show that the stabiliser dimension of $\mathcal{U}|\phi\rangle$ is at least $n-2t$. For such, consider $\operatorname{Weyl}(|\phi\rangle)$ with $\operatorname{dim}(\operatorname{Weyl}(|\phi\rangle)) \geq n - 2(t-1)$ by the induction assumption. Note that, for any $\mathbf{x}\in\operatorname{Weyl}(|\phi\rangle)$, if $\mathcal{W}_{\mathbf{x}}$ commutes with $\mathcal{U}$, then
    \begin{align*}
        \langle \phi|\mathcal{U}^\dagger \mathcal{W}_{\mathbf{x}}\mathcal{U}|\phi\rangle = \langle \phi| \mathcal{W}_{\mathbf{x}}|\phi\rangle \in \{1,\omega,\dots,\omega^{p-1}\}.
    \end{align*}
    Consider the set of commuting elements with $\mathcal{U}$, $\mathscr{E} \triangleq \{\mathbf{x}\in\operatorname{Weyl}(|\phi\rangle) : \mathcal{U}^\dagger \mathcal{W}_{\mathbf{x}}\mathcal{U} = \mathcal{W}_{\mathbf{x}}\}$. Then the stabiliser dimension of $\mathcal{U}|\phi\rangle$ is at least the dimension of $\mathscr{E}$, but $|\mathscr{E}| \geq |\operatorname{Weyl}(|\phi\rangle)|/p^2$, because $\mathscr{E}$ contains all elements $\mathbf{x}\in\operatorname{Weyl}(|\phi\rangle)$ for which $\mathcal{W}_{\mathbf{x}}$ restricts to the identity on the qudit to which $\mathcal{U}$ applies, i.e., $x_i = (v_i,w_i) = (0,0) \in \mathbb{F}_p^{2}$ if $\mathcal{U}$ acts on the $i$-th qudit. Therefore, the stabiliser dimension of $\mathcal{U}|\phi\rangle$ is at least $n-2t$ as required.
\end{proof}

\begin{replemma}{lem:Lipschitz}
    For any Weyl operator $\mathcal{W}_{\mathbf{x}}\in\mathscr{P}^n_p$, the function $f_{\mathbf{x}}:\mathbb{S}^{p^n}\to\mathbb{R}$ defined as $f_{\mathbf{x}}(|\psi\rangle) = \langle \psi|\mathcal{W}_{\mathbf{x}}|\psi\rangle$ is $2$-Lipschitz.
\end{replemma}
\begin{proof}
    $\begin{aligned}[t]
        |\langle\psi|\mathcal{W}_{\mathbf{x}}|\psi\rangle - \langle\phi|\mathcal{W}_{\mathbf{x}}|\phi\rangle| &= |\langle\psi|\mathcal{W}_{\mathbf{x}}|\psi\rangle - \langle\phi|\mathcal{W}_{\mathbf{x}}|\psi\rangle + \langle\phi|\mathcal{W}_{\mathbf{x}}|\psi\rangle -\langle\phi|\mathcal{W}_{\mathbf{x}}|\phi\rangle| \\ 
        &\leq |(\langle\psi| - \langle\phi|)\mathcal{W}_{\mathbf{x}}|\psi\rangle| + |\langle\phi|\mathcal{W}_{\mathbf{x}}(|\psi\rangle -|\phi\rangle)| \\
        &\leq \|\mathcal{W}_{\mathbf{x}} |\psi\rangle\|\||\psi\rangle - |\phi\rangle\| + \|\mathcal{W}_{\mathbf{x}} |\phi\rangle\|\||\psi\rangle - |\phi\rangle\| \\
        &= 2\||\psi\rangle - |\phi\rangle\|. \hspace{7.6cm}\qedhere
    \end{aligned}$
\end{proof}

\end{document}